%% file: arxiv.tex
\documentclass[sigplan, 9pt]{acmart}

\copyrightyear{2018}
\acmYear{2018}
\setcopyright{acmlicensed}
\acmConference[LICS '18]{LICS '18: 33rd Annual ACM/IEEE Symposium on Logic in Computer Science}{July 9--12, 2018}{Oxford, United Kingdom}
\acmBooktitle{LICS '18: LICS '18: 33rd Annual ACM/IEEE Symposium on Logic in Computer Science, July 9--12, 2018, Oxford, United Kingdom}
\acmPrice{15.00}
\acmDOI{10.1145/3209108.3209181}
\acmISBN{978-1-4503-5583-4/18/07}
\startPage{1}

\setcopyright{acmlicensed}

\usepackage{booktabs}   
\usepackage{subcaption} 

\usepackage{amsmath, amsthm, amssymb}
\usepackage{comment}
\usepackage{xspace}

\usepackage{enumitem}

\usepackage{pgf, pgfcore, tikz}
\usetikzlibrary{arrows} 
\usetikzlibrary{automata,positioning}
\usetikzlibrary{decorations.pathmorphing}
\usetikzlibrary{decorations.pathreplacing}
\usetikzlibrary{shapes.geometric,calc,shapes.misc}
\tikzset{cross/.style={cross out, draw=black, fill=none, minimum size=2*(#1-\pgflinewidth), inner sep=0pt, outer sep=0pt}, cross/.default={2pt}}
\usepackage{todonotes}
\usepackage{thm-restate}
\usepackage{fdsymbol}

\newcommand{\set}[1]{\left\{ #1 \right\}}
\newcommand{\tuple}[1]{\left( #1 \right)}

\newcommand{\ie}{\emph{i.e.}\xspace}
\newcommand{\eg}{\emph{e.g.}\xspace}
\newcommand{\Eg}{\emph{E.g.}\xspace}
\newcommand{\aka}{also known as\xspace}
\newcommand{\wog}{without loss of generality\xspace}
\newcommand{\wo}{without\xspace}

\newcommand{\ldata}{\ensuremath{\mathcal L_D}\xspace}
\newcommand{\eldata}{\ensuremath{\exists\mathcal L_D}\xspace}
\newcommand{\lorigin}{\ensuremath{\mathcal L_T}\xspace}
\newcommand{\elorigin}{\ensuremath{\exists\mathcal L_T}\xspace}

\newcommand{\binmso}{\ensuremath{\text{MSO}_{\text{bin}}}\xspace}
\newcommand{\mso}{\ensuremath{\text{MSO}}\xspace}
\newcommand{\fo}{\ensuremath{\text{FO}}\xspace}

\newcommand{\msot}{\ensuremath{\text{MSOT}}\xspace}
\newcommand{\nmso}{\ensuremath{\text{NMSO}}\xspace}
\newcommand{\nmsot}{\ensuremath{\text{NMSOT}}\xspace}

\newcommand{\lpp}{MCP\xspace}
\newcommand{\dom}{\mathrm{dom}}
\newcommand{\nat}{\mathbb N}

\newcommand{\expspace}{{\sc ExpSpace}\xspace}

\newcommand{\expspacec}{{\sc ExpSpace}-c\xspace}
\newcommand{\nlogspace}{{\sc NLogSpace}\xspace}

\newcommand{\td}{\text{t2d}}

\newcommand{\out}{\texttt{out}}
\newcommand{\inp}{\texttt{in}}
\newcommand{\ori}{\texttt{o}}
\newcommand{\vo}{\mathrm{vo}}

\newcommand{\exinput}{\exists^{\inp}}
\newcommand{\forinput}{\forall^{\inp}}
\newcommand{\exoutput}{\exists^{\out}}
\newcommand{\foroutput}{\forall^{\out}}
\newcommand{\leqinput}{\mathrel{\leq_{\inp}}}
\newcommand{\leqoutput}{\mathrel{\leq_{\out}}}

\newcommand\sem[1]{\llbracket #1 \rrbracket}
\newcommand\semo[1]{\llbracket #1 \rrbracket_o}
\newcommand{\Prods}[2]{\mathcal{OG}(#1,#2)}
\newcommand{\ProdSG}{\Prods{\Sigma}{\Gamma}}
\newcommand{\dataw}[2]{\mathcal{TDW}(#1,#2)}
\newcommand{\datawSG}{\dataw{\Sigma}{\Gamma}}

\newcommand{\cS}{S_\Psi}
\newcommand{\Select}{{SP}}
\newcommand{\profile}{\lambda}
\newcommand{\profseq}{\text{Seq}}
\newcommand{\fprofseq}{\text{fullSeq}}

\newcommand{\clauses}{\mathcal{C}}

\newtheorem{theorem}{Theorem}
\theoremstyle{plain}
\newtheorem{proposition}[theorem]{Proposition}
\newtheorem{corollary}[theorem]{Corollary}
\theoremstyle{definition}
\newtheorem{definition}{Definition}
\newtheorem{example}[theorem]{Example}
\theoremstyle{remark}
\newtheorem{remark}[theorem]{Remark}

\renewcommand{\paragraph}[1]{\par\smallskip\noindent\textbf{#1.}}

\begin{document}

\title{Logics for Word Transductions with Synthesis}         


\author{Luc Dartois}
\affiliation{
	\institution{Universit\'e Libre de Bruxelles}
 \country{Belgium}                    
}
\email{luc.dartois@ulb.ac.be}          

\author{Emmanuel Filiot}
\affiliation{
	\institution{Universit\'e Libre de Bruxelles}
 \country{Belgium}                    
}
\email{efiliot@ulb.ac.be}          

\author{Nathan Lhote}
\affiliation{
  \institution{Universit\'e Libre de Bruxelles}           
  \country{Belgium}                   
}
\affiliation{
 \institution{LaBRI, Universit\'e de Bordeaux}           
  \country{France}                   
}
\email{nlhote@labri.fr}

\begin{abstract}
    We introduce a
    logic, called $\lorigin$, to express properties of
    transductions, \ie binary relations from input to output (finite)
    words. In $\lorigin$, the input/output dependencies are 
    modelled via an \emph{origin function}  
    which associates to any position of the output word,
    the input position from which it originates. \lorigin is
    well-suited to express relations (which are not necessarily
    functional), and can express all regular functional
    transductions, \ie transductions definable for instance by
    deterministic two-way transducers.

    Despite its high expressive power, $\lorigin$ has decidable satisfiability and equivalence
    problems, with tight non-elementary and elementary complexities,
    depending on specific representation of \lorigin-formulas. 
    Our main contribution is a synthesis result: from any 
    transduction $R$ defined in $\lorigin$, it is possible to synthesise a
    \emph{regular} functional transduction $f$ such that for all input
    words $u$ in the domain of $R$, $f$ is defined and $(u,f(u))\in
    R$. As a consequence, we obtain that any functional transduction
    is regular iff it is $\lorigin$-definable.

    We also investigate the
    algorithmic and expressiveness properties of several extensions of
    $\lorigin$, and explicit a correspondence between transductions
    and data words. As a side-result, we obtain a new decidable logic
    for data words. 
\end{abstract}

%


\maketitle

\begin{acks}
{This work was supported by the French ANR  project \emph{ExStream}
	(ANR-13-JS02-0010), the Belgian FNRS CDR project \emph{Flare} (J013116) and the ARC project \emph{Transform} (F\'ed\'eration Wallonie Bruxelles).
	
	We are also grateful to Jean-Fran\,cois Raskin from fruitful discussions on this work.}
\end{acks}
\section{Introduction}

The theory of regular languages of finite and infinite words is rich and robust, founded on
the equivalence of a descriptive model (monadic second-order logic, MSO)
and a computational one (finite automata), due the works of B\"uchi, Elgot, 
McNaughton and Trahtenbrot \cite{Tho97handbook}. Since then, many
logics have been designed and studied to describe languages (see
for instance \cite{str94,DiekertGastin08}), among which temporal
logics, with notable applications in model-checking~\cite{VW86}.

In this paper, we consider transductions, \ie binary
relations relating input to output words. \Eg the transduction $\tau_{\text{shuffle}}$ associates with a word all its permutations
 -- $(ab,ab), (ab,ba)\in\tau_{\text{shuffle}}$. Operational models, namely extensions of automata with outputs,
called \emph{transducers}, have been studied for computing
transductions. This includes finite transducers, \ie finite automata with outputs,
which have been studied since the 60s \cite{Bers79,sakarovich:2009a} and
two-way transducers (two-way automata with a one-way
output tape). When restricted to transducers defining functions (called
functional transducers), the latter model has recently received a lot of
attention due to its appealing algorithmic properties, its expressive
power and its many equivalent models: deterministic two-way
transducers~\cite{dS13}, reversible  two-way transducers~\cite{DFJL17},
deterministic (one-way) transducers with registers~\cite{AC10} (\aka
streaming string transducers), regular combinator expressions~\cite{AFR14}
and Courcelle's MSO-transducers~\cite{EH01} (MSOT), a model we will
come back to in the related work section. Because of these many characterisations, the class defined by these models has been coined \emph{regular
  transductions}, or regular functions.

However, much less is known about logics to describe
transductions (see for instance~\cite{DBLP:journals/siglog/FiliotR16}
for a brief overview). Recently, Boja\'{n}czyk, Daviaud, Guillon and
Penelle have considered an expressive logic, namely MSO over \emph{origin graphs} (o-graphs) \cite{B17}. Such graphs
encode pairs of words together with an \emph{origin mapping}, relating
any output position to an input position, as
depicted in Fig.~\ref{fig:ographs}. Intuitively, if one thinks of an
operational model for transductions, the origin of an output position
is the input position from which it has been produced.%
\begin{figure}[h]
\vspace{-2mm}
\begin{center}
\begin{tikzpicture}[>=stealth',auto,node distance=3cm,scale=0.8,every node/.style={scale=0.8}]
\node[] (u0) at (-3,0) {input};
\node[] (v0) at (-2,-0.5) {origin};
\node[] (w0) at (-3,-1) {output};

\node[] (i1) at (-1,0) {$a$};
\node[] (i2) at (0,0) {$b$};
\node[] (i3) at (1,0) {$c$};
\node[] (i4) at (2,0) {$a$};

%

\node[] (j1) at (-1,-1) {${a}$};
\node[] (j2) at (2,-1) {${b}$};
\node[] (j3) at (0,-1) {${c}$};
\node[] (j4) at (1,-1) {${a}$};

\draw [->] (j1) -- (i1) ;
\draw [->] (j2) -- (i2) ;
\draw [->] (j3) -- (i3) ;
\draw [->] (j4) -- (i4) ;

\draw [->] (i1) -- (i2) ;
\draw [->] (i2) -- (i3) ;
\draw [->] (i3) -- (i4) ;
\draw [->] (j1) -- (j3) ;
\draw [->] (j3) -- (j4) ;
\draw [->] (j4) -- (j2) ;

\node[] (i1') at (3.5,0) {$a$};
\node[] (i2') at (4.5,0) {$b$};
\node[] (i3') at (5.5,0) {$c$};
\node[] (i4') at (6.5,0) {$a$};


\node[] (j1') at (6.5,-1) {${a}$};
\node[] (j2') at (5.5,-1) {${b}$};
\node[] (j3') at (4.5,-1) {${c}$};
\node[] (j4') at (3.5,-1) {${a}$};

\draw [->] (j1') -- (i1') ;
\draw [->] (j2') -- (i2') ;
\draw [->] (j3') -- (i3') ;
\draw [->] (j4') -- (i4') ;

\draw [->] (i1') -- (i2') ;
\draw [->] (i2') -- (i3') ;
\draw [->] (i3') -- (i4') ;
\draw [->] (j4') -- (j3') ;
\draw [->] (j3') -- (j2') ;
\draw [->] (j2') -- (j1') ;

 \end{tikzpicture}
\end{center}
\vspace{-2mm}
\caption{Possible o-graphs for $\tau_{\text{shuffle}}$}\label{fig:ographs}
\vspace{-2mm}
\end{figure}
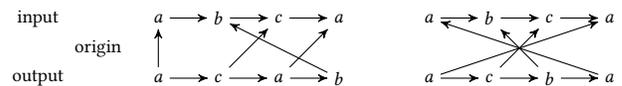
As noticed in~\cite{Bojanczyk14}, most known transducer models not only define
transductions, but \emph{origin transductions} (o-transductions), \ie sets
of o-graphs, and can thus be naturally interpreted in both
\emph{origin-free semantics} (\ie usual semantics) or the richer
\emph{origin semantics}. %
We denote by $\mso_\ori$ monadic second-order logic over
o-graphs. Precisely, it is MSO equipped with monadic predicates $\sigma(x)$ for 
position labels, a linear order
$\leq_\inp$ (resp. $\leq_\out$) over input (resp. output) positions,
and an origin function $\ori$. We denote by $\semo{\phi}$ the
origin-transduction defined by $\phi$, \ie the set of o-graphs
satisfying $\phi$, and by $\sem{\phi}$ the transduction defined by
$\phi$ (obtained by projecting away the origin mapping of
$\semo{\phi}$). While \cite{B17} was mostly concerned
with characterising classes of o-graphs generated by particular
classes of transducers, the authors have shown another interesting
result, namely the decidability of model-checking regular functions
with origin against $\mso_\ori$ properties: it is decidable, given an $\mso_\ori$ sentence $\phi$ and a deterministic two-way
transducer $T$, whether all $o$-graphs of $T$ satisfy $\phi$. 

\paragraph{Satisfiability and synthesis} Important and natural
verification-oriented questions are not considered in~\cite{B17}. The first is the satisfiability
problem for $\mso_\ori$: given a sentence
$\phi$, is it satisfied by some o-graph? While being one of the most
fundamental problem in logic, its decidability would also entail the
decidability of the equivalence problem, a fundamental problem in
transducer theory: given two sentences $\phi_1,\phi_2$ of $\mso_\ori$, does
$\semo{\phi_1} = \semo{\phi_2}$ hold? The second problem is the
regular synthesis problem: given an $\mso_\ori$-sentence $\phi$, does there exist a deterministic
two-way transducer $T$ such that $(1)$ $T$ has the same domain as $\sem{\phi}$
(the set of words which have some image by $\sem{\phi}$) and $(2)$ for
all $u$ in the domain of $T$, its image $T(u)$ satisfies $(u,T(u))\in
\sem{\phi}$. Note that
without requirement $(1)$, any transducer $T$ with empty domain
would satisfy $(2)$. So, instead of designing a transducer and then verifying a posteriori that
it satisfies some $\mso_\ori$ properties, the goal is to check whether
some transducer can be automatically generated from
these properties (and to synthesise it), making it correct \emph{by
  construction}. Unsurprisingly, we show that both these
problems are \emph{undecidable} for $\mso_\ori$. 

\paragraph{Contribution: The fragment $\lorigin$} We define a fragment of $\mso_\ori$ called
$\lorigin$ for which, amongst other interesting properties, the two problems mentioned before are decidable. 
Before stating our precise results on $\lorigin$, let us intuitively
define it and provide examples. $\lorigin$ is the
two-variable fragment\footnote{Only
  two variable names can be used (and reused) in a formula, see \eg ~\cite{str94}} 
  of first-order logic -- $\fo^2$. The predicates in its signature are the output labels,
the linear order $\leq_\out$ for the output positions, the origin function $\ori$, and 
\emph{any} binary MSO predicate restricted to input positions, using  
input label predicates and the input order
$\leq_\inp$. We write it $\lorigin
:=\fo^2[\Gamma,\leqoutput,\ori,\binmso[\leqinput,\Sigma]]$ where
$\Gamma$ is the output alphabet and $\Sigma$ the input
alphabet.

As an example, let us define the transduction $\tau_{\text{shuffle}}$ in $\lorigin$. 
We express that $(1)$ $\ori$ preserves the labelling: $\forall^\out x \bigwedge_{\sigma\in \Gamma}
\sigma(x)\rightarrow \{\sigma(\ori(x))\}$,
and $(2)$ $\ori$ is  bijective, \ie injective: $\forall^\out x,y\ \{\ori(x)= \ori(y)\}\rightarrow x=y
$ and surjective: $\forinput x \exoutput y\ \{x=\ori(y)\}$. 
The notation $\forall^\out$ is a macro which restricts quantification over output
positions, and we use brackets $\{,\}$ to distinguish the
binary MSO predicates. 
Extending this, suppose we have some alphabetic linear
order $\preceq$ over $\Sigma$ and we want to sort the input
labels by increasing order. This can be done by adding the requirement 
$\forall^\out x,y\bigwedge_{\sigma\prec \sigma'} \sigma(x)\wedge
\sigma'(y)\rightarrow x\leq_\out y$. This simply defined transduction
can be realised by a two-way transducer, which would make one pass per symbol $\sigma$ (in increasing order),
during which it copies the $\sigma$-symbols on the output tape and not
the others. 

\paragraph{Results} We show the following results on $\lorigin$:
\begin{itemize}[leftmargin=*]
  \item it is \emph{expressive}: any regular functional transduction is definable
    in $\lorigin$. Beyond functions, $\lorigin$ is incomparable with
    non-deterministic two-way transducers and non-deterministic
    streaming string transducers (it can express
    $\tau_{\text{shuffle}}$ which is definable in none of these
    models). 
  \item it \emph{characterises} the regular functional transductions:
    a functional transduction is regular iff it is
    $\lorigin$-definable. Moreover, given an $\lorigin$-formula, it is
    decidable whether it defines a functional transduction. 
  \item the satisfiability problem is decidable (in non-elementary
    time, which is unavoidable because of the binary MSO predicates),
    and \expspacec if the binary MSO
    predicates are given by automata. Since $\lorigin$ is closed under
    negation, we obtain as a consequence  the decidability of the
    equivalence problem for $\lorigin$-definable o-transductions. 
  \item it admits regular synthesis: from any $\lorigin$-sentence $\phi$, one can
    \emph{always} synthesise a deterministic two-way transducer which
    has the same domain as $\sem{\phi}$ and whose o-graphs all satisfy
    $\phi$. 
\end{itemize}

Finally, we provide two strictly more expressive extensions of
$\lorigin$, shown to admit regular synthesis, and hence decidable
satisfiability problem. The first one $\elorigin$ extends any
$\lorigin$-formula with a block of existential monadic second-order
quantifiers  and it captures all transductions defined by
non-deterministic \mso-transducers or equivalently non-deterministic
streaming string transducers~\cite{AD11}. Then, we introduce 
$\elorigin^{\mathrm {so}}$ which extends $\elorigin$ with unary
predicates $L(x)$ called \emph{single-origin predicates}, where $L$ is
a regular language, which holds in an input position $x$ if the word
formed by the positions having origin $x$ belongs to $L$. 
For instance one could express that any input position labelled by $a$
has to produce a word in $(bc)^*$, which cannot be done with a $\fo^2$ formula. 
This extension allows us to additionally  capture any
rational relation, \ie the transductions defined by (nondeterministic) one-way
transducers~\cite{Bers79}.


Our main and most technical result is regular synthesis. Indeed, it
entails satisfiability (test domain emptiness of the constructed
transducer), and, since no automata/transducer model is known to be
equivalent to $\mso_\ori$ nor $\lorigin$, we could not directly rely on
automata-based techniques. The techniques of \cite{B17} for
model-checking do not apply
either because the target model is not given when considering
satisfiability and synthesis. Instead, we introduce a sound and complete
bounded abstraction of the o-graphs satisfying a given
\lorigin-formula. This abstraction was inspired by techniques used in data
word logics~\cite{SZ12}, although we could not directly reuse known results, since
they were only concerned with the satisfiability
problem. Nonetheless, we exhibit a tight connection between o-graphs and data words.

\paragraph{A consequence on data words} As a side contribution, we explicit a bijection between non-erasing  origin graphs (the origin mapping
is surjective) and words over an infinite alphabet of totally ordered
symbols, called data words. Briefly, the origin becomes the data and conversely the data
becomes the origin. We show that this bijection carries over to the
logical level, and we obtain a \emph{new} decidable logic for data words, which
strictly extends the logic $\fo^2[\leq,\preceq,S_\preceq]$ (linear position order and
linear order and successor over data), known to be decidable from \cite{SZ12}, with any binary MSO$[\preceq]$ predicate talking only about
the data.

\paragraph{Related Work} First, let us mention some known logical way
of defining transductions. Synchronised (binary) relations, \aka automatic
relations, are relations defined by automata running over word
convolutions~\cite{sakarovich:2009a}. A convolution $u\otimes v$ is obtained by
overlapping two words $u,v$ and by using a padding symbol $\bot$ if they do not
have the same length. E.g. $aba\otimes bc = (a,b)(b,c)(a,\bot)$. By
taking MSO over word convolutions, one obtains a logic to define
transductions. It is however quite weak in expressive power, as it
cannot even express all functional transductions definable by one-way input-deterministic
finite transducers.

Courcelle has introduced MSO-transducers to define graph
transductions~\cite{DBLP:journals/iandc/Courcelle90} and which, casted to words, gives a logic-based
formalism to define word transductions. Roughly, the predicates of the
output word are defined by several MSO-formulas with free variables,
interpreted over a bounded number of copies of the input
structure. Additionally, several free parameters can be used to add a
form of non-determinism. Functional MSO-transducers correspond exactly
to functional regular transduction~\cite{EH01}. However, they have a
relatively limited expressive power when it comes to relations,
because, unlike $\lorigin$, the number of images of a word is always
finite. For instance, the universal transduction $\Sigma^*\times
\Sigma^*$ is not definable in this formalism, while it is simply
definable by the $\lorigin$-formula $\top$, nor is
$\tau_{\text{shuffle}}$ (this can be shown using cardinality arguments).

Finally, there is a number of recent works on reactive synthesis~\cite{syntcomp}, since the
seminal problem raised by Church~\cite{Chur62}, and studied by Pnueli and
Rosner for LTL specifications~\cite{PnuRos:89}. In these works however, the
specification is always a synchronised relation and the target
implementation is a Mealy machine (an input-deterministic finite transducer
alternatively reading and producing exactly one symbol at a time). While
$\lorigin$ does not make any synchronicity assumption, the target
implementations in this paper are  deterministic two-way transducer which are, computationally speaking, more powerful. We
leave open the question of whether the following synthesis problem is
decidable: given an $\lorigin$-formula $\phi$, is there a (one-way) input-deterministic (\aka
sequential) transducer realising $\phi$?

Transducer synthesis is also equivalently known as \emph{uniformisation} in transducer theory~\cite{sakarovich:2009a}. This problem has been
studied in the origin-free semantics for the class of rational
relations. It is known that from any
rational relation one can synthesise a rational function~\cite{Elgot:Mezei:ibmjrd:1965}, and
that checking whether it is realisable by a sequential function is
undecidable~\cite{CarayolL14,conf/icalp/FiliotJLW16}. The former
result is a consequence of our results on the extension
$\lorigin^{\mathrm {so}}$: we show that any rational relation defined
by a one-way transducer is $\lorigin^{\mathrm {so}}$-definable (while preserving the origin
mappings) and moreover, any transduction defined in
$\lorigin^{\mathrm {so}}$ is realisable by a regular function
. Hence, from rational relation given as a one-way transducer
$T$ we obtain an order-preserving and functional regular o-transduction
that realises the relation defined by $T$. Such o-transductions are
easily seen to be equivalent to rational functions~\cite{Bojanczyk14,Filiot-ICLA15}. 
Finally, we mention that
transducer synthesis has also been recently studied in the context of
trees, where the specification is a tree automatic relation~\cite{DBLP:journals/corr/LodingW14}.


Due to the lack of space, some proofs are omitted or only sketched in the body of the paper. The full proofs are given in the appendix.

\section{Logics with origin for transductions}\label{sec:lorigin}

\paragraph{Words and transductions} We denote by $\Sigma^*$ the set of
finite words over some alphabet $\Sigma$, and by $\epsilon$ the empty
word. The length of a word $u\in\Sigma^*$ is denoted by $|u|$, in
particular $|\epsilon|=0$. The set of positions of $u$ is $\dom(u) =
\{1,\dots,|u|\}$, an element $i\in\dom(u)$ denoting the $i$th position
of $u$, whose symbol is denoted $u(i)\in\Sigma$.

Let $\Sigma$ and $\Gamma$ be two alphabets, \wog assumed to be disjoint.
A \emph{transduction} is a subset of $\Sigma^+\times
\Gamma^*$ of pairs $(u,v)$, where $u$ is called the input word and $v$
the output word. An \emph{origin mapping} from a word $v\in\Gamma^*$ to a
word $u\in\Sigma^+$ is a mapping $o : \dom(v)\rightarrow
\dom(u)$. Intuitively, it means that position $i$ was produced when
processing position $o(i)$ in the input word $u$. We exclude the empty
input word from the definition of transductions, because we require 
every output position to have some origin. This does not weaken
the modelling power of the logics we consider, up to putting
some starting marker for instance. Following the
terminology of \cite{B17}, an \emph{origin-graph}
(o-graph for short) is a pair $(u,(v,o))$ such that
$(u,v)\in\Sigma^+\times \Gamma^*$ and $o$ is an origin mapping from $v$
to $u$. We denote by $\ProdSG$
the set of o-graphs from $\Sigma$ to $\Gamma$. A \emph{transduction
  with origin} (or just o-transduction) $\tau$ from $\Sigma$ to
$\Gamma$ is a set of o-graphs. We say that $\tau$ is
functional (or is a function) if for all $u$, there is at most one pair $(v,o)$ such that
$(u,(v,o))\in\tau$, and rather denote it by $f$ instead of $\tau$. 
The \emph{domain} of an o-transduction $\tau$ is the set 
$\dom(\tau) = \{ u\mid \exists (u,(v,o))\in\tau\}$. 
Finally, the \emph{origin-free projection} of $\tau$ is the
transduction $\{ (u,v)\mid \exists (u,(v,o))\in\tau\}$. 
Many results of this paper hold with or \wo origins. 
We always state them in their strongest version, usually
\wo origin.

\paragraph{Regular functional transductions} Regular functional transductions (or
regular functions) have many characterisations, as mentioned in the
introduction. We will briefly define them as the transductions
definable by deterministic two-way transducers, which are pairs
$(A,\rho)$ such that $A$ is a deterministic two-way automaton with set
of transitions $\Delta$, and $\rho$ is a morphism of type
$\Delta^*\rightarrow \Gamma^*$. The transduction defined by $(A,\rho)$
has domain $L(A)$ (the language recognised by $A$) and for all words
$u$ in its domain, the output of $u$ is the word $\rho(r)$, where $r$
is the accepting sequence of transitions of $A$ on $u$. Such transducers (as well as other
known equivalent models) can be naturally equipped with an origin semantics~\cite{Bojanczyk14} and
we say that a functional $o$-transduction is regular if it is equal to the set of
$o$-graphs of some deterministic two-way transducer.



\paragraph{FO and MSO logics for transductions}\label{subsec:fomsotrans} We consider FO
and MSO over particular signatures. Without defining their syntax formally (we refer the reader \eg to \cite{str94}),
recall that MSO over a set of predicates $S$ allows for first-order
quantification $\exists x$ over elements, second-order quantification
$\exists X$ over element sets, membership predicates $x\in
X$, predicates of $S$ and all Boolean operators. We use the
notation $\mso[S]$ (or $\fo[S]$) to emphasise that formulas are built
over a particular signature $S$. As usual, $\phi(x_1,\dots,x_n)$
denotes a formula with $n$ free first-order variables, and we call
\emph{sentence} a formula \wo free variables. Finally, $\models$ denotes
the satisfiability relation.

Origin-graphs $(u,(v,o))$ of $\ProdSG$ are seen as
\emph{structures} with domain
$\dom(u)\uplus \dom(v)$ over the 
signature $\mathcal S_{\Sigma,\Gamma}$ composed of unary predicates
$\delta(x)$, for all $\delta\in \Sigma\cup \Gamma$, holding true on
all positions labelled $\delta$, $\leqinput$ a linear-order on the
positions of $u$, $\leqoutput$ a linear-order on the positions of $v$,
and $\ori$ a unary function for the origin, which is naturally
interpreted by $o$ over $\dom(v)$, and as the identity
function\footnote{As functional symbols must be interpreted by
  total functions, we need to interpret $\ori$ over $\dom(u)$ as well.} over
$\dom(u)$. We also use the predicates $=$, $<_\inp$ and $<_\out$,
which are all definable in the logics we consider. We denote by
$\mso_\ori$ the logic $\mso[\mathcal S_{\Sigma,\Gamma}]$.
Any $\mso_\ori$ sentence $\phi$ defines an
o-transduction $\semo{\phi} = \{ (u,(v,o))\in\ProdSG\mid (u,(v,o))\models
\phi\}$ and its origin-free counterpart $\sem{\phi}$. An o-transduction (resp. transduction) $\tau$
is $\mso_\ori$-definable if $\tau = \semo{\phi}$ (resp. $\tau =
\sem{\phi}$) for some sentence $\phi\in\mso_\ori$.

\begin{example}\label{ex:mso}
First, we define the 
transduction $\tau_{\text{cfl}}$ mapping $a^nb^n$ to $(ab)^n$, as the
origin-free projection of the set of o-graphs defined by some
$\mso_\ori$-sentence $\phi_{\text{cfl}}$, which expresses that (1) the domain is in
$a^*b^*$, (2) the codomain in
$(ab)^*$ (both properties are regular and, hence, respectively
$\mso[\Sigma,\leqinput]$- and $\mso[\Gamma,\leqoutput]$-definable),
and (3) the origin-mapping is bijective and label-preserving (see
introduction).
\end{example}

\paragraph{Satisfiability and synthesis problems} The satisfiability
problem asks, given an $\mso_\ori$-sentence $\phi$, whether it is
satisfied by some $o$-graph, \ie whether
$\semo{\phi}\neq\varnothing$ (or equivalently
$\sem{\phi}\neq\varnothing$) holds. %
By encoding the Post Correspondence Problem, we show that $\mso_\ori$ has
undecidable satisfiability problem, even if restricted to the
two-variable fragment of \fo with the $S_\out$ predicate, denoting the
successor relation over output positions:

\begin{restatable}{proposition}{propUndecfoD}\label{prop:undecfo2}
    Over o-graphs, the logic
    $\fo^2[\Sigma,\Gamma,\leqinput,\leqoutput,S_\out,\ori]$ has
    undecidable satisfiability problem. 
\end{restatable}

Given a transduction $\tau$ and a functional transduction $f$, we say
that $f$ \emph{realises} $\tau$ if $\dom(f) = \dom(\tau)$, and
for all input $u$, $(u,f(u))\in \tau $. The \emph{regular synthesis problem} asks whether
given an o-transduction $\tau$, there exists a regular functional o-transduction $f$
which realises it. As claimed in the introduction, this problem is
undecidable when $\tau$ is defined in $\mso_\ori$.

\begin{restatable}{proposition}{propUndecSynth}\label{prop:undecsynth}
    The regular synthesis problem is undecidable for
    $\mso_\ori$-definable transductions. 
\end{restatable}
\begin{proof}[Sketch]
    We reduce the $\mso_\ori$ satisfiability problem. First, consider
    the $\mso_\ori$-sentence $\phi_{\text{cfl}}$ of Ex.~\ref{ex:mso}
    defining a transduction with non-regular domain. 
    Then, given an $\mso_\ori$-formula
    $\psi$ of which one wants to check satisfiability, we define in $\mso_\ori$,
    using $\psi$ and $\phi_{\text{cfl}}$, 
    the transduction $\tau$ mapping
    any word $u_1\#u_2$ to $v_1\#v_2$ such that
    $(u_1,v_1)\in\sem{\psi}$ and
    $(u_2,v_2)\in\sem{\phi_{\text{cfl}}}$. 
    Then, $\dom(\tau)$ is non-regular
    iff it is nonempty. Since regular functions have regular domains, $\tau$ is realisable by a
    regular function iff $\dom(\tau) = \varnothing$
    iff $\sem{\psi}=\varnothing$ iff $\semo{\psi} = \varnothing$.
\end{proof}

\paragraph{The logic \lorigin for transductions}\label{subsec:lorigin}
Informally, the logic $\lorigin$ extends the two-variable logic
$\fo^2[\Sigma,\Gamma,\leqinput,\leqoutput,\ori]$ with any binary predicate definable in
$\mso[\leqinput,\Sigma]$, \ie any binary MSO predicate that is only
allowed to talk about the input positions, in order to capture regular
input properties. Formally, we denote by $\binmso[\leqinput,\Sigma]$
the set of $n$-ary predicates, $n\in\{0,1,2\}$, denoted by $\{ \phi\}$, where 
$\phi$ is an $\mso[\leqinput,\Sigma]$-formula
with at most $n$ free first-order variables. Over a word $u$, such a
formula $\phi$ defines an $n$-ary relation $R_{\phi,u}$ on its
position, and over an o-graph $(u,(v,o))$, we interpret $\{\phi\}$
by $R_{\phi,u}$. The logic $\lorigin$ is the two-variable fragment of first-order
logic over the output symbol predicates, the linear-order
$\leqoutput$, and all predicates in $\binmso[\leqinput,\Sigma]$,
\ie $\lorigin
:=\fo^2[\Gamma,\leqoutput,\ori,\binmso[\leqinput,\Sigma]]$. Modulo
removing the brackets $\{,\}$, it is a proper fragment of
$\mso_\ori$. 


\paragraph{Examples of \lorigin-transductions}\label{ex:lorigin}
The true formula $\top$ is satisfied by any o-graph.
Hence $\sem{\top}=\Sigma^+\times\Gamma^*$.
Let us now define several macros that will be useful throughout the
paper. The formula  $\inp(x)\equiv x\leqinput x$ (resp. 
    $\out(x)\equiv x\leqoutput x$) holds true if $x$
    belongs to the input word (resp. output word). Now
    for $\alpha\in\{\inp,\out\}$, we define the guarded quantifiers
    $\exists^\alpha x\ \phi$ and $\forall^\alpha x\ \phi$ as shortcuts for 
    $\exists x\ \alpha(x)\wedge \phi$ and $\forall x\
    \alpha(x)\rightarrow \phi$ (note that $\neg \exists^\alpha x\ \phi$
    is equivalent to $\forall^\alpha x\ \neg\phi$).

Preservation of
the input/output orders is expressed by the $\lorigin$-formula
$\forall^\out x,y\ 
(x\leqoutput y) \rightarrow  \{ x'\leqinput y'\}(\ori(x),\ori(y))
$. 
Note that we could equivalently
replace $x'$ and $y'$ by any variable (even $x$ and $y$), without changing the
semantics: the formula $x'\leqinput y'$ defines a binary relation
on the input word, which is used as an interpretation of the predicate $\{
x'\leqinput y'\}$ in o-graphs.  To ease the notations, any predicate
$\{\phi\}(t_1,t_2)$ where $\phi$ has two free variables $x_1$ and
$x_2$ may be sometimes written $\{ \phi[x_1/t_1,x_2/t_2]\}$, \ie $\phi$ in
which $t_i$ has been substituted for $x_i$. We keep the
brackets $\{$ and $\}$ to emphasise the fact that it is a binary MSO
formula which speaks about the input word. Hence, the previous
formula may also be written $\phi_{\text{pres}} \equiv \forall^\out x,y\ (x\leqoutput y) \rightarrow \{ \ori(x)\leqinput \ori(y)\}
    $. 

The fact that $o$ is a bijective mapping is expressible by some
$\lorigin$-formula $\phi_{\text{bij}}$, as seen in the introduction. %
Then, the shuffle transduction $\tau_{\text{shuffle}}$ 
is defined by
$
\phi_{\text{shuffle}} \equiv \phi_{\text{bij}} \wedge \forall^\out x 
\bigwedge_{\sigma\in \Gamma} \sigma(\ori(x))\rightarrow {\sigma}(x)
$.
    If the origin mapping is also required to be order-preserving, we
    get a formula defining identity: $\phi_{\text{id}}  \equiv
                                  \phi_{\text{shuffle}}\wedge \phi_{\text{pres}}$.

    Let us now consider the transduction $\tau:(ab)^n \mapsto
    a^nb^n$. By taking any bijective
    and label-preserving origin mapping, \eg as follows: 
\begin{center}
\vspace{-3mm}
\begin{tikzpicture}[>=stealth',auto,node distance=3cm,scale=0.8,every node/.style={scale=0.8}]
\node[] (i1) at (0,0) {$a$};
\node[] (i2) at (1,0) {$b$};
\node[] (i3) at (2,0) {$a$};
\node[] (i4) at (3,0) {$b$};
\node[] (i5) at (4,0) {$a$};
\node[] (i6) at (5,0) {$b$};
\node[] (i7) at (6,0) {$a$};
\node[] (i8) at (7,0) {$b$};

\node[] (j1) at (0,-1) {$a$};
\node[] (j2) at (1,-1) {$a$};
\node[] (j3) at (2,-1) {$a$};
\node[] (j4) at (3,-1) {$a$};
\node[] (j5) at (4,-1) {$b$};
\node[] (j6) at (5,-1) {$b$};
\node[] (j7) at (6,-1) {$b$};
\node[] (j8) at (7,-1) {$b$};

\draw [->] (j1) -- (i1) ;
\draw [->] (j2) -- (i3) ;
\draw [->] (j3) -- (i5) ;
\draw [->] (j4) -- (i7) ;
\draw [->] (j5) -- (i2) ;
\draw [->] (j6) -- (i4) ;
 \draw [->] (j7) -- (i6) ;
\draw [->] (j8) -- (i8) ;
 \end{tikzpicture}
\vspace{-3mm}
\end{center}

\noindent one can define $\tau$, as long as the input word is in
$(ab)^*$, which is regular, hence definable by some $\mso[\leqinput,\Sigma]$-formula
$\phi_{(ab)^*}$. Then, $\tau$ is defined by:
$\{\phi_{(ab)^*}\} \wedge \phi_{\text{bij}}\wedge
\bigwedge_{\alpha\in\{a,b\}}\forall^\out x\ 
\newline 
\alpha(x)\rightarrow
\{\alpha(\ori(x))\} \wedge \forall^\out x, y\ a(x)\wedge
                          b(y)\rightarrow x\leqoutput y$.
More generally, one could associate with any word $(ab)^n$ the set of
all well-parenthesised words of length $n$ over $\Gamma$. 

\begin{remark}\label{rm:nonclosure}
According to the previous examples, one can express in $\lorigin$ the
transduction $\tau_1$ defined as the shuffle over the language $a^*b^
*$, and also $\tau_2:(ab)^n\mapsto a^nb^ n$.
Hence the composition $\tau_2\circ \tau_1:a^nb^ n\mapsto a^ nb^ n$ has
a non-regular domain. However, as we will see in
Section~\ref{sec:properties}, the domain of an $\lorigin$-transduction
is always regular, which means that $\lorigin$-transductions are not
closed under composition.
\end{remark}

\section{Expressiveness, satisfiability and synthesis}\label{sec:properties}

\subsection{Expressiveness of  \lorigin}\label{subsec:expresslorigin}

Our first result is that $\lorigin$ can express all regular
functions. To show this result, we use their characterisation as deterministic
\mso-transducers~\cite{EH01}. We briefly recall that an
\mso-transducer is defined by some $\mso[\leq_\inp,\Sigma]$-formulas
interpreted over the input word structure (with linear order denoted
here by $\leq_\inp$), which specify the predicates of the output
word structure, the domain of which are copies of the input nodes. More
precisely, a constant $k$ specifies  the number 
of copies of the input word structure, $\mso[\leq_\inp]$-formulas $\phi_{pos}^c(x)$
specify whether the $c$th copy of node $x$ is kept in the
output structure, monadic formulas $\phi_{\gamma}^c(x)$
for each copy $c\in\{1,\dots,k\}$ and $\gamma\in \Gamma$, specify
whether the $c$th copy of input node $x$ is labelled $\gamma$ in the output
structure, and ordering formulas $\phi_{\leq_\out}^{c,d}(x,y)$,
say if the $c$th copy of $x$ is before the $d$th copy of
$y$ in the output. 

\begin{restatable}{theorem}{thmMSOexpr}\label{thm:msoexpr}
Any regular function is \lorigin-definable.
\end{restatable}

\begin{proof}[Sketch of proof]
    Let $f$ be a regular function. Since it is regular, there exists an \mso-transducer
    defining it. We convert it into an
    $\lorigin$-formula. First, it is not difficult to define an  
    $\mso[\leq_\inp,\Sigma]$-formula $\phi_{c_1,\dots,c_l,v}(x)$,
    $c_1,\dots,c_l\in\{1,\dots,k\}$ and $v\in\Gamma^*$, 
    which
    holds true if and only if in the output structure generated by the
    \mso-transducer, the copies of $x$ that are used are exactly
    $c_1,\dots,c_l$, they occur in this order in the output structure,
    and they are respectively labelled $v(1),\dots,v(l)$. In other
    words, input position $x$ generates the subword $v$ in the output
    structure. Then, we define $\lorigin$-formulas $C_i(x)$, for all
    $i\in\{1,\dots,k\}$ and $x$ an output node (in the o-graph),
    which hold, respectively, iff $x$ is the $i$th node (in the output order) whose
    origin is $\ori(x)$. This can be done using only two
    variables: $C_1(x)\equiv \out(x)\wedge \forall^\out y,\  y <_\out x
    \rightarrow  \{\ori(x) \neq \ori(y)\} $ and for $i\geq 1$, $C_{i+1}(x)\equiv$
$$ \begin{array}{ll}
 \exists^\out y\ (y<_\out x \wedge \{ \ori(x) {=} \ori(y)\} \wedge C_i(y)) \wedge \big(\forall y\ (y<x \wedge  & \\
   \{ \ori(x) {=} \ori(y)\})\to \neg (\exists^\out x\ (x<_\out y \wedge \{ \ori(x) {=} \ori(y)\} \wedge C_i(x))\big) & \\
       
\end{array}
$$
Finally, we construct the final \lorigin-formula (omitting some minor details)
as a conjunction, for all $m,l\leq
    k$, all copies $ c_1,\dots,c_l$ and $d_1,\dots,d_m$, all words
    $v\in\Gamma^l$ and $w\in\Gamma^m$, all $i\leq l$ and $j\leq m$, of
    the formulas:
    
$$
\begin{array}{lll}
\forall^\out x,y\ \\
 \ \big(\{     \phi^{c_i,d_j}_{\leq} (\ori(x),\ori(y))\wedge
    \phi_{c_1,\dots,c_l,v} (\ori(x)) \wedge \phi_{d_1,\dots,d_l,w}(\ori(y))\}  \\
\ \wedge C_i(x)\wedge C_j(y)\big) \rightarrow \tuple{x\leq_\out y\wedge v(i)(x) \wedge w(j)(y)}
\end{array}
$$
\end{proof}

\mso-transducers have been extended with nondeterminism
(\nmso-transducers or just \nmsot) to express
non-functional transductions, by using a set of 
monadic second-order parameters $X_1,\dots,X_n$ \citep{EH01}. Each formula of
an $\nmso$-transduction can use $X_1,\dots,X_n$ as free variables. Once
an interpretation for these variables as sets of positions has been fixed,
the transduction becomes functional. Therefore, the maximal number of 
output words for the same input word is bounded by the number of
interpretations for $X_1,\dots,X_n$. \nmso-transducers are
linear-size increase (the length of any output word is linearly
bounded by the length of the input word), hence the universal transduction $\Sigma^+\times
\Gamma^*$ is not definable in $\nmso$, while it is
$\lorigin$-definable by $\top$. The shuffle transduction is not
definable in \nmsot as well (this can be shown by cardinality
arguments). Conversely, it turns out that a transduction like $(u,vv)$ where $v$
is a subword of $u$ of even length is not \lorigin-definable whereas is it in \nmsot.

Rational relations are transductions defined by (non-deterministic)  finite
transducers
(finite automata over the product monoid $\Sigma^*\times \Gamma^*$),
denoted 1NFT~\cite{Bers79}. This class is incomparable with $\lorigin$: the shuffle
is not a rational relation, while the relation $\{a\}\times L$, where
$L$ is a non-$\fo^2$-definable regular language is not
$\lorigin$-definable. Indeed,
when all inputs are restricted to the word $a$, the expressive power
of $\lorigin$ is then restricted to $\fo^2[\leq_\out,\Gamma]$ over the
output.

Non-deterministic two-way transducers (2NFT), are incomparable to \nmso
~\cite{EH01}, and also to \lorigin, since they extend 1NFT and cannot
define the shuffle transduction. Fig.~\ref{fig:expr} depicts these
comparisons, summarised by the following proposition: 


\begin{restatable}{proposition}{propIncomp}
    The classes of $\lorigin$, 2NFT (resp. 1NFT), and \nmsot-definable
    transductions are
    pairwise incomparable.
\end{restatable}

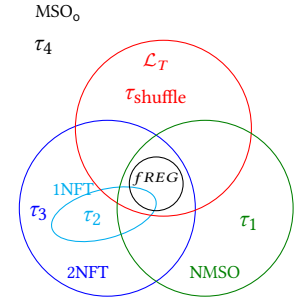
\begin{figure}[t]
\begin{center}
\def\firstcircle{(0,0) circle (1.8cm)}
\def\secondcircle{(55:2cm) circle (1.8cm)}
\def\thirdcircle{(0:2cm) circle (1.8cm)}

\begin{tikzpicture}[scale=0.65]
    \draw[color=blue] \firstcircle node[below left] {};
    \node[color=blue] (tnft) at (-.4,-1.3) {\footnotesize $\text{2NFT}$};
    \draw[color=cyan,rotate=15] (-0.1,-0.1) ellipse (1.1cm and 0.5cm);
    \node[color=cyan] (nft) at (-0.7,0.4) {\footnotesize $\text{1NFT}$};
    \draw[color=red] \secondcircle node [above] {};
    \node[color=red] (lt) at (1,3) {\footnotesize $\lorigin$};
    \draw[color=green!50!black]  \thirdcircle node [below right] {};
    \node[color=green!50!black] (nmsot) at (2.2,-1.3) {\footnotesize $\nmso$};
    \node (msot) at (1,0.6) {\tiny $fREG$};
    \draw (1,0.5) circle (0.55cm);

    \node (msor) at (-1,4) {\footnotesize $\mso_\ori$};

    \node[color=red]  (t1) at (1,2.3) { $\tau_{\text{shuffle}}$};
    \node[color=green!50!black] (t2) at (2.9,-0.3) { $\tau_1$};
    \node[color=blue] (t3) at (-1.4,-0.1) { $\tau_3$};
    \node[color=cyan] (t3) at (-0.3,-0.2) { $\tau_2$};
    \node (t9) at (-1.3,3.3) { $\tau_4$};
\end{tikzpicture}
\end{center}

\caption{Expressiveness of \lorigin, compared to $\mso_\ori$,
  non-deterministic \mso transductions, non-deterministic one-way and two-way
  transducers and regular functions. Here, $\tau_1=\{(u,vv)\mid v\text{ is a subword of $u$ of even length}\}$,
  $\tau_2=\{a\}\times (ab)^*$, $\tau_3=\{(u,u^n)\mid n\geq 0\}$ and $\tau_4=\{a^nb^n,(ab)^n\mid\ n>0\}$. }\label{fig:expr}
\end{figure}

\subsection{Satisfiability and equivalence problems}

Our first main contribution is the following result, whose proof is sketched in
Section~\ref{sec:sketch}. Here and throughout the paper, by
effectively we mean that the proof effectively constructs a finite object.

\begin{restatable}{proposition}{propRegDom}\label{prop:regDom}
    The input domain of any $\lorigin$
    -transduction is (effectively) regular. 
\end{restatable}
\begin{theorem}\label{Thm-LoriginSat} Over $o$-graphs, the
    logic $\lorigin$ has decidable satisfiability problem. 
\end{theorem}

This latter theorem is a consequence of Thm~\ref{thm:unif}.
We point out that it holds also for origin-free
transductions, because given an $\lorigin$-formula $\phi$,
$\sem{\phi}=\varnothing$ iff $\semo{\phi}=\varnothing$. 
The \emph{equivalence
problem} asks, given two formulas $\phi_1,\phi_2$, whether 
$\semo{\phi_1} = \semo{\phi_2}$, \ie whether $\phi_1\leftrightarrow
\phi_2$ is universally true. As a consequence of
Thm.~\ref{Thm-LoriginSat} and
closure under negation of $\lorigin$ we have the decidability of the equivalence problem for $\lorigin$.

With respect to satisfiability, $\lorigin$ seems to lie
at the decidability frontier. Adding just the successor relation over
outputs already leads to undecidability, by Prop.~\ref{prop:undecfo2}.





\subsection{Regular synthesis of \lorigin and consequences}

Our main result is the regular synthesis of
$\lorigin$-transductions. 

\begin{restatable}[Regular synthesis of \lorigin]{theorem}{thmUnif}
\label{thm:unif}
Let $\phi$ be an \lorigin formula. The transduction defined by $\phi$ is
(effectively) realisable by a regular function.
\end{restatable}

In other words, from any specification $\phi$ written in $\lorigin$, one can
synthesise a functional transduction  $f$,  in the proof
represented by an \mso-transducer $T$, such that $\dom(f) =
\dom(\sem{\phi})$ and $f = \sem{T}\subseteq \sem{\phi}$. Moreover, it turns out
that the constructed transducer $T$ defines a functional
o-transduction $\semo{T}$ such that $\semo{T}\subseteq
\semo{\phi}$. In other words, $T$ does not change the origins
specified in $\phi$. Since we rely on MSO-to-automata translation in the
construction, the size of the constructed \mso-transducer is
non-elementary in the size of $\phi$. %
%
%
One of the main consequences of the synthesis and expressiveness results is a new
characterisation of the class of regular functions. 

\begin{theorem}[New characterisation of regular functions]
    Let $f : \Sigma^*\rightarrow \Gamma^*$. Then, $f$ is regular iff $f$ is
    \lorigin-definable. 
\end{theorem}
\begin{proof}
    By Thm.~\ref{thm:msoexpr}, $f$ regular implies $f$ is
    $\lorigin$-definable, 
    which implies by Thm.~\ref{thm:unif} that $f$ is regular.
\end{proof}

A consequence of synthesis is the following positive
result on functionality:

\begin{corollary}[Functionality]
    Given an $\lorigin$-sentence $\phi$, it is decidable whether the
    o-transduction $\semo{\phi}$ is functional.
\end{corollary}

\begin{proof}
To test whether $\semo{\phi}$ is functional, first realise it
by a regular function (Thm.~\ref{thm:unif}), defined \eg by a
deterministic two-way transducer $T$, and then test
whether $\semo{\phi} \subseteq \semo{T}$. The latter is decidable since $T$ can
be converted (while preserving origins) into an equivalent \lorigin-formula  $\psi$ (Thm.~\ref{thm:msoexpr})
 and test that $\phi\rightarrow \psi$ is satisfiable (Thm.\ref{Thm-LoriginSat}).
\end{proof}



\begin{figure*}[t]
\begin{center}
\includegraphics[scale=0.19]{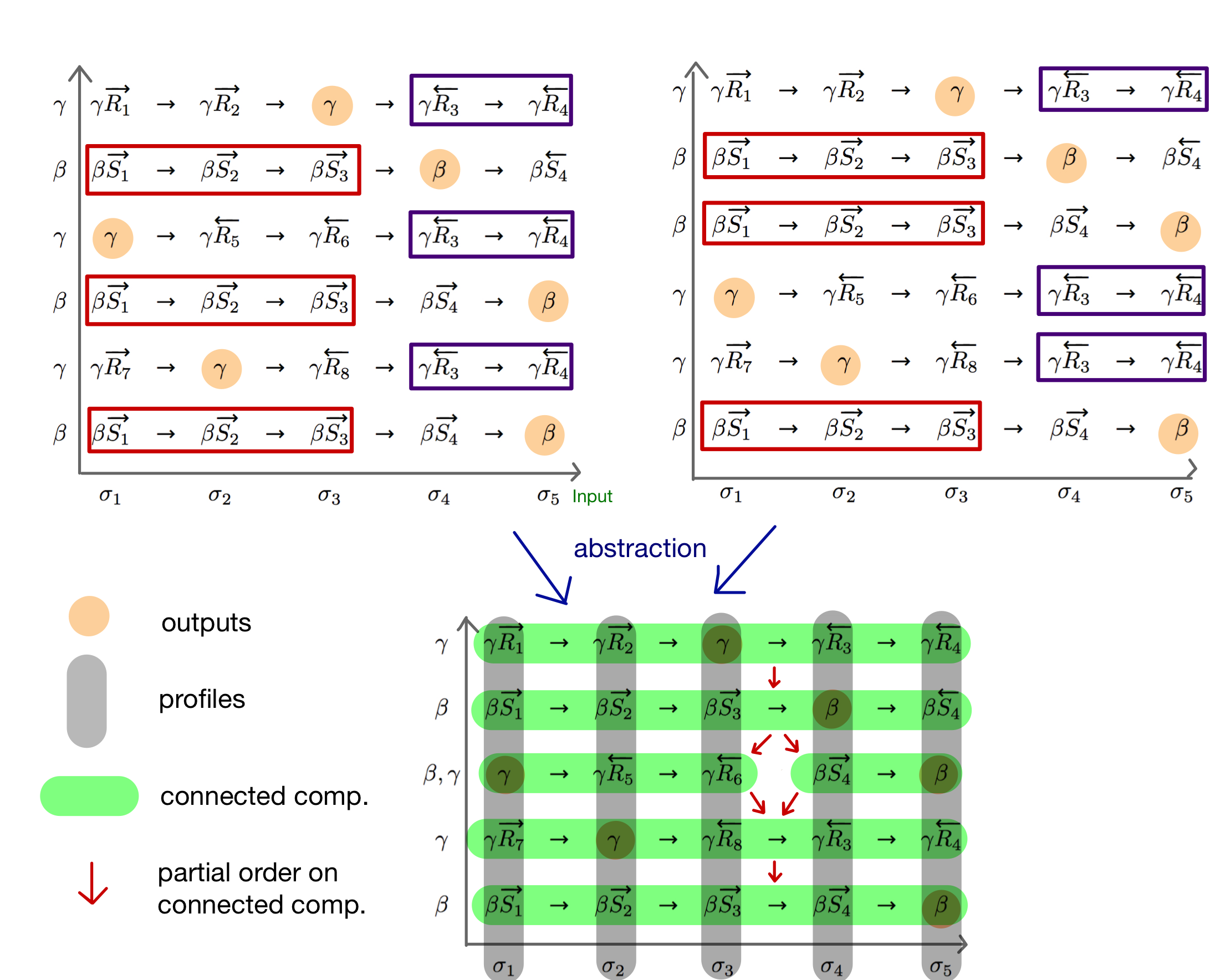}
\caption{The profile abstraction and the graph of
  clauses}\label{fig:abstraction}
\end{center}
\end{figure*}

\section{Domain regularity and synthesis: sketch of proofs}\label{sec:sketch}
In this section, we sketch  the proofs of Prop.~\ref{prop:regDom} (domain
regularity of $\lorigin$-transductions) and
Thm.~\ref{thm:unif} (regular synthesis). These two results are
based on common tools which we now describe. We let $\phi$ be an
$\lorigin$-sentence over input and output alphabets $\Sigma,\Gamma$
respectively. We assume that $\lorigin$ defines a non-erasing
o-transduction, \ie an o-transduction which uses every input position
at least once (the origin mapping is surjective). This can be done
\wog, \ie one can transform in polynomial time any
$\lorigin$-sentence into a non-erasing one (by adding dummy output
positions the origins of which are the erased input positions), while preserving the
domain and set of regular functions realising it (modulo the previous
encoding).

\paragraph{Scott normal form} The $\lorigin$
formula $\phi$ is then transformed into a Scott normal form (SNF),
a standard transformation when dealing with two-variable
logics (see for instance~\cite{GO99}).
By enriching the alphabet, the transformation allows to restrain ourself to the easier setting of formulas of quantifier-depth two.
Precisely, we obtain a formula of the form: $$\forall^\out x,y\ \psi(x,y)\wedge\bigwedge_{i=1}^m\forall^\out x\exists^\out y\ \psi_i(x,y)$$
where the formulas $\psi$ and $\psi_i$, $i=1,\dots,m$, are
quantifier free, but over an extended output alphabet $\Gamma\times
\Gamma'$ (where $\Gamma'$ may be exponential in $\phi$). 
These subformulas can also still contain binary MSO
predicates over the input, which are not restricted in any way. 
Up to projection over $\Gamma$, the SNF formula accepts the
same models as $\phi$, and hence we now just assume that $\phi$ 
is a formula of the above form over an input alphabet $\Sigma$ and
output alphabet $\Gamma$. In the full proof (Appendix), the SNF is
further equivalently transformed into what we call a system of universal and
existential constraints (in the vein of \cite{SZ12}), which are easier
to manipulate in the proofs than the formulas $\psi$ and $\psi_i$, but
are not necessary at a conceptual level, so we do not include them in
the sketch.


\paragraph{The profile abstraction} We define an abstraction
which maps any o-graph $(u,(v,o))$ to a sequence of $|u|$ tuples
$\lambda_1\dots \lambda_{|u|}$ called \emph{profiles}, one for each
input position. A profile contains bounded information (bounded in
the size of $\phi$) about the binary input \mso predicates, the input
symbol and some output positions. To explain this abstraction, we
first informally define what we call the \emph{full graph} of an
o-graph $(u,(v,o))$. Intuitively, the full graph contains a node for
each pair $(p,p')\in\dom(u)\times \dom(v)$, labelled by some
information called \emph{clause} about the ``effect'' of position $p'$
at position $p$. To understand it, it is convenient to
see the full graph as a two-dimensional structure with the input
position as $x$-axis (ordered by $\leq_\inp$) and the output position
as the $y$-axis (ordered by $\leq_\out$). Figure~\ref{fig:abstraction} shows
such a representation. \Eg the
top-left figure represents the full graph of an $o$-graph which translates
$\sigma_1\dots \sigma_5$ into $(\beta\gamma)^3$ (for instance, the
origin of the last output position, labelled $\gamma$, is the third
input position, labelled $\sigma_3$),  plus some additional
information which we now detail.

Each row contains a single node labelled in $\Gamma$, corresponding to
an output position, and placed according to its origin. Let $(p,p')$
(output position $p'$ with origin $p$) be such a node, labelled by
some $\gamma\in\Gamma$. This node generates an horizontal trace around it, whose
elements are of the form $\gamma
\overleftarrow{R}$ or $\gamma \overrightarrow{R}$. 
The arrows indicate in which direction the $\gamma$-labelled node
is. The elements $R$, say at coordinates $(s,p')$, are
$\mso[\Sigma,\leqinput]$-types (of bounded quantifier rank) talking about the input word $u$ with
the positions $s$ and $p$ marked. In the proof, we represent these
\mso-types as state information of node selecting automata (or query automata, see \eg \cite{NPTT05}). 
The idea behind this information is that, by looking independently at each column of
the full graph of an o-graph, it is possible to decide whether this
o-graph satisfies $\phi$. Suppose for instance we want to check
whether the o-graph satisfies a formula of the form $\forall^\out
x\exists^\out y\cdot \gamma(x)\rightarrow \gamma'(y)\wedge y<_\out
x\wedge \{\xi\}(\ori(y),\ori(x))$. Then, for every column containing a
$\gamma$-labelled node, say at coordinate $(p,p')$, one has to check
that there exists a node in the same column, say at position $(p,p'')$, labelled by some
$(\gamma'\overleftarrow{R})$ or some $(\gamma'\overrightarrow{R})$, such that $p''<p'$ and $R$ satisfies
$\xi$. Suppose that in the SNF we also have a conjunct of the form 
$\forall^\out x,y\cdot (\gamma(x)\wedge \gamma'(y)\wedge \{\ori(x)<_\out \ori(y))\}\rightarrow
\{\xi'\}(\ori(x),\ori(y))$, then we must additionally checks that for
every column, for every $\gamma$-labelled node in this column and every
$\gamma' \overrightarrow{R}$-labelled node on the same column, $R$
satisfies $\xi'$. We call a column which satisfies the SNF formula
$\phi$ a \emph{valid} column. 

A \emph{key property} we now use is that, if on a column there exists at least three nodes
with the same label, then removing all but the smallest and greatest
(in output order) of these nodes does not influence the validity of
the column. It is easy to see for subformulas of $\phi$ of the form 
$\forall^\out x,y\ \psi(x,y)$ (removing nodes makes such a formula
``easier'' to satisfy). For subformulas of the form
$\forall^\out\exists^\out y\ \psi_i(x,y)$, it is due to the fact that
$\psi_i$ is quantifier-free, and therefore it is safe to keep only the
extremal witnesses $y$ for $x$.

This observation leads us to the notion of \emph{abstract graph}, the subgraph of the full graph
obtained by keeping only the extremal occurrences of every node with
same labels. Figure~\ref{fig:abstraction} illustrates this abstraction,
on hypothetical full graphs where label equalities have been
underlined. Each column indexed by position $p$ of this abstract graph, together with the
input symbol, is what we call the \emph{profile} of $p$. Note
that this is a bounded object. Then, to any o-graph one can associate
a sequence of profiles this way, but this association is not
injective in general since we may lose information, as shown in
the figure. Put differently, the abstract graph can in general be
concretised in more than one full graph. 

\paragraph{Properties of profile sequences} The key ingredient of the
proof is to define properties on profile sequences $s$ (which are
nothing but words over the finite alphabet of profiles), that can be
checked in a regular manner (by an automaton) so that there exists at
least one o-graph $g$ such that $(1)$  $s$ is the profile sequence of
$g$ and $(2)$ $g\models \phi$. Property $(2)$ is ensured by the notion
of validity defined before, and by a notion of \emph{maximality} for the
MSO-types $R$ (no information can be withheld). 
Property $(1)$ is ensured by a notion of consistency between
profiles. Intuitively, it asks that the information declared in one
profile is consistent, in some precise way, with the information
declared in the next profile.
Roughly, 
since we use automata to represent the information $R$, one step
consistency corresponds to one step in the runs of the automata. 
Maximal and consistent sequences of
valid profiles are called \emph{good} profile sequences.

We then prove a completeness result: the
profile sequence of any model of $\phi$ is good. We also prove a
soundness result: any good profile sequence is the profile sequence of at least one
model of $\phi$. As a matter of fact, we prove a slightly stronger
result which allows one to recover not just one but potentially
several models of $\phi$. As illustrated on the figure, every
connected component of the abstract graph corresponds to exactly one
node labelled in $\Gamma$. The notion of consistency ensures this
property as well, and, as a matter of fact, the output positions of
the models we reconstruct out of good profile sequences are in bijection
with these connected components (CC). We can even order them partially, as
illustrated on the figure, by overlapping: a CC is the successor of
another one if they overlap horizontally, and the former is above the
latter (again, our definition of consistency ensures that there is no ``crossing'' in
abstract graphs, hence this relation can indeed be shown to be a
partial order). Hence, a good profile sequence
defines an abstract graph which gives us: the input position with their labels, the output
positions with their labels and origins, and some partial order
between these output positions. What's missing is a linear order on
these output positions, but we prove that \emph{any} linearisation of
this partial order actually defines an o-graph which satisfies
$\phi$. Coming back to the example, the two possibly linearisations
give the output words $\beta\gamma\beta\beta\gamma$ and 
$\beta\gamma\gamma\beta\gamma$.

\paragraph{Back to the theorems} To show domain regularity
(Prop.~\ref{prop:regDom}), we observe that the domain is the projection
on input alphabet $\Sigma$ of the set of good profile sequences, which
turns out to be regular (the whole point of defining the notions of
validity, maximality and consistency is that they can be checked by an
automaton). Since regular languages are closed under projection, we get
the result.

Showing regular synthesis (Thm.~\ref{thm:unif}) is a bit more
technical. The main idea is to show that the mapping which takes as input
a word $u$ over $\Sigma$, and which outputs all the abstract graphs of
o-graphs which satisfy $\phi$ and have $u$ as input, is definable by a
non-deterministic 
MSO word-to-DAG transduction
$T_1$. It is possible since the notions of consistency, maximality and
validity are all MSO-definable, and an abstract graph is always a DAG. Then, we use a result of Courcelle which states that there
exists a deterministic MSO DAG-to-word transduction $R_2$ which, given a DAG,
produces a topological sort of it \cite{cou95}. The DAG additionally needs to be
\emph{locally ordered} (the successors of a node are linearly ordered),
but we can ensure this property in our construction. Then, we use
closure under composition of \nmsot to
show that $R_2\circ R_1$ is definable by some word-to-word \nmsot,
which can be easily realised by a (deterministic) \msot, concluding
the proof.

\paragraph{Comparison with \cite{SZ12}} We would like to point out
that this proof was inspired by a decidability proof for the logic
$\fo^2[\leq,\preceq,S_\preceq]$ over data words (a linear order over positions and
a linear order and successor over data). We somehow had to cast it to transductions,
and extend it with binary MSO predicates. Moreover, further
manipulations and notions were needed to extract the synthesis
result. In particular, the ideas of using Scott normal form, to see
o-graphs as two-dimensional structures, and the abstraction, were
directly inspired from \cite{SZ12}.

\section{A decidable logic for typed data words}\label{sec:datawords}

We make here a bijective connection between o-transductions and
what we call typed data words, which slightly generalise data words, 
and introduce a new decidable logic $\ldata$ for typed data words, whose decidability stems from the equivalence with \lorigin. 



\paragraph{Typed data words} 
We consider \emph{typed data words} over an ordered data domain, such
that each datum also carries a label (type) from a finite alphabet. 
Formally, a typed data word of length $n$ and \emph{data size} $m$ over two disjoint
alphabets $\Gamma$ and $\Sigma$ is a word over the alphabet
$\Gamma \times \nat \times \Sigma$, $w=(\gamma_1,d_1,\sigma_1) \cdots (\gamma_n,d_n,\sigma_n)$ verifying the following properties:
$d_i$ is called the \emph{datum of position $i$}, we have that $\set{d_1,\ldots, d_n}= \set{1,\ldots, m}$ \footnote{We make this assumption \wog, because the logic we
  define will only be able to compare the order of data, and so cannot
  distinguish typed data words up to renaming of data, as long as the
  order is preserved. \Eg $(a,b,1)(c,d,3)(e,f,2)$ and
  $(a,b,2)(c,d,5)(e,f,4)$ will be indistinguishable by the logic.}
  and we also have for any positions $i,j$ that $d_i=d_j \Rightarrow \sigma_i=\sigma_j$, hence $\sigma_i$ is called the \emph{type of datum $d_i$}.

We denote by
$\datawSG$ the set of typed data words over alphabets $\Sigma,\Gamma$ 
of \emph{any} length $n$ and \emph{any} data size $m$.

The data of a typed data word $w$ induce a total preorder 
$\preceq$ over the positions of $w$ defined by
$i\preceq j$ if $d_i\leq d_j$. This preorder induces itself an
equivalence relation $\sim$ defined by $i\sim j$ iff $i\preceq j$
and $j\preceq i$, which means that the positions $i$ and $j$ carry
the same datum. 
Hence, a typed data word will equivalently be seen as a structure with 
letter predicates $\gamma\in \Gamma$, $\sigma\in\Sigma$, the linear order over positions
and the total preorder $\preceq$ previously defined.

\paragraph{The logic $\ldata$ for typed data
  words}\label{subsec:ldata} It is known from 
\cite{BMSSD06} that the logic 
\mso
 over
untyped data words (\ie $|\Sigma|=1$) is undecidable (even the
first-order fragment). We consider here a decidable fragment, over
typed data words, called 
$\ldata$. A formula
of \ldata can be seen as an $\fo^2$ formula using the linear order of
the positions and some additional binary data predicates. The logic
$\ldata$ is indeed built on top of MSO $n$-ary
predicates, for $n\leq 2$, which are allowed to speak only about the
data. Precisely, we define $\binmso[\Sigma,\preceq]$ to be
the set of $n$-ary predicates written $\{\phi\}$, for $n\leq 2$, where $\phi$
is an MSO-formula with $n$-free first-order variables, over the
unary predicates $\sigma(x)$ and the preorder $\preceq$, with the
following semantic restriction\footnote{Note that the semantic restriction could also be enforced in the logic by guarding quantifiers
$\exists X\psi$ with $\exists X[\forall x\forall y\ x\in X\wedge
y\sim x)\rightarrow y\in X]\rightarrow \psi$.}: second-order variables are
interpreted by $\sim$-closed sets of positions. Over typed data words,
 predicates $\{\phi\}$ are interpreted by relations on positions defined by formulas $\phi$.

Due to the semantic restriction, formulas in 
$\binmso[\Sigma,\preceq]$ cannot distinguish positions with the same
data and therefore, they can be thought of as formulas which quantify over data and sets of data. 
 As an example, the
formula $\forall y\ x\preceq y$ expresses that the datum of position $x$ is the
smallest, and it holds true for any $x'$ with the same datum. 
Then, the
logic $\ldata$ is defined as $\ldata := \fo^2[\Gamma,\leq, \binmso[\Sigma,\preceq]]$.

%
\begin{example}\label{ex:ldata}
    First, let us mention that $\binmso[\Sigma,\preceq]$ predicates
    can express any regular properties about the data, in the
    following sense.     Given a typed data word $w$, the total preorder $\preceq$ over
    positions of $w$ can be seen as
    a total order $\leq_\sim$ over the equivalence classes of  $\dom(w)/_\sim$, by
    $[i]_\sim \leq_\sim [j]_\sim$ if $i\preceq j$. Then, any typed
    data word induces a word $\sigma_1\dots \sigma_n\in\Sigma^*$ such
    that $\sigma_i$ is the type of the elements of the $i$th
    equivalence class of $\leq_\sim$. Any regular property of these
    induced words
    over $\Sigma$ transfers into a regular property about the data of typed data words (it
    suffices to replace in the MSO-formula on $\Sigma$-words
    expressing the property, the linear order by $\preceq$ and the
    equality by $\sim$). Examples of properties are: $n$ is even,
    which transfers into ``there is an even number of pieces of data'', or 
    $\sigma_1\dots \sigma_n$ contains an even number of
    $\sigma\in\Sigma$, for some $\sigma$, meaning ``there
    is an even number of pieces of data of type $\sigma$''.


%
%
\end{example}

\paragraph{From transductions to data words and
  back}\label{sec:transtodata}
 %
%
There is a straightforward encoding $\td$ of non-erasing o-graphs into
typed data words, and conversely. 
A non-erasing o-graph $(u,(v,o))$, with $v=v_1\ldots v_n$ and $u=u_1\ldots u_m$ is encoded as the typed data word
$\td((u,(v,o)))=(v_1,o(1),u_{o(1)})\ldots(v_n,o(n),u_{o(n)})$. Given a typed
data words $w=(\gamma_1,d_1,\sigma_1)\ldots
(\gamma_n,d_n,\sigma_n)$, we set $\td^{-1}(w)$ the non-erasing o-graph 
$\td^{-1}(w) = (u,(v,o))$ such that $v=\gamma_1\ldots\gamma_n$, $o(i)
= d_i$, and if we write $d_{i_j}=j$ then $u=\sigma_{i_1}\cdots\sigma_{i_m}$ where $m = \text{max}_i\
d_i$. 
We give here an example of this transformation:
\begin{center}
\begin{tikzpicture}[>=stealth',auto,node distance=3cm,scale=0.75,every node/.style={scale=0.8}]

\node[] (i1) at (0,0) {$\#$};
\node[] (i2) at (1,0) {$\$$};
\node[] (i3) at (2,0) {$@$};
\node[] (i5) at (3,0) {$\#$};
\node[] (i6) at (4,0) {$\#$};

\node[] (j1) at (0,-1) {$a$};
\node[] (j2) at (1,-1) {$b$};
\node[] (j3) at (2,-1) {$c$};
\node[] (j4) at (3,-1) {$c$};
\node[] (j5) at (4,-1) {$a$};
\node[] (j6) at (5,-1) {$b$};

\draw [->] (j1) -- (i3) ;
\draw [->] (j2) -- (i2) ;
\draw [->] (j3) -- (i1) ;
\draw [->] (j4) -- (i3) ;
\draw [->] (j5) -- (i6) ;
\draw [->] (j6) -- (i5) ;

\foreach \i/\l/\n/\t in {0/a/3/@,1/b/2/\$,2/c/1/\#,3/c/3/@,4/a/5/\#,5/b/4/\#}
{
\node (k\i) at (5.5+\i,-.2) {$(\l,\n,\t)$};
}

%

%
 \end{tikzpicture}
\end{center}

%
%

\begin{restatable}{theorem}{thmOrtoData}\label{Thm-OriginToData}
\label{thm:origin-to-data}
    Non-erasing o-graphs
    of $\ProdSG$ and typed data words
    of $\dataw{\Sigma}{\Gamma}$ are in bijection by $\td$. 
Moreover, a non-erasing o-transduction $\tau$ is \lorigin-definable  iff
$\td(\tau)$ is \ldata-definable. Conversely, a language of typed data words
$L$ is \ldata-definable iff $\td^{-1}(L)$ is \lorigin-definable.
\end{restatable}

	
    The main idea of the proof is to make a bijective
    syntactic transformation that mimics the
    encoding $\td$: once inconsistent use of terms have been
    removed (such as \eg, $o(x)\leqoutput y$), terms $o^n(x)$
    are  replaced by $x$, predicates $\leqinput$ by $\preceq$ and
    $\leqoutput$ by $\leq$. 
Hence, this theorem and the decidability of \lorigin (Thm.~\ref{Thm-LoriginSat}) gives the following corollary.

\begin{corollary}\label{thm:ldata-dec}
    Over typed data words, the logic $\ldata$ has a decidable
    satisfiability problem. 
\end{corollary}

\noindent As a remark, we also note that thanks to the correspondence between
transductions and data words and some minor manipulations, we can also obtain the decidability 
of $\fo^2[\leq_\inp,\leq_\out,S_\inp,\ori]$ (for $\inp$ the input successor), which is a strict fragment of \lorigin, over o-graphs from
the decidability of $\fo^2[\leq,\preceq,S_\preceq]$ over data words, proved in \cite{SZ12}.
However, the logic $\fo^2[\leq,\preceq,S_\preceq]$ is a strict fragment of $\ldata$.


\section{Complexity of satisfiability}

To achieve decidability results for \lorigin, the binary \mso
predicates over the input of \lorigin-formulas are decomposed into
\mso-types, that we handle using query automata, as explained in the
sketch of proof in Section~\ref{sec:sketch}. A query automaton for a binary \mso formula $\psi(x,y)$ is
a non-deterministic finite automaton $\mathcal{A}=(Q,\Sigma,I,\Delta,F)$ equipped
with a set $\Select\subseteq Q^2$ of \emph{selecting pairs} 
with the following property: for any word $u\in\Sigma^*$ and any pair of positions
$(i,j)$ of $u$, we have $u\models \psi(i,j)$ if, and only if, there
exists an accepting run $\pi$ of $\mathcal{A}$ and a pair $(p,q)\in\Select$
such that $\pi$ reads $u(i)$ in state $p$ and $u(j)$ in state $q$.
Due to the \mso-formulas and their translation into query automata,
the complexity of the satisfiability of \lorigin is non-elementary,
and this is unavoidable~\cite{Stock74}.
However, if the binary \mso-formulas are already given as query
automata, we get a tight elementary complexity. Likewise, the binary
MSO predicates of the data word logic \ldata can be also represented
as query automata, and we get the same complexity as \lorigin. 

\begin{restatable}{theorem}{ThmComplex}\label{thm-Complex}
The satisfiability problem of \lorigin and \ldata is
\expspace-complete when the binary \mso predicates are given as query automata.
\end{restatable}
\begin{proof}[Sketch of proof]
First, as the translation between \lorigin and \ldata is linear, the complexity of both logics is equivalent.
In showing decidability of the satisfiability of \lorigin, we obtain that the set of "good" profile sequences is effectively regular by Prop.~\ref{prop:regDom}.
With a careful analysis it is possible to construct a doubly exponential deterministic automaton recognising the good sequences. By checking emptiness on-the-fly instead of constructing the automaton, we get the \nlogspace emptiness of the automaton, and hence the \expspace complexity. 
Finally, since the logic $\fo^2[\Gamma,\preceq, S_\preceq,\leq]$ is \expspace-complete~\cite{SZ12}, we get \expspace-hardness as this logic is a syntactic fragment of \ldata.
\end{proof}

\section{Decidable Extensions of \lorigin}
\label{sec:ext}
We present here two main extensions to $\lorigin$ showing its robustness. The first one consists in adding a block of existential monadic second-order quantifiers in front of the formula while the second one consists in adding new predicates to the logic; both extensions preserve many properties of the logic which we describe below.

\paragraph{Existential \lorigin}
This new logic is denoted by $\elorigin$ and allows us to capture all non-deterministic \mso-transductions, but we lose the closure under negation of the logic.
Formally, we consider all formulas of the form
$\exists X_1\dots \exists X_n \phi$ where $\phi$ is a formula of
$\lorigin$ which can additionally use predicates of the form $x\in
X_i$. The variables $X_i$ range over sets of output, and also input
positions.

\begin{restatable}{proposition}{NmsoToElorigin}
Any \nmso-transduction is \elorigin-definable.
\end{restatable}

The synthesis result extends to $\elorigin$ using a quite common trick of considering for a formula $\exists X_1\dots \exists X_n \phi$, the formula $\phi$ but over 
an extended alphabet.
\begin{restatable}{proposition}{UnifElorigin}\label{prop-UnifElorigin}
Any \elorigin-transduction can be (effectively) realised by a regular function.
\end{restatable}

One result of \lorigin does not carry over to \elorigin, namely the decidability of the equivalence problem. Indeed \elorigin is not closed under negation and thus equivalence of formulas cannot be reduced to satisfiability. Equivalence turns out to be undecidable for \elorigin and in fact the validity problem, which asks given a formula if it is satisfied by all o-graphs and which can be seen as the particular case of the equivalence with the formula $\top$, is itself undecidable for \elorigin.
\begin{restatable}{proposition}{UnivElorigin}
The validity and equivalence problems for \elorigin over o-graphs are undecidable. 
\end{restatable}

\paragraph{Single-origin predicates}
One ``weak'' point of $\lorigin$ is that if the input is restricted to, for instance, a single position, then the expressive power over the output is only $\mathrm{FO}^2[\leq_{\out}]$. For instance the transduction $\{a\}\times L$ is not definable if $L$ not an $\mathrm{FO}^2[\leqoutput]$-definable language.
A more general expression of this problem is that the class of transductions definable by one-way transducers, \aka rational transductions \citep{Bers79}, is incomparable with the class of \lorigin (resp. \elorigin) transductions.
The following extension, called $\lorigin^{\mathrm {so}}$  adds new predicates, called here \emph{single-origin predicates}, and we show that it captures all the rational transductions. These new predicates allow to test any regular property of a subword of the output word restricted to positions with a given origin position.

Given an o-graph $(u,(v,o))$ and an input position $i$ of $u$, we denote by $v_{|i}$ the subword of $v$ consisting of all the positions of $v$ whose origin is $i$, and we call this word the \emph{single-origin restriction of $v$ to $i$}.


Given any regular language $L$ (represented as an MSO formula for
instance), we define a unary input predicate $L(x)$, whose semantics
over an $o$-graph $(u,(v,o))$ is the set of input positions
$i\in\dom(u)$ such that $v_{|i}\in L$. The logic
$\lorigin^{\mathrm{so}}$ (resp. $\elorigin^{\mathrm{so}}$)
 is the extension of $\lorigin$ (resp. $\elorigin$) with the predicates $L(x)$,
for any regular language $L$. These predicates can be used just as the other unary input predicates and using the previous notation we have $\lorigin^{\mathrm{so}}:=\fo^2[\Gamma,\leqoutput,\ori,\binmso[\leqinput,\Sigma\uplus \{L(x)|\ L \text{ regular}\}]]$. For instance, let $L$ denote the language $(ab)^*$ then the formula $\foroutput x\ a(x) \rightarrow \{\mathrm{even}(\ori(x)) \wedge L(\ori(x)) \}$ states that the origin of each output position $x$ labelled by $a$ must be even and that the subword of origin $o(x)$ must be in $L$.

\begin{restatable}{proposition}{NFTtoLTSO}
Any rational transduction is $\lorigin^{\mathrm {so}}$ definable.
\end{restatable}

Our synthesis result transfers to $\lorigin^{\mathrm {so}}$ (and $\elorigin^{\mathrm {so}}$):
\begin{restatable}{proposition}{UnifLTso}
Any $\elorigin^{\mathrm {so}}$-transduction can be (effectively) realised by a regular function.
\end{restatable}

\begin{remark}
From the regular synthesis of $\elorigin^{\mathrm {so}}$, we can
deduce several results which we express in their strongest form: The
input domain of any $\elorigin^{\mathrm {so}}$-transduction is
effectively regular, the satisfiability problem for
$\elorigin^{\mathrm {so}}$ is decidable, the equivalence problem for
$\lorigin^{\mathrm {so}}$ is decidable. Finally, given a functional
transduction, it is regular if and only if it is $\elorigin^{\mathrm
  {so}}$-definable, and, given an $\elorigin^{\mathrm {so}}$ sentence
$\phi$, it is decidable whether $\sem{\phi}_o$ is functional.
\end{remark}

\paragraph{Extended logics over data words}
We define similarly the extensions $\eldata$, $\ldata^{\mathrm {sd}}$ and $\eldata^{\mathrm {sd}}$ of the logic \ldata and we obtain the same transfer results as in Thm.~\ref{thm:origin-to-data}. In terms of data, the single-origin predicates become \emph{single datum predicates (sd)} which can specify any regular property over a subword induced by a single datum.

\section{Summary and Discussion}
\begin{figure}
\begin{tikzpicture}

\node (nft) at (0,-0.5) {$1NFT$};
\node (reg) at (2,-0.7) {$fREG$};

\node (tft) at (0,0.5) {$2NFT$};
\node (lt) at (2,0) {$\lorigin$};
\node (nmso) at (4,0) {$\nmso$};

\node (ltso) at (2,1) {$\lorigin^{\mathrm {so}}$};
\node (elt) at (4,1) {$\elorigin$};

\node (eltso) at (3,2) {$\elorigin^{\mathrm {so}}$};

\node (msoo) at (2,2.7) {$\mso_\ori$};

\foreach \a/\b in {reg/tft,reg/lt,reg/nmso,nft/tft,lt/ltso,lt/elt,nmso/elt,ltso/eltso,elt/eltso,tft/msoo,eltso/msoo}
\draw[->, >=stealth] (\a) -- (\b);

\draw[->,>=stealth] (nft) to [bend left=10] (ltso);

\draw[color=white!70!black] (-0.5,2.4) -- (5.5,2.4);
\node[color=white!50!black] (synt) at (5,2.2) {synthesis};
\node[color=white!50!black] (synt) at (5,1.8) {+satisfiability};

\draw[color=white!70!black] (5,0.1) arc (60:80:17);
\node[color=white!50!black] (synt) at (6.2,0.5) {equivalence (with origin)};

\end{tikzpicture}
\caption{Summary of models for transductions and their inclusions. The lines are decidability frontiers.}\label{Fig-Summary}
\end{figure}
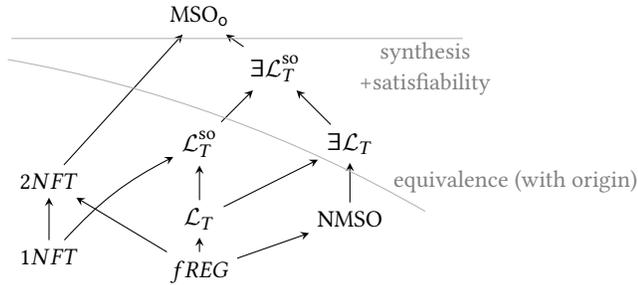
In this paper, we have introduced an \emph{expressive} logic to define
transductions, which we believe is a great tool from both a
theoretical and a more practical point of view. It allows for high-level specification of transductions, 
while having some good properties for synthesis. As an interesting
side contribution, we obtain a new characterisation of the class of
regular transductions, as the (functional) transductions definable in $\lorigin$ (and its
extensions up to $\elorigin^{\text{so}}$). The expressiveness and decidability frontiers on the logic $\lorigin$
and its extensions are summarised in
Fig.~\ref{Fig-Summary}.
We obtained tight complexity results for satisfiability of $\lorigin$
both in the case of binary input predicates given by \mso-formulas
(non-elementary) or automata (\expspace).
We have also shown that slightly
extending the expressiveness by adding the successor over output
positions leads to undecidability. 

Another question is the definition of an automata model equivalent to
$\lorigin$, or even to $\mso_\ori$. Automata for data words have been defined~\cite{BMSSD06,MZ13}, but
none of these models capture $\ldata$. 
 
 
The equivalence problem for $\lorigin$ is origin-dependent. One could relax it by projecting
away the origin information: given two $\lorigin$-formulas
$\phi_1,\phi_2$, are the origin-free transductions they define equal,
\ie $\sem{\phi_1} = \sem{\phi_2}$ ? This (origin-free) equivalence
problem is known to be decidable for regular functions~\cite{Gurari82}, and
undecidable for 1NFT (and hence 2NFT)~\cite{DBLP:journals/jacm/Griffiths68} as well as NMSOT~\cite{AD11}. It is shown by
reduction from the Post Correspondence Problem and it turns out that
the transductions constructed in the reduction of
\cite{DBLP:journals/jacm/Griffiths68} are definable in $\lorigin$,
proving undecidability for $\lorigin$ as well. An interesting line of research would be to consider less drastic relaxations of the
equivalence problem with origin, by comparing transductions with
\emph{similar} origin, as done for instance in~\cite{conf/icalp/FiliotJLW16} for rational
relations. Similarly, the model-checking of two-way transducers against
$\mso_\ori$-sentences is decidable, but it is again
origin-sensitive. Instead, the origin-free version of this problem is
to decide whether for all the pairs of words $(u,v)$ defined by a
two-way transducer $T$, there exists some origin mapping $o$ such that
the $o$-graph $(u,v,o)$ satisfies some formula $\phi$
. Once again, it is possible
to show, by reducing PCP, that this relaxation yields
undecidability, but it could be interesting to consider a stronger
problem where the origin of $T$ is ``similar'' to the origin specified in
$\phi$. 
A related problem is the satisfiability of
logics where two or more origin mappings are allowed.


Another direction would be extending the logic to other structures (\eg 
trees or infinite words), and other predicates over output
positions. However, one has to be careful since the data point of view
shows how close we are to undecidability (\eg over data words,
$\fo^2$ with successor over data and positions is
undecidable~\cite{journals/corr/ManuelSZ13}).


Finally, we have established a tight connection between transductions and data
words, and thus a new decidable logic for data words. The data point of view allowed
us to get decidability of the transduction logic $\lorigin$, inspired
by the decidability result of~\cite{SZ12}. 
Conversely, the logic $\ldata$ extends the known results on data words by adding \mso predicates on the ordered and labelled data. 
We would like to investigate if
further results from the theory
of transductions can be translated into interesting results in the
theory of data words. 


\bibliography{biblio}
\newpage

\input{appendix}
\end{document}

%% file: appendix.tex
\appendix

\section{Logics with origin for transductions}

\propUndecfoD*
\begin{proof}
    The proof is a reduction from the Post Correspondence
    Problem (PCP) and is an adaptation to
o-tranductions of the undecidability, over data words, of
$\fo^2$ with a linear order and successor predicates over positions, and a
linear-order on data \cite{BMSSD06}.

    Given an alphabet $A$ and $n$ pairs
    $(u_i,v_i)\in A^+\times A^+$ (they can be assumed to be non-empty
    without losing undecidability), we construct a sentence 
    $\phi\in\fo[\mathcal T_{\Sigma,\Gamma}]$ which is satisfiable iff
    there exist $i_1,\dots,i_k\in\{1,\dots,n\}$ such that 
    $u_{i_1}\dots u_{i_k} = v_{i_1}\dots v_{i_k}$. 
    We let $\Sigma = A$ and $\Gamma = A_1\cup A_2$, where $A_i =
    A\times \{i\}$. Given a word $w=a_1\dots a_p\in A^*$ and $\ell=1,2$, we let 
    $\ell(w) = (a_1,\ell)\dots (a_p,\ell)\in A_\ell^*$. For any
    two sequences of words $s = w_1,w_2,\dots,w_k \in A^*$ and 
    $s' = w'_1,\dots,w'_k\in A^*$, we define $s\otimes s'\in \Gamma^*$ their
    interleaving, by $1(w_1)2(w'_1)1(w_2)2(w'_2)\dots
    1(w_k)2(w'_k)$. \\
    \Eg $(ab,ca)\otimes (a,bca) =
    (a,1)(b,1)(a,2)(c,1)(a,1) (b,2) (c,2)(a,2)$.

    We will construct the formula $\phi$ in such a way that it defines
    the o-tranduction from $\Sigma$ to $\Gamma$ which maps any word
    $u\in \Sigma^*$ for which there exist
    $i_1,\dots,i_k\in\{1,\dots,n\}$ such that $u_{i_1}\dots u_{i_k} =
    u = v_{i_1}\dots v_{i_k}$, to 
    $w = (u_{i_1},\dots, u_{i_k})\otimes (v_{i_1},\dots, v_{i_k})$, with
    origin mapping $o$ which maps any position of $w$ corresponding to
    some $u_{i_j}$ (or to some $v_{i_j}$) to the same position in
    $u$. \Eg, over $A = \{a,b,c\}$, if one takes
    $u_1 = ab$, $u_2 = ca$, $v_1 = a$, $v_2 = bca$, then the sequence
    $1,2$ is a solution to PCP, and it gives rise to the following
    o-graph:

\begin{center}
\begin{tikzpicture}[>=stealth',auto,node distance=3cm]
\node[] (i1) at (0,0) {$a$};
\node[] (i2) at (1,0) {$b$};
\node[] (i3) at (2,0) {$c$};
\node[] (i4) at (3,0) {$a$};

\node[] (j1) at (-2,-1) {$a,1$};
\node[] (j2) at (-1,-1) {$b,1$};
\node[] (h1) at (0,-1) {$a,2$};
\node[] (j3) at (1,-1) {$c,1$};
\node[] (j4) at (2,-1) {$a,1$};
\node[] (h2) at (3,-1) {$b,2$};
\node[] (h3) at (4,-1) {$c,2$};
\node[] (h4) at (5,-1) {$a,2$};

\draw [->] (j1) -- (i1) ;
\draw [->] (j2) -- (i2) ;
\draw [->] (j3) -- (i3) ;
\draw [->] (j4) -- (i4) ;

\draw [->] (h1) -- (i1) ;
\draw [->] (h2) -- (i2) ;
\draw [->] (h3) -- (i3) ;
\draw [->] (h4) -- (i4) ;

\draw[decoration={brace,mirror,raise=5pt},decorate]
  (-2.3,-1) -- node[below=6pt] {$u_1$} (-0.7,-1);

\draw[decoration={brace,mirror,raise=5pt},decorate]
  (-0.3,-1) -- node[below=6pt] {$v_1$} (0.3,-1);

\draw[decoration={brace,mirror,raise=5pt},decorate]
  (0.7,-1) -- node[below=6pt] {$u_2$} (2.3,-1);

\draw[decoration={brace,mirror,raise=5pt},decorate]
  (2.7,-1) -- node[below=6pt] {$v_2$} (5.3,-1);
 \end{tikzpicture}
\end{center}

    First, we express that the output word is of the form 
    $(u_{i_1},\dots,u_{i_k})\otimes (v_{i_1},\dots,v_{i_k})$ for some 
    $i_1,\dots,i_k$. For that, we need to define a formula $\phi_{\text{cut}}(x)$ which holds
    true at output position $x$ if either $x$ is the first output
    position, or it is labelled in $A_1$ while its predecessor is
    labelled in $A_2$:
    $$
    \phi_{\text{cut}}(x)\equiv \forall^\out
    y\cdot S_\out(y,x)\rightarrow A_1(x)\wedge A_2(y)
    $$
    where for all $\ell=1,2$, $A_\ell(x)$ stands for $\bigvee_{a\in A}
    (a,\ell)(x)$. 

    Now, the idea when $x$ is a cut (\ie satisfies the formula
    $\phi_{\text{cut}}(x)$), is to  guess an index $i\in\{1,\dots,n\}$
    and check that the sequence of labels from position $x$ ($x$
    included) to the next cut (if it exists) or to the end (if not) is 
    $1(u_i)2(v_i)$.

    To define this, we introduce, for all 
    formulas $\phi$ with one free variable, the formula
    $\Lbag\phi\Rbag_j(x)$ which holds true if the $j$-th successor of $x$ exists 
    and satisfies $\phi$. It is inductively defined by:
    $$
    \Lbag \phi \Rbag_0(x)\equiv \phi(x) \qquad \Lbag\phi\Rbag_j(x)\equiv \exists^\out y\cdot S_\out(x,y)\wedge \Lbag\phi\Rbag_{j-1}(y)
    $$
    where $y$ is a variable  different from $x$. 
    Then, we define the following formula for $i\in\{1,\dots,n\}$:
    $$
    \begin{array}{lllllll}
    \phi_{u_i,v_i}(x)& \equiv & 
                                \bigwedge_{j=1}^{|u_i|} \Lbag (u_i(j), 1)(x) \Rbag_{j{-}1}(x)
      \\
      & & \wedge \bigwedge_{j=1}^{|v_i|} \Lbag (v_i(j), 2)(x)
          \Rbag_{j{-}1+|u_i|}(x) \\
& & \wedge \Lbag \phi_{\text{cut}}(x)\vee \text{max}_\out(x)\Rbag_{|u_i|+|v_i|-1}(x)
   \end{array}
     $$
     Finally, the following formula expresses that the 
     output word is of the form $(u_{i_1},\dots,u_{i_k})\otimes
     (v_{i_1},\dots,v_{i_k})$ for some $i_1,\dots,i_k$:
     $$
     \phi_{\text{well-formed}}\equiv \forall^\out x\cdot (
     \phi_{\text{cut}}(x)\rightarrow \bigvee_{i=1}^n \phi_{u_i,v_i}(x))
     $$

     So far, we have not checked any property of the origin mapping,
     nor the fact that the output decomposition satisfies
     $u_{i_1}\dots u_{i_k} = v_{i_1}\dots v_{i_k} = u$ if $u$ is the
     input word. To achieve that, it remains to express, for all $\ell=1,2$,  that the
     origin mapping restricted to positions labelled in $A_\ell$ is
     bijective and preserves the orders and labels.
     $$
     \begin{array}{rllllllllll}
     \phi_{\text{bij},\ell} & \equiv & \forinput x \exists^\out
     y\ A_\ell(y)\wedge\ori(y)=x\\
         && \!\!\!\!\wedge \forall^\out x,y\
            (\ori(x)=\ori(y)\wedge A_\ell(x)\wedge
            A_\ell(y)) \\
 & &\qquad \qquad \rightarrow x=y \\
       \phi_{\text{ord-pres},\ell} & \equiv & \forall^\out x,y        \           (\ori(x)<_\inp \ori(y)\wedge
                                              A_\ell(x)\wedge
                                              A_\ell(y))
       \\
   & & \qquad \qquad \rightarrow x<_\out y 
\\
\phi_{\text{lab-pres},\ell}  &\equiv & \forall^\out x\bigwedge_{a\in A}
                                         (a,\ell)(x)\rightarrow a(\ori(x)))
     \end{array}
     $$
The final formula $\phi$ is then:
$$
\phi\equiv     \phi_{\text{well-formed}} \wedge \bigwedge_{\ell=1}^2
     \phi_{\text{bij},\ell}\wedge \phi_{\text{ord-pres},\ell}\wedge \phi_{\text{lab-pres},\ell}
$$

Note that we have only used two variables $x$ and $y$ all over the
construction. 
\end{proof}

\propUndecSynth*

\begin{proof}
    First, let $\phi_{\text{cfl}}$ be the $\mso_\ori$-sentence
    defining the transduction $\tau_{\text{cfl}}$ of Example~\ref{ex:mso}.
    We reduce the $\mso_\ori$ satisfiability
    problem, which is undecidable by Prop.~\ref{prop:undecfo2}, to the
    regular synthesis problem. Let
    $\psi$ be an $\mso_\ori$-sentence of which we want to test
    satisfiability, over alphabets $\Sigma,\Gamma$, which do not
    contain $a,b$. We construct an $\mso_\ori$ sentence $\psi'$ over
    the input alphabet $\Sigma\cup\{a,\#\}$ and output alphabet
    $\Gamma\cup \{b,\#\}$, which defines the
    transduction consisting of the o-graphs
    $(u_1\#u_2,(v_1\#v_2,o))$ such that $(u_1,(v_1,o_1))\models \psi$,
    where $o_1$ is the restriction of $o$ to $v_1$, and
    $(u_2,(v_2,o_2))\models \phi_{\text{cfl}}$, where $o_2$ is the restriction of $o$
    to $v_2$.

    Before explaining how to construct $\psi'$, let us convince the
    reader that $\sem{\psi'}$ is realisable by a regular functional
    transduction iff $\dom(\sem{\psi'}) = \varnothing$ iff $\sem{\psi}=\varnothing$ 
    iff $\semo{\psi}=\varnothing$. Clearly, if $\semo{\psi}=\varnothing$, then
    $\dom(\sem{\psi'}) = \varnothing$ and hence $\sem{\psi'}$ is
    realisable by the regular function with empty domain. Conversely,
    if $\semo{\psi}\neq \varnothing$, there exists $(u_1,(v_1,o_1))\models
    \psi$. Towards a contradiction, suppose that $\sem{\psi'}$ is
    realisable by a regular function $f$. Since regular functions are
    closed under regular domain restriction, the function $f' = f|_L$
    where $L = u_1\#(a+b)^*$ is regular, and hence has regular
    domain. This contradicts the fact that $\dom(f') = \dom(f)\cap L =
    \{ u_1\# a^nb^n\mid n\geq 0\}$ is non-regular.

    Finally, we let $\psi' = \psi_{dom}\wedge \psi_{codom}\wedge \psi^{<\#}
    \wedge \phi_{\text{cfl}}^{>\#}$ where $\psi_{dom}$ expresses that the domain is
    included in $\Sigma^*\#(a+b)^*$, $\psi_{codom}$ that the codomain
    is included in $\Sigma^*\#(a+b)^*$, $\psi^{<\#}$ is just the
    formula $\psi$ where the input (resp. output) quantifiers are
    guarded to range before the unique input (resp. output)
    position labelled $\#$, and symmetrically for
    $\phi_{\text{cfl}}^{>\#}$. 
\end{proof}

\section{Expressiveness, decidability and synthesis for \lorigin}
\subsection{Expressiveness of \lorigin}

\thmMSOexpr*
\begin{proof}

First let us define some unary and binary predicates for the input.
Let $P$ be a subset of $\set{1,\ldots,k}$, we define the formula which states that the copies of $x$ which are used for the output are the ones of $P$:

$$\phi_P(x)=\bigwedge_{c\in P}\phi_{pos}^c(x)\bigwedge_{c\notin P}\neg\phi_{pos}^c(x)$$
Let $c_1, \ldots, c_l$ be a sequence of non-repeating integers smaller than $k$, then we define the formula which says that the order of the copies of $x$ in the output follow the sequence:

$$\phi_{c_1,\dots,c_l}(x)=\phi_{\set{c_1,\dots,c_l}}(x)\bigwedge_{1\leq
  i\leq j\leq l} (\phi_{\leq}^{c_i,c_j}(x,x))$$
Now let $v\in \Gamma^l$, we define the formula specifying the letters of the output positions:

$$\phi_{c_1,\dots,c_l,v}(x)=\phi_{c_1,\dots,c_l}(x)\bigwedge_{i\leq l}\phi_{v(i)}^{c_i}(x)$$
Let $d_1,\ldots, d_m$ be a sequence of non-repeating integers smaller than $k$ and $w\in \Gamma^m$, then we define:

$$\phi_{c_1,\dots,c_l,v,d_1,\dots,d_m,w}(x,y)=\phi_{c_1,\dots,c_l,v}(x)\wedge \phi_{d_1,\dots,d_m,w}(y)$$

Now we define an $\lorigin$-formula $C_i(x)$ which states that $x$ is
exactly the $i$th output position of some input position.

$$C_1(x)= \out(x)\wedge \forall^\out y\  y <_\out x \rightarrow  \{\ori(x) \neq \ori(y)\} $$
And for $i\geq 1$:
$$ \begin{array}{ll}
C_{i+1}(x)=&\exists^\out y\ (y<_\out x \wedge \{ \ori(x) = \ori(y)\} \wedge C_i(y))\\
&\wedge \forall^\out y\ (y<_\out x \wedge \{ \ori(x) = \ori(y)\} \wedge C_i(y))\\
&\  \rightarrow \neg \exists^\out x\ (x<_\out y \wedge \{ \ori(x) = \ori(y)\} \wedge C_i(x))
\end{array}
$$
Note that we have used only two variables $x$ and $y$. 
Now we can define an \lorigin formula which defines the \mso-transduction:

$$\begin{array}{l}
\{\phi_{dom} \} \wedge \foroutput x\ \neg C_{k+1}(x)
        \wedge \forinput x\ \{\phi_\varnothing
        (x)\}\rightarrow (\foroutput y\ \{ \ori(y)\neq x\})\\
\wedge \forall^\out x,y\ \bigwedge_{m,l\leq
    k,c_1,\dots,c_l,v\in\Gamma^l,d_1,\dots,d_m,w\in\Gamma^m,i\leq l,j\leq m} \\

\qquad \qquad\quad \big(C_i(x)\wedge C_j(y)\wedge \\
\qquad \qquad \quad   \{ \phi^{c_i,d_j}_{\leq}\}(\ori(x),\ori(y))\wedge \\
\qquad \qquad \quad \{\phi_{c_1,\dots,c_l,v,d_1,\dots,d_l,w}\}(\ori(x),\ori(y))  \big)\\
\qquad \qquad  \rightarrow \tuple{x\leq_\out y\wedge v(i)(x) \wedge w(j)(y)} \\

\end{array}
$$
\end{proof}

\propIncomp*
\begin{proof}
Firstly, since all \mso-transductions are \lorigin-definable, and as \nmso are defined as \mso-transducers with additional existential parameters, it should be clear that $\exists$\lorigin subsumes \nmso-transductions.

We now turn to the incomparability results.
All witnesses of incomparability are given in Fig.~\ref{fig:expr} that is recalled here.
\setcounter{figure}{1}
\begin{figure}[t]
\begin{minipage}{0.55\linewidth}
\begin{center}
\def\firstcircle{(0,0) circle (1.8cm)}
\def\secondcircle{(55:2cm) circle (1.8cm)}
\def\thirdcircle{(0:2cm) circle (1.8cm)}

\begin{tikzpicture}[scale=0.8]
    \draw[color=blue] \firstcircle node[below left] {};
    \node[color=blue] (tnft) at (-.4,-1.3) {\footnotesize $\text{2NFT}$};
    \draw[color=red] \secondcircle node [above] {};
    \node[color=red] (lt) at (1,3) {\footnotesize $\lorigin$};
    \draw[color=green!50!black]  \thirdcircle node [below right] {};
    \node[color=green!50!black] (nmsot) at (2.2,-1.3) {\footnotesize $\nmso$};
    \node (msot) at (1,0.5) {\tiny REG};
    \draw (1,0.5) circle (0.55cm);

    \draw[rotate=30,dashed] (2.3,-0.2) ellipse (2.7cm and 3.2cm);
    \node (elt) at (3.7, 2.5) {\footnotesize $\elorigin$};

    \node (msor) at (-1,4) {\footnotesize $\mso_\ori$};

    \node[color=red]  (t1) at (1,2.3) { $\tau_1$};
    \node[color=green!50!black] (t2) at (2.9,-0.3) { $\tau_2$};
    \node[color=blue] (t3) at (-1.4,-0.3) { $\tau_3$};
    \node[color=magenta] (t4) at (-0.1,1.3) {
      $\tau_4$};
    \node[color=yellow!55!black] (t5) at (2.1,1.3) { $\tau_5$};
    \node[ color=cyan] (t6) at (1,-0.8) { $\tau_6$};
    \node (t7) at (4,1.6) { $\tau_7$};
    \node (t8) at (0,-0.1) { $\tau_8$};
    \node (t9) at (-1.3,3.3) { $\tau_9$};
\end{tikzpicture}
\end{center}
\end{minipage}
\begin{minipage}{0.45\linewidth}
{\footnotesize
\begin{itemize}[leftmargin=*]
\item[] {\color{red}$\tau_1=\text{shuffle}$}
\item[] {\color{green!50!black}$\tau_2=\{(u,vv)\mid v\preceq u, |v| \text{ is even}\}$}
\item[]{\color{blue}$\tau_3=\{a\}\times (ab)^*$}
\item[]{\color{magenta}$\tau_4=\Sigma^+\times \Gamma^*$}
\item[]{\color{yellow!55!black}$\tau_5=\{(u,vv)\mid v\preceq u\}$}
\item[]{\color{cyan}$\tau_6=\{(u,v)\mid v\preceq u, |v| \text{ is even}\}$}
\item[]$\tau_7=\tau_1\circ \tau_6$
\item[]$\tau_8=\{a\}\times (a+b)^*aa(a+b)^*$
\item[]$\tau_9=\{a^nb^n,(ab)^n\mid\ n>0\}$
\end{itemize}
}
\end{minipage}

\caption{Expressiveness of \lorigin and \elorigin, compared to
  non-deterministic \mso transductions, non-deterministic two-way
  transducers and regular functions.}
\end{figure}
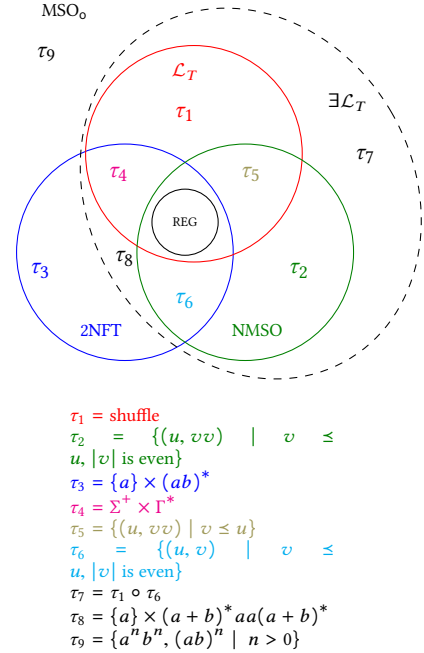
First, $\mso_\ori$ is strictly more expressive than the other formalisms since it is able to specify relations with non regular domain. Indeed a formula for $\tau_9$ simply states that the origin is bijective and label -preserving, that the output domain is $(ab)^*$ and that the input has all $a$s before $b$s.


Now the logic \lorigin and 2NFT are not included in \nmso as they can describe the universal relation $\tau_4=\Sigma^+\times \Gamma^*$, which cannot be defined in \nmso as the number of images of a word $u$ of length $n$ by an \nmso is bounded by the number of possible evaluation of the second order parameters $X_i$, hence bounded by $2^{cn}$, where $c$ is the number of parameters.

\nmso and \lorigin are not included in 2NFT as they can synchronise nondeterministic choices over several readings of the input word, which 2NFT cannot do. This is illustrated by relation $\tau_5$ which first selects a subword of the input and copies it twice. 
An \lorigin formula defining $\tau_5$ states that the input positions producing output produce exactly 2 outputs, that labels are preserved, and that the input order is respected within first copies, as well as the second copies.
An \nmso describing $\tau_5$ simply non deterministically selects a subword via a parameter $X$ and produces $X$ twice, ordering the copies as in the input order.

Finally, \nmso and 2NFT are not included in \lorigin since they are able to specify arbitrary properties of the output that are not definable in $\fo^2$, which is not doable with \lorigin.
The relation $\tau_6$ is easily done with a 2NFT, and can be done in \nmso with a single parameter $X$ which is required to be of even size.
\end{proof}

\section{Domain regularity and synthesis: proofs}
\subsection{Scott Normal Form}
\subsubsection{Non-erasing o-tranductions}
The first step of the transformation consists in ensuring that all the o-graphs satisfying the formula are \emph{non-erasing}, meaning that each input position produces at least one output position.
Formally, an o-graph $(u,(v,o))\in\ProdSG$ is said to be \emph{non-erasing} if
$o$ is a surjective function, and an $\lorigin$-formula $\phi$ is
\emph{non-erasing} if all o-graphs of $\sem{\phi}$ are non-erasing.
Satisfiability of \lorigin is reducible to satisfiability of non-erasing formulas, by adjoining to the output a copy of the input.
\begin{restatable}{proposition}{thmNonErasing}\label{Thm-NonErasing}
For any \lorigin-formula $\phi$ there exists a non-erasing
\lorigin-formula $\phi'$ such that $\dom(\sem{\phi}) =
\dom(\sem{\phi'})$. In particular, $\phi$ is satisfiable if, and only if, $\phi'$ is.
\end{restatable}
\begin{proof}
let $\phi$ be an \lorigin-formula, we want to obtain an \lorigin-formula $\phi^{\text{n.e.}}$ which is non-erasing.
The idea is to extend the output of all o-graphs by a copy of the input word.
We add a new output letter $\sharp$ which will separate the normal output and the copy of the input.
We want to obtain $(u,(v,o))\models\phi$ iff $(u,(v\sharp u,o'))\models\phi^{\text{n.e.}}$ where $o'(i)=o(i)$ if $i\leq |v|$, $o'(1+i+|v|)=i$ if $i\leq |u|$ and $o'(|v|+1)=1$.
From $\phi$, we construct $\phi^{<\sharp}$ where every quantification over the output positions is relativised as being before a position labelled by $\sharp$. Similarly, for $\phi_{\text{id}}$ the identity o-tranduction, we define
$\phi_{\text{id}}^{>\sharp}$ where quantifications over the output are
relativised as appearing after a position labelled by $\sharp$. Adding
the guards can be done while staying in the two-variable fragment. 
Then we define $\phi^{\text{n.e.}}$ to be equal to:
$$\phi^{<\sharp}\wedge \phi_{\text{id}}^{>\sharp} \wedge \exoutput
x\ \sharp(x) \wedge  \text{min}_\inp(\ori(x))\wedge \foroutput y\
\sharp(y)\rightarrow x=y$$
\end{proof}

An \emph{output formula} is an $\lorigin$ formula which is only allowed to quantify over output positions.
The point of considering non-erasing formulas is that one can always transform a non-erasing formula into an equivalent output formula.

\begin{proposition}\label{prop:outfor}
For an $\lorigin$ formula, one can construct an output formula which is equivalent (over non-erasing o-graphs).
\end{proposition}
\begin{proof}
This is shown by constructing inductively an output formula. Atomic formulas are not affected, and boolean connectives are left unchanged.
The remaining case is when $\phi$ is of the form $\exists x\ \psi(x)$.
Then $\phi$ is transformed into $\phi'= \exoutput x\ \psi(x) \vee \exoutput x\ \psi(\ori(x))$. Over non-erasing o-graphs, the two formulas are satisfied by the same models, since any input position is the origin of some output position.
\end{proof}

\subsubsection{Normal Form}
The third step is to normalise any formula in \lorigin into a Scott
normal form (SNF). The procedure to put a formula in SNF is the same
as for $\fo^2$ logics in general (see \cite{GO99} for instance). The point of the SNF is to obtain a formula with additional predicates, which are axiomatised in the formula itself, but with a quantifier depth limited to 2, which lowers the complexity of the formulas.
We prove in our context, along with some preservation property, that
any \lorigin formula can be put in SNF while preserving satisfiability.
Since we aim to get stronger properties than satisfiability, we state
a stronger result, yet the proof is similar.

\begin{restatable}{lemma}{lemScott}\label{lemma-Scott}
For any $\lorigin$-formula $\varphi$ over input alphabet $\Sigma$ and output alphabet $\Gamma$, one can construct an $\lorigin$-formula $\phi$ over $\Sigma$ and $\Gamma\times\Gamma'$ such
that:
\begin{itemize}[leftmargin=*]
\item $\Gamma'$ is a finite alphabet,
\item up to projection on $\Gamma$, $\phi$ and $\varphi$ have the same models,
\item $\phi$ is of the form $\forall^\out x\forall^\out y\
  \psi(x,y)\wedge\bigwedge_{i=1}^m\forall^\out x\exists^\out y\ \psi_i(x,y)$
  where the formulas $\psi$ and $\psi_i$, $i=1,\dots,m$, are
  quantifier free.
\end{itemize}
\end{restatable}

\begin{proof}

The proof is similar to~\cite{SZ12}. We first assume without loss of generality that $\varphi$ is in negation normal form.
We now construct the formula $\phi$ iteratively.
At each iteration, we get formulas $\theta_i$ and $\phi_i$ where $\varphi$ is equivalent to $\theta_i\wedge \phi_i$, $\theta_i$ is in correct form, and $\phi_i$ has a number of quantifiers reduced by $i$ compared to $\varphi$, while using some additional unary predicates $P_1,\ldots,P_i$. 
At first let $\theta_0=\top$ and $\phi_0=\varphi$. 
Then, at each step, consider a subformula $\xi_i(x)$ of $\phi_{i-1}$ with a single quantifier.
Then $\xi_i(x)$ is either $\exists y\ \rho_i(x,y)$ or $\forall y\ \rho_i(x,y)$ where $\rho_i$ a quantifier free formula.
In the first case, we set $\theta_i=\theta_{i-1}\wedge \forall x\exists y\ (P_i(x)\to\rho_i(x,y))$ and
$\phi_i$ is obtained by replacing $\exists y\ \rho_i(x,y)$ by $P_i(x)$.\\
In the second case, we set  $\theta_i=\theta_{i-1}\wedge \forall x\forall y\ (P_i(x)\to\rho_i(x,y))$ and
$\phi_i$ is obtained by replacing $\forall y\ \rho_i(x,y)$ by $P_i(x)$.

This process ends as at each step the number of quantifiers of $\phi_i$ decreases.
In the end, we get $\phi_k$ which is quantifier free and thus equivalent to $\forall x\forall y\ \phi_k$.
By combining all the double $\forall$ conjuncts into one formula $\psi$, we finally set $\phi=\theta_k\wedge \forall x\forall y\ \psi$ which is in the required form.
The size of $\phi$ is linear in the size of the negative normal form of $\varphi$.
Finally, the unary predicates $P_i$ are added to the alphabet to be treated as letters.
This is done by replacing the alphabet $\Gamma$ by $\Gamma\times\Gamma'$ , where $\Gamma'=2^{P_1,\ldots,P_k}$, and replacing in the formula the predicates $P_i(x)$ by the conjunction of letter predicates $\bigvee\limits_{(\gamma,R)\mid P_i\in R}(\gamma,R)(x)$.

We need now to prove the first statement regarding domains.
We prove this by induction on the formulas $\theta_i\wedge \phi_i$.
Assume that the o-graph $w_i$ is a model for $\theta_i\wedge \phi_i$, we construct
$w_{i+1}$ a model for $\theta_{i+1}\wedge\phi_{i+1}$ by adding truth values for the predicate $P_{i+1}$ by setting 
$$P_{i+1}=\{p\mid p\text{ is a position of $w$ and }(w,p)\models \xi_i(x)\}.$$
Conversely, if $(w_i,P_{i+1})$ is a model for $\theta_{i+1}\wedge\phi_{i+1}$, then for any position $p$ of $w_i$ such that $(w_{i+1},p)\models P_{i+1}(x)$, we also have $(w_i,p) \models \xi_{i+1}(x)$ since $P_{i+1}$ does not appear in $\xi_{i+1}$. And since $\varphi$ is in negative normal form, $\xi_{i+1}$ only appears positively and thus $w_i\models \varphi_i$. We conclude by noting that if 
$(w_i,P_{i+1})\models \theta_{i+1}$ then $w_i\models \theta_{i+1}$.
Notice that the number of predicates added is equal to the number of
quantifications in $\varphi$ and hence is linear. However, since they
are not mutually exclusive, thisleads to an exponential blow-up of the
alphabet $\Gamma'$.

Finally, we apply Proposition~\ref{prop:outfor} to remove
quantifications over input positions, a construction which preserves the
normal form. 
\end{proof}

\subsubsection{Sets of constraints}
In the spirit of \cite{SZ12}, we introduce another formalism, called
system of constraints, which is equivalent to SNF $\lorigin$. \emph{Constraints}  are built over label
predicates, and some input
and output predicates. Given an o-graph $(u,(v,o))$, a label predicate $\gamma\in\Gamma$ is
satisfied by an output position $p$ of $v$ if $p$ is labelled $\gamma$. 
Output predicates are restricted to 
directions $\uparrow,\downarrow$, which are satisfied by a pair of output positions 
$(p,p')$ if, respectively, $p < p'$ and $p' <
p$. 
Input predicates are any \mso-definable binary predicate over 
the input using the labels $\Sigma$ and the input order $\leqinput$. A pair of output positions $(i,j)$
satisfies an input predicate $\psi(x,y)\in \mso[\Sigma,\leq]$ if their origin satisfy it, \ie
$u\models \psi(\ori(i),\ori(j))$.

An \emph{existential constraint} is a pair $(\gamma,E)$ where $\gamma\in \Gamma$ and $E$ is a set of tuples $(\gamma',d,\psi)$ such that $d\in\{\uparrow,\downarrow\}$ is an output direction and $\psi$ is an input predicate.
Given an o-graph $(u,(v,o))$, an output position $p$ of $v$
satisfies an $\exists$-constraint $(\gamma,E)$ if whenever $p$ is
labelled $\gamma$, there exist a triple $(\gamma',d,\psi)\in E$ and a $\gamma'$-labelled position $q$ of $v$
 such that  $(p,q)$ satisfies $d$ and $\psi$. In the latter case, we call
 $q$ a \emph{valid witness} of $p$ for $(\gamma,E)$.

A \emph{universal constraint} is a tuple $(\gamma,\gamma',d,\psi)$ where $d$ is an output direction and $\psi$ an input predicate.
A pair $(p,q)$ of positions of $v$ satisfies a $\forall$-constraint
$(\gamma,\gamma',d,\psi)$ if it is not the case that $p$
is labelled by $\gamma$, $q$ by $\gamma'$, and $(p,q)$ satisfy $d$ and $\psi$.
A universal constraint can be thought of as a forbidden pattern over pairs of points.

An instance of the \emph{\mso\ constraint problem} (\lpp) is a pair $C=(C_\exists,C_\forall)$ of sets of existential and universal constraints respectively.
A non-erasing o-graph $w=(u,(v,o))$ is a solution of (or model for) $C$, denoted $w\models C$, if every output position of $v$ satisfies
all constraints in $C_\exists$ and every pair of output positions satisfy all constraints in $C_\forall$.

\begin{restatable}{proposition}{propMSOLPP}\label{prop-MSOLPP}
From any output $\lorigin$-sentence $\phi$ in SNF over $\Sigma$ and $\Gamma$, 
one can construct an instance $C$ of \lpp over $\Sigma$ and $\Gamma$
such that for any non-erasing o-graph $w\in \ProdSG$,  $w\models
\phi$ iff $w\models C$.
\end{restatable}

\begin{proof}

Let  $\phi=\foroutput x\foroutput y\ \varphi(x,y)\wedge\bigwedge\limits_{i=1}^n\foroutput x \exoutput y\ \varphi_i(x,y)$ with the binary input predicates $(\alpha_i)_{i=1}^k$. We treat $0$-ary and unary predicates as binary predicates.
Now given $x$ and $y$ quantifying positions of the output, an \emph{atomic type} for $x$ and $y$ gives
truth value for the predicates  $(\alpha_i)_{i=1}^k$ (evaluated over their origin for the input predicates).
Formally, it is composed of labels for $x$ and $y$, an output direction $x\sim y$ for $\sim\in\{=,\leftarrow,\rightarrow\}$ and truth values for the binary formulas $\alpha_i$. Then a couple of output positions $(p,q)$ is of type $t$ if they satisfy exactly the true properties of $t$ when $x$ and $y$ are evaluated as $p$ and $q$ respectively, and the predicates $\alpha_i$ are evaluated on $\ori(x)$ and $\ori(y)$.
Note that any atomic type can be described by a universal constraint using boolean combinations of the predicates $\alpha_i$.
Note also that any model of $\phi$ has to satisfy the universal part $\foroutput x\foroutput y\ \varphi(x,y)$.
Hence we want to weed out all atomic types that do not satisfy it.
Then the set of universal constraints $C_\forall$ is set as all forbidden types, \ie the atomic types that do not satisfy $\varphi(x,y)$.
Then if $w=(u,(v,o))$ is an o-graph which satisfies $\phi$, any pair of positions of $v$ satisfy $\varphi(x,y)$ if and only if they satisfy every constraint of $C_\forall$.

We now turn to the formulas $\foroutput x \exoutput y\ \varphi_i(x,y)$. By doing an extensive case study over all atomic types for $x$ and $y$, and then factorising for each label $\gamma$ of $\Gamma$, we can rewrite the formulas as 
$$\foroutput x \bigwedge\limits_{j=1}^k ( \gamma_j(x) \to \exoutput y \bigvee\limits_{\ell=1}^{m}t_{j,\ell})$$
where $t_{j,\ell}$ are atomic types. 
We conclude depending on the nature of the direction $d_{j,\ell}$ of $t_{j,\ell}$. If $d_{j,\ell}$ is $x=y$, then if $t_{j,\ell}$ is compatible with $\gamma_j$ the conjunct $\gamma_j(x) \to \exoutput y\ t_{j,\ell}$ is either a tautology and the whole conjunct is trivially satisfied, or it cannot be satisfied and  $t_{j,\ell}$ is removed from the disjunction.
The remaining elements of the disjunction can be combined in a set $E$ to form an existential constraint with $\gamma_j$.
Now if $w\models \foroutput x \bigwedge\limits_{j=1}^k ( \gamma_j(x) \to \exoutput y \bigvee\limits_{\ell=1}^{m}t_{j,\ell})$, then for every position $p$ of $v$, if $p$ is labelled by $\gamma$ then there exists a position $q$ such that $(p,q)$ is of one of the types $t_{j,\ell}$ and thus $q$ is a valid witness for $p$.
Conversely, the fact that any position $p$ has a valid witness means that for any output position labelled by $\gamma$, there is an other position corresponding to its witness which is a valid candidate for $y$, and thus $w$ satisfies $\foroutput x\bigwedge\limits_{j=1}^k ( \gamma_j(x) \to \exoutput y \bigvee\limits_{\ell=1}^{m}t_{j,\ell})$.

This gives an instance $C=(C_\exists,C_\forall)$ of constraints such that for any non-erasing o-graph $w$, $w\models C$ if, and only if, $w\models \phi$.
\end{proof}

\subsection{The profile abstraction}

We define here the important notion of profile,
which is a bounded abstraction, given an o-graph, of 
an input position, the output positions it produces, and its context
within the o-graph (other output and input positions). An o-graph
can then be abstracted by a sequence of profiles.

We define the notions of validity, with respect to an \lpp instance
$C$, and of maximal consistency, for sequences of profiles,
which respectively talk about the satisfaction of constraints by the
profiles of the sequence, and the consistency between consecutive
profiles (the information stored in consecutive profiles is correct and consistent).

\subsubsection{Automata for binary predicates}
In the following we will be using automata for binary predicates. They will serve as \mso types that we can easily manipulate.
It is well-known (see \eg \cite{NPTT05}) that any binary $\mso[\Sigma,\leq]$-predicate
$\psi(x,y)$ over $\Sigma$-labelled words, can be equivalently
defined by a non-deterministic finite automaton (called here a predicate
automaton)
$\mathcal{A}_\psi=(Q_\psi,\Sigma,I_\psi,\Delta_\psi,F_\psi)$ equipped
with a set $\Select_\psi\subseteq Q_\psi^2$ of \emph{selecting pairs} 
with the following semantics: for any word $u\in\Sigma^*$ and any pair of positions
$(i,j)$ of $u$, we have $u\models \psi(i,j)$ if, and only if, there
exists an accepting run $\pi$ of $\mathcal{A}_\psi$ and a pair $(p,q)\in\Select_\psi$
such that $\pi$ is in state $p$ before reading $u(i)$ and in state $q$
before reading $u(j)$.

\begin{example}
\label{Ex-predicate-automaton}
Let us consider as an example the binary between predicate
$Bet_\sigma(x,y)=\exists z\ \sigma(z)\wedge(x<z)\wedge (z<y)$, which
cannot be expressed using only two variables. The automaton for this
predicate is depicted below and its unique selecting pair is $(q_x,q_y)$.
\begin{center}
\begin{tikzpicture}[scale=.9,every node/.style={scale=0.9}]

\node[circle,draw] (qx) at (0,0) {\footnotesize{$q_x$}};
\node[circle,draw] (qs) at (2,0) {\footnotesize{$q_\sigma$}};
\node[circle,draw] (qy) at (4,0) {\footnotesize{$q_y$}};
\node[circle,draw] (qf) at (6,0) {\footnotesize{$q_f$}};

\draw (qx) edge[->,loop below,>=stealth]  (qx);
\draw (qs) edge[->,loop below,>=stealth]  (qs);
\draw (qy) edge[->,loop below,>=stealth]  (qy);
\draw (qf) edge[->,loop below,>=stealth]  (qf);

\node at (0.4,-.8) {\footnotesize{$\Sigma$}};
\node at (2.4,-.8) {\footnotesize{$\Sigma$}};
\node at (4.4,-.8) {\footnotesize{$\Sigma$}};
\node at (6.4,-.8) {\footnotesize{$\Sigma$}};

\draw[->, >=stealth] (-1,0) --(qx);
\draw[<-, >=stealth] (7,0) --(qf);

\draw (qx) edge[->, >=stealth] node[midway,above] {\footnotesize{$\Sigma$}} (qs);
\draw (qs) edge[->, >=stealth] node[midway,above] {\footnotesize{$\sigma$}} (qy);
\draw (qy) edge[->, >=stealth] node[midway,above] {\footnotesize{$\Sigma$}} (qf);

\end{tikzpicture}

\end{center}
\end{example}

\subsubsection{Profiles}

Let $C$ be an instance of \lpp over $\Sigma$ and $\Gamma$, and 
$\Psi$ the set of
$\mso$-predicates occurring in $C$. For all $\psi{\in} \Psi$, we
let $\mathcal{A}_\psi$ with set of states $Q_\psi$ and set of
selecting pairs $\Select_\psi$ be the predicate automaton
for $\psi$.  Let $S_\Psi= \biguplus_{\psi\in\Psi} Q_\psi$. 

The main ingredient of profiles is a sequence of \emph{clauses}, where
a clause is an element of the set
$\clauses=\Gamma\times(\{\cdot\}\cup\mathcal{P}(S_\Psi\times
S_\Psi)\times\set{\leftarrow,\rightarrow})$. 
Clauses of the form
$(\gamma, \cdot)$ are called \emph{local clauses} and clauses of the
form
$(\gamma,R,v)$ are called \emph{consistency clauses}. Intuitively, in a
o-graph, if
the profile of an input position $i$ contains a local clause
$(\gamma,\cdot)$, this clause describes an output position produced by
$i$ (its origin is $i$) and labelled by $\gamma$. If the profile of
$i$ contains a clause $(\gamma,R,v)$, it describes an output position
whose origin $j$ appears in the direction $v$ with respect to position
$i$ (\ie if $v=\leftarrow$ then $j<i$, $i>j$ otherwise), is labelled $\gamma$ and such that for any pair $(p,q)$ of
$R$, there exists an accepting run of $A_\psi$ which reads position $i$ in state $p$ and
position $j$ in state $q$. A clause $A$ is \emph{compatible} is with a
set of states $S\subseteq S_\Psi$ if whenever $A$ is a consistency
clause $(\gamma,R,v)$, then $\dom(R) = \{ p\mid \exists (p,q)\in
R\}\subseteq S$.

A \emph{$C$-profile} (or just profile) is a tuple $\profile=(\sigma,S,A_1\dots A_n)$ where
$\sigma\in\Sigma$ is an input label, $S\subseteq S_\Psi$ and $A_1\ldots
A_n$ is a sequence of clauses from $\clauses$
such that any clause $A_i$ is compatible with $S$ and 
appears at most twice in the sequence, for
all $i=1,\dots,n$. 
By definition, the number of profiles is bounded by $N=|\Sigma|\cdot 2^{|\cS|}(|\Gamma|\cdot( 2(2^{|\cS|^2+1}+1))!$.

\subsubsection{Profile of an input position}\label{subsec:inprof}

To define the profile of an input position $k$ of an o-graph
$w=(u,(v,o))$, with respect to some \lpp instance $C=(C_\exists,C_\forall)$, we first
define the notion of full profile of that position, which keeps
complete (and unbounded) information about the whole o-graph. The profile will be
then a bounded abstraction of the full profile. The \emph{full profile} of $k$ is
defined as the tuple $\lambda^f = (\sigma,S,B_1\ldots B_{|v|})$, where 
each of its elements are defined as follows.
The letter $\sigma$ is the $k$th letter of $u$ and $S$ is the set of
states reached by all the accepting runs of the predicate automaton $A_\Psi$ on $u$ after
reading the prefix $u_1\dots u_{k-1}$.

Let now $j\leq |v|$ be an
output position with origin $k'$. The element $B_j$ is a clause
generated by the output position $j$, defined as follows. If $k'=k$,
then $B_j = (\gamma,\cdot)$. If $k<k'$ (resp. $k>k'$), then we define
$B_j$ to be the consistency clause $(\gamma,R,\rightarrow)$ (resp. $(\gamma,R,\leftarrow)$) where $R$ is the set of all pairs $(p,q)$ from $S_\Psi$ such that
there exists an accepting run on $u$ that reaches $p$
after reading $u_1\dots u_{k-1}$ and $q$ after reading $u_1\dots
u_{k'-1}$ (hence $p\in S$). Therefore
in $R$, the first component always refers to the state at the current
position $k$, and the second component to the state of described
position $k'$. The direction indicate whether the
described position $k'$ is to the right or the left of $k$. Hence, if
$(p,q)\in R$ and the direction is $\rightarrow$, it means that
the state $p$ eventually
reaches $q$ on the right, and if the direction is $\leftarrow$, that the state $q$ was
visited before $p$.

The \emph{profile} $\lambda$ of position
$k$ is obtained from $\lambda^f$ by keeping in $B_1\dots B_{|v|}$ the
outermost occurrences of each clause. Formally, it is $\lambda =
(\sigma,S,\alpha(B_1\dots B_{|v|}))$  where
$\alpha:\clauses^*\to\clauses^*$ erases in a sequence of clauses, all
but the left- and right-most occurrences of each clause of the
sequence.
%
In the following, we show that valid o-graphs give valid sequences of profiles, and that conversely we can construct valid o-graphs from valid sequences.
An example of profile sequence is given in
Fig.~\ref{fig:abstraction}, in which clauses $(\gamma, R,
\rightarrow)$ are denoted $\gamma\overrightarrow{R}$ (and similarly for
clauses $(\gamma, R, \leftarrow)$.



\subsubsection{Profile validity}

As we have seen, an o-graph can be abstracted by the sequence of its
profiles. We aim at defining conditions on profiles and profile
sequences $s$ under which from such a sequence of
profiles we can reconstruct an o-graph which is a model of an \lpp
instance $C$. The notion of profile validity takes care of the
property of being a model, but not any sequence of profiles will be
the profile sequence of an o-graph in general. The notion of profile
consistency, defined in the next section, is introduced to ensure this
property. 
First, we start with the notion of profile validity with respect to an
\lpp-instance $C$.

\begin{definition}[Profile validity] Let $C$ be an \lpp-instance and 
$\profile=(\sigma,S,A_1\ldots A_n)$ a profile. It satisfies 
an existential constraint  $c=(\gamma,E)$
if for every $i$ such that $A_i=(\gamma,\cdot)$ there exists 
a tuple $(\gamma',d,\psi)$ of $E$, $j\leq n$ such that $i$ and $j$
respect the direction $d$ (\ie $i<j$ if $d=\uparrow$ and $j<i$ otherwise), 
and  either
there exists $v\in\{\leftarrow,\rightarrow\}$ such that
$A_j=(\gamma',R,v)$ and $R\cap
\Select_\psi\neq\emptyset$, 
or $A_j=(\gamma',\cdot)$ and there exists $p\in S$ such that $(p,p)\in \Select_\psi$.

The profile $\lambda$ satisfies a universal constraint $(\gamma,\gamma',d,\psi)$ if 
there do not exist $i$ and $j$ such that $i$ and $j$ respects
direction $d$, $A_i=(\gamma,\cdot)$
and $A_j=(\gamma',R,v)$ such that $R\cap \Select_\psi\neq\emptyset$ 
or $A_j=(\gamma',\cdot)$ and there exists $p\in S$ such that $(p,p)\in \Select_\psi$.

Given an instance $C$ of \lpp, a profile is \emph{$C$-valid} (or just
valid) if it satisfies every constraint. A sequence of valid profiles
is also called \emph{valid} sequence.
\end{definition}



Intuitively, let us take the case of existential constraints. In the
o-graph we aim to reconstruct from a sequence of profiles, the
clause $A_i$ will correspond to an output position $p$ with origin $i$,
and the clause $A_j$ will refer to an output position $p'$ with origin
$i'$ produced before
or after $i$ (depending on $v$) which is a witness $p$ and the
constraint $(\gamma',d,\psi)$, because the fact that $R\cap
\Select_\psi\neq \emptyset$ indicates that there is an accepting run
of $A_\psi$ on the input word which selects the pair $(i,i')$,
\ie $(i,i')$ satisfies the MSO-condition $\psi$. 


\subsection{Properties of profile sequences}
\subsubsection{Profile consistency}

Consistency is first defined between two consecutive profiles,
ensuring consistent run information between the clauses of these two
consecutive profiles.%
Then the consistency of every pair of successive profiles in a given
sequence ensures a global consistency allowing to reconstruct full
runs of $A_\Psi$ on the whole input.

We now need a central notion, that of successor (and predecessor) of
clauses. Informally, a clause $A'$ is a successor of $A$ if there is a
o-graph and an input position $k$ such that in the full profiles
$\profile^f,\profile^f_{k+1}$ of positions $k$ and $k+1$ respectively, there
exists $i$ such that $A$ is the $i$th clause of $\profile^f_k$ and $A'$ is
the $i$th clause of $\profile^f{k+1}$. This is just an intuition, and as
a matter of fact a consequence of the formal definition, which is more
constructive and not dependent on o-graphs.


Let us now give here the formal definitions concerning consistency of
profiles. To do so, we first define a successor relation between clauses, parameterized by 
two sets $S,S'\subseteq S_\Psi$ (we remind that $S_\Psi$ is disjoint union of the set of
states $Q_\psi$ of all predicate automata). Informally, since a clause occurring
at input position $k$ stores information about some output position
(whose origin is either $k$, $k'<k$ or $k'>k$), the successor relation
tells us how this information is updated at input position $k+1$,
depending on these cases. We will use the following notation
$s'\in s\cdot \sigma$ whenever there exists a transition of $A_\Psi$
from state $s$ to $s'$ on $\sigma$. We will also say that a pair of
binary relations $(R,R')$ on $S_\Psi$ is \emph{compatible} with $\sigma$
if for all $(p,q)\in R$, there exists $(p',q)\in R'$ such that $p'\in
p\cdot \sigma$ and conversely, for all $(p',q)\in R'$, there exists
$(p,q)\in R$ such that $p'\in p\cdot \sigma$. Note that if $(R,R')$
and $(R,R'')$ are compatible with $\sigma$, then $(R,R'\cup R'')$ is
as well compatible with $\sigma$. Hence, given $R$, there exists
maximal relation $R_m$ (for inclusion) such that $(R,R_m)$ is
compatible with $\sigma$. Symmetrically, given $R'$, there exists a
maximal relation $R_m$ such that $(R_m, R')$ is compatible with
$\sigma$.

\begin{definition}[Successors of clauses]
Let $S,S'\subseteq S_\Psi$, $\sigma\in \Sigma$ and $A$ a clause.
The successors of $A$ with respect to $S,S'$ and
$\sigma$ are clauses $A'$ defined as follows:
\begin{enumerate}
  \item if $A =(\gamma,\cdot)$, there is a unique successor $A' =
    (\gamma, \{ (p',p)\in S'\times S \mid p'\in p\cdot \sigma\},
    \leftarrow)$,

  \item if $A = (\gamma, R, \leftarrow)$, there is a unique successor $A' = (\gamma, R',
    \leftarrow)$ such that $R'$ is the maximal relation such that
    $(R,R')$ is compatible with $\sigma$ and $\dom(R') \subseteq S'$. 

  \item if $A = (\gamma, R, \rightarrow)$, there are several possible
    successors $A'$:
    \begin{enumerate}
      \item either $A' = (\gamma,\cdot)$ under the condition that $R=\set{(p,q)\in
          S\times S'|\ q\in p\cdot \sigma}$
      \item or $A' = (\gamma,R',\rightarrow)$ where
        $\dom(R')\subseteq S'$ and $R$ is the maximal relation such
        that $(R,R')$ is compatible
        with $\sigma$. 
    \end{enumerate}
\end{enumerate}


By extension, given two profiles
$\profile=(\sigma,S,A_1\ldots A_n)$ and
$\profile'=(\sigma',S',A'_1\ldots A'_{n})$, a clause $A_i$ and a
clause $A'_j$, we say that $A'_j$ is a successor of $A_j$ with respect
to $\lambda,\lambda'$ if it is a successor w.r.t. $S,S'$ and $\sigma$. 
\end{definition}

As a remark, we notice that the successor relation is not necessarily
functional in the case where $A = (\gamma, R, \rightarrow)$. This is
consistent with the following observation. Given the full profile of
an input position $k$, two occurrences of a clause $A = (\gamma, R,
\rightarrow)$ may describe two output positions $j_1,j_2$ whose origins $k_1,k_2$ are to the
right of the current position $k$. If for instance $k_1=k+1$ and
$k_2>k_1$, then in the full profile of $k_1$, output position $j_1$ is
described by a clause of the form $(\gamma',\cdot)$ while output
position $j_2$ be a clause of the form $(\gamma'', R, \rightarrow)$,
both clauses being successors of $A$. We also notice that in the
profile of $k$ (the abstraction of the full profile), one occurrence, or both, of the clause $A$ may have
been deleted.

Similarly, we define the predecessors of a clause $A$ with respect to
$S,S'$ and $\sigma$ as the set of clauses $B$ such that $A$ is a successor of $B$
with respect to $S,S'$ and $\sigma$. We prove the following useful proposition:
\begin{proposition}\label{prop:uniquepred}
Let $S,S'\subseteq S_\Psi$, $\sigma\in \Sigma$ and $A$ be a clause of
type $(\gamma, R, \rightarrow)$ or $(\gamma,\cdot)$. Then $A$ has a
unique predecessor with respect to $S,S'$ and $\sigma$.

Symmetrically, if $A$ is of type $(\gamma,R,\leftarrow)$ or
$(\gamma,\cdot)$, then $A$ has a unique successor with respect to
$S,S'$ and $\sigma$
\end{proposition}
\begin{proof}
    The second statement is direct by definition. For the first
    statement, if $A$ is of the form $(\gamma,\cdot)$, then by
    definition, the unique predecessor of $A$ is $(\gamma,R, \rightarrow)$
    where $R = \set{(p,q)\in S\times S'|\ q\in p\cdot \sigma}$. If $A$ is of the form
        $(\gamma, R, \rightarrow)$, then suppose there are two
        different predecessors. They are necessarily of the form 
        $(\gamma, R_1,\rightarrow)$ and $(\gamma, R_2,\rightarrow)$ where
        $R_1\neq R_2$. Then, neither $R_1$ nor $R_2$ are maximal
        relations such that $(R_i, R)$ is compatible with $\sigma$ (it
        suffices to take $R_1\cup R_2$, contradicting the definition
        of successor 3.a. 
\end{proof}




The notion of consistency is first defined between two profiles, 
then extended to sequence of profiles. Between two profiles
$\lambda_1,\lambda_2$, consistency is defined as structural properties of
a bipartite graph $G_{\lambda_1,\lambda_2}$ which we now define. The
vertices of $G_{\lambda_1,\lambda_2}$ are clause occurrences in
$\lambda_1$ and $\lambda_2$, labelled by clauses, and the set of edges is a
subset of the successor relation between those
occurrences. Formally, let $s_1 = A_1\dots
A_n$ and $s_2 = B_1\dots B_m$ be the sequence of clauses 
of $\lambda_1,\lambda_2$ respectively. We let $G_{\lambda_1,\lambda_2}
= (V,E,\ell : V\rightarrow \clauses)$ where $V = \{1\}\times
\{1,\dots,n\}\cup \{2\}\times \{1,\dots,m\}$, where for all $i$, 
$\ell(1,i) = A_i$ and $\ell(2,i) = B_i$. We say that $i$ is the
smallest occurrence of a clause $A$ in $s_1$ if $i = \text{min} \{
j\mid A = A_i\}$ (and similarly of $s_2$, and the notion of greatest
occurrence).  The set of edges is defined
as follows. There is an edge from $(1,i)$ to $(2,j)$ if one of the
following condition holds:
\begin{enumerate}
  \item $A_i$ is of the form $(\gamma, R, \rightarrow)$, $i$ is the
    smallest (resp. greatest) occurence of $A_i$ in $s_1$, and $j$ is
    the smallest index (resp. greatest index) such that $B_j$ is a
    successor of     $A_i$ w.r.t.  $\lambda_1,\lambda_2$. 
  \item $B_j$ is of the form $(\gamma, R, \leftarrow)$, $j$ is the
    smallest (resp. greatest) occurence of $B_j$ in $s_2$, and $i$ is
    the smallest index (resp. greatest index) such that $A_i$ is a
    predecessor of $B_i$ w.r.t. $\lambda_1,\lambda_2$. 
\end{enumerate}

We can now define consistency:

\begin{definition}[Profile consistency]
A profile $\profile_1 = (\sigma_1,S_1,A_1\dots A_n)$ is \emph{consistent}
with a profile $\profile_2 = (\sigma_2,S_2,B_1\dots B_m)$
if the following three conditions hold:
\begin{enumerate}
\item for any state $s_1\in S_1$, there is a state $s_2\in S_2$ such that
$s_2\in s_1\cdot \sigma_1$ and conversely for all $s_2\in S_2$, there exists
$s_1\in S_1$ such that $s_2\in s_1\cdot \sigma_1$,
\item  for all clause $A_i$, there exists $B_j$ such that $B_j$ is a
  successor of $A_i$, and conversely for all clause $B_j$, there
  exists $A_i$ such that $A_i$ is a predecessor of $B_j$,

\item the graph $G_{\lambda_1,\lambda_2}$ does not contain the
  following patterns:
  \begin{enumerate}
    \item a vertex with two outgoing edges
    \item a vertex with two ingoing edges
    \item a crossing, \ie two edges $((1,i_1),(2,j_1))$ and
      $((1,i_2),(2,j_2))$ such that $i_1<i_2$ and $j_2<j_1$.
  \end{enumerate}
\end{enumerate}
\end{definition}

A profile $\profile=(\sigma,S,A_1\ldots A_n)$ is \emph{initial} if all
states of $S$ are initial states and there is no consistency clause
pointing to the left (\ie clause $(\gamma,R,\leftarrow)$).
It is \emph{final} if for all states $s$ of $S$, we can reach a final
state by reading $\sigma$, and there is no consistency clause pointing
to the right.
A sequence of profiles $\profile_1\ldots \profile_n$ is
\emph{consistent} if $\profile_1$ is initial, $\profile_n$ is final,
and for all $i<n$, $\profile_i$ is consistent with $\profile_{i+1}$.

We generalise the notion of graph associated with two profiles, to
sequences of profiles $s = \lambda_1\dots \lambda_n$. It is the
disjoint union of all the graphs $G_{\lambda_i,\lambda_{i+1}}$ where
the second component of $G_{\lambda_i,\lambda_{i+1}}$ is glued to
the first component of $G_{\lambda_{i+1},\lambda_{i+2}}$. Formally,
an occurrence of a clause $A$ in $s$ is a pair $(i,j)$ such that $A$
is the $j$th clause of $\lambda_i$. Then, we define $G_s = (V,E,\ell)$
where $V$ is the set of all clause occurrences, $\ell(i,j)$ is the
$j$th clause of $\lambda_i$, and for all $1\leq i<n$, there is an edge 
from $(i,j)$ to $(i+1,j')$ in $G_s$ iff there is an edge from $(1,j)$
to $(2,j')$ in $G_{\lambda_i,\lambda_{i+1}}$.  The following
lemma gives some structural properties of this DAG:
\begin{proposition}\label{prop-consist}
    For any consistent sequence of profiles $s$, the following hold true:
\begin{enumerate}
  \item $G_s$ is a union of disjoint directed paths,
  \item each maximal directed path $\pi$ of $G_s$ is of the form 
    $$
    \pi = (\gamma,R_1, \rightarrow)\dots (\gamma,
    R_i,\rightarrow)(\gamma,\cdot) (\gamma,
    R_{i+1}, \leftarrow)\dots (\gamma, R_{i+k},
    \leftarrow)
    $$
    where $i,k\geq 0$ and $i+k<n$, $\gamma\in \Gamma$ and the $R_j$
    are binary relations on $S_\Psi$,
  \item there is bijection between local clause occurrences of $s$ and
    maximal paths of $G_s$. Therefore, we identify a local clause
    occurrence $(i,j)$ with its maximal directed path, which we denote
    by $\pi_{i,j}$. 
\end{enumerate}
\end{proposition}
\begin{proof}
    $(1)$ It is a direct consequence of conditions 3a. and 3b. in the
    definition of consistency.

    $(2)$ First, any path which contains a local clause contains a
    unique local clause and is
    necessarily of this form. It is a direct consequence of the definition of the
    successor relation. Indeed, the successor of a clause of the form
    $(\gamma, R,\rightarrow)$ are clauses of the form
    $(\gamma,R',\rightarrow)$ or $(\gamma, \cdot)$. The successor of
    clauses of the form $(\gamma,\cdot)$ is necessarily of the form
    $(\gamma,R,\leftarrow)$ and the successors of the latter are
    necessarily of the form $(\gamma, R', \leftarrow)$.

    It remains to prove that any maximal path contains a local
    clause. Suppose that it does not contain any local clause. Then,
    by similar arguments as before, we can show that it contains only
    clauses of the form $(\gamma,R,\rightarrow)$ or only clauses of
    the form $(\gamma, R, \leftarrow)$. Let us assume the first case
    (the other one being symmetric). Suppose that the last two
    vertices of this path are $(i,j_1)$ and $(i+1,j_2)$. If $i+1=n$,
    then we get a contradiction since $\lambda_n$ would not be final
    (which contradicts the fact that $s$ is supposed to be
    consistent). Suppose that $i+1<n$, let $A = \ell(i,j_1)$ and
    $B=\ell(i+1,j_2)$ the clauses associated with these vertices. We
    know that $B$ is a successor of $A$
    w.r.t. $\lambda_i,\lambda_{i+1}$. By definition of
    $G_{\lambda_i,\lambda_{i+1}}$, it is even an extremal successor of
    $A$. Therefore, $(i,j_2)$ is either the smallest occurrence of $B$
    in $\lambda_{i+1}$ or the greatest one. By definition of
    consistency (second condition), we know that $B$ has necessarily a
    successor $C$ in $\lambda_{i+2}$, and by definition of
    $G_{\lambda_{i+1},\lambda_{i+2}}$, there must exist an edge from
    $(i+1,j_2)$ to some occurrence of $C$ in $\lambda_{i+2}$,
    contradicting the fact that the considered path is maximal.

    $(3)$ It is a direct consequence of $(1)$ and $(2)$.
\end{proof}

We now define the notion of maximal consistency for a sequence
of profiles. Intuitively a consistent profile sequence is maximal if one cannot add
states in the clauses without making it inconsistent. 

\begin{definition}[Maximality]
    A consistent profile sequence $s$, with sequence of
    $\Sigma$-components $u = \sigma_1\dots\sigma_n$ and sequence of $S$-components
    $S_1\dots S_n$, is \emph{maximal} if
    \begin{enumerate}

      \item for all $1\leq i\leq n$, $S_i$ is the set of \emph{all} states
        $q$ such that there exists an accepting run $q_1\dots q_{n+1}$
        of $A_\Psi$ on $u$ such that $q_i = q$, and,
      \item 
        for all local
    clause occurrence $(i,j)$, for all vertex $(i',j')$ labelled
    $(\gamma, R, v)$ on $\pi_{i,j}$, $R$ is the set
          of \emph{all} pairs $(p,q)$ such that there exists
          an accepting run $q_1\dots q_{n+1}$ on $u$ satisfying
         $p = q_{i'}$ and $q = q_i$. 
      \end{enumerate}
\end{definition}

We call \emph{good} a sequence of profiles which is maximal,
consistent and valid.

\subsubsection{Completeness}
%
 Given an o-graph
$w$, we denote by $\profseq_C(w)$ (or just $\profseq(w)$ when $C$ is
clear from context) the sequence of
$C$-profiles of its input word. In the following, we prove that given an instance $C$ of \lpp, the set
$\{\profseq_C(w)\mid w\models C\}$ is included in the set of valid and
maximally consistent sequences of $C$-profiles.

We first show that the profile sequence of an o-graph for an \lpp is maximally consistent, while next lemma proves that validity of an o-graph ensures validity of its profile sequence.
\begin{lemma}\label{lem:consmax}
Given an instance $C$ of \lpp and an o-graph $w$,
$\profseq_C(w)$ is maximally consistent.
\end{lemma}

\begin{proof}
Let $w=(u,(v,o))$, $\profseq_C(w) = \profile_1\ldots \profile_n$,
and for all $i$, $\profile_i=(\sigma_i,S_i,A_1^i\ldots A_{m_i}^i)$. 
Let also $\theta_i = (\sigma_i, S_i, B_1^i\dots B_{n_i}^i)$ be the
full profile of position $i$. 

\paragraph{Consistency} We prove every condition of the definition one
by one.

$(1)$ The first condition of the definition
of consistency is fulfilled by construction of the sets $S_k,
S_{k+1}$, which are the set of states reached by the accepting runs of
$A_{\Psi}$ on the prefixes $\sigma_1\dots \sigma_{k-1}$ and
$\sigma_1\dots \sigma_k$ respectively. Clearly, for all
$s\in S_k$, there exists $s'\in s\cdot \sigma_k\cap S_{k+1}$ and
conversely.

$(2)$ To show the other two conditions, let us define the \emph{full graph}
$G_w$ of $w$, whose vertices are clause occurrences of the full
profiles of each position respectively.  By definition of full profiles,
every clause of a full profile of an input position contains
information about some output position of $v$, in particular there is
a bijection between the clauses of each full profile and the output
positions of $v$. In other words, any output position $j$ of $v$ 
gives rise to a clause $B_j^i$ in the full profile of input position
$i$, for all $i$. We let $trace(j) = B_j^1\dots B_j^n$ be the sequence
of such clauses, taken in order of input positions. Note that by
definition of full profiles, $trace(j)$ is necessarily of the form 
$$
(\gamma, R_1, \rightarrow)\dots (\gamma, R_{\ori(j)-1},
\rightarrow)(\gamma,\cdot)(\gamma,R_{\ori(j)+1},\leftarrow)\cdots
(\gamma, R_n,\leftarrow)
$$ where $\gamma$ is the label of $j$ and 
$R_1,\dots, R_n\subseteq S_\Psi^2$. 
It is not difficult to see that by definition of the successor
relation, $B_j^{i+1}$ is a successor of $B_j^i$ with respect to
$S_i,S_{i+1}$ and $\sigma_i$, for all $i<n$ and all output position $j$. We can
therefore picture the sequence of full profiles of $w$ as follows,
where $\xrightarrow{S}$ denote the successor relation:
$$
\begin{array}{l|ccccccccccccccccc}
  & \profile^f_1 & & \profile^f_2 & & \ldots & & \profile^f_n \\
\hline \text{input symbol} & \sigma_1  & & \sigma_2 & & \ldots & & \sigma_n \\
\hline \text{$S$-component} & S_1 && S_2 & & \ldots & & S_n \\

\hline trace(|v|) & B_{|v|}^1 & \xrightarrow{S} & B_{|v|}^2 & \xrightarrow{S} & \ldots
                                     & \xrightarrow{S} & B_{|v|}^n \\

trace(|v|-1) & B_{|v|-1}^1 & \xrightarrow{S} & B_{|v|-1}^2 & \xrightarrow{S} & \ldots
                                     & \xrightarrow{S} & B_{|v|-1}^n \\
\vdots & \vdots & \vdots &\vdots & \vdots & \vdots &\vdots & \vdots  \\ 

trace(1) & B_{1}^1 & \xrightarrow{S} & B_{1}^2 & \xrightarrow{S} & \ldots
                                     & \xrightarrow{S} & B_{1}^n \\
\end{array}
$$
It shall now be easy to see that condition (2) of the definition of consistency
is fullfiled. Indeed, every clause occurrence in the full profile of
position $k$ has at least one successor in the full profile of
position $k+1$, and conversely for predecessors. Moreover, the profiles
$\lambda_k,\lambda_{k+1}$ are obtained by $\alpha$-abstraction of the
full profiles of positions $k$ and $k+1$ respectively, and
$\alpha$-abstraction preserves the set of clauses (\ie for any
full profile $\profile^f$, the set of clauses occurring in $\profile^f$ is the
same as the set of clauses occurring in $\alpha(\profile^f)$).

$(3)$ Now, to prove the last condition, we formally define the \emph{full
  graph} of $w$, which is roughly the labelled DAG given by the traces
in the picture. Formally, it is the triple $G_w = (V,E,\ell)$ where 
$V = \{ 1,\dots,|u|\}\times \{ 1,\dots,|v|\}$ with $\ell(i,j) =
B_j^i$, and $E = \{ (i,j),(i+1,j)\mid 1\leq i<|u|, 1\leq j\leq
|v|\}$. For all $1\leq i \leq |u|$, let denote by $D_i$ the set of
clause occurrences $(i,j)$ in $\profile_i$ which are removed by the
profile abstraction $\alpha$, \ie all the $(i,j)$ such that there
exist $j_1<j<j_2$ such that $B_{j_1}^i = B_j^i = B_{j_2}^i$. Let $K =
V\setminus (\bigcup_{1\leq i\leq |u|} D_i)$ (the vertices that are kept
by the $\alpha$-abstraction). The subgraph of $G_w$ induced by $K$ is
defined as the graph obtained by removing all vertices which are not
in $K$, and their incoming / outgoing vertices. Let us denote by $G_K$
this subgraph. 

\vspace{2mm}
\textbf{Claim} $G_{\profseq(w)}$ and $G_K$ are isomorphic.
\vspace{2mm}

The graph $G_{\profseq(w)}$ is obtained from the sequence of profiles
of each input position, each profiles being itself obtained by the
$\alpha$-abstraction on full profiles, the same operation as the one actually performed
on the full graph to obtain $G_K$. Hence there exists a natural
label-preserving bijection $\mu$ from the set of vertices of $G_{\profseq(w)}$ to the
set of vertices $G_K$, which preserves the vertical order, in the
sense that any two vertices $(i,j_1)$ and $(i,j_2)$ of
$G_{\profseq(w)}$ satisfy $j_1\leq j_2$ iff $\mu(i,j_1) = (i,k_1)$ and
$\mu(i,j_2) = (i,k_2)$ satisfy $k_1\leq k_2$.

We now prove that edges of $G_K$ appear in $G_{\profseq(w)}$, and conversely, which concludes the proof of the claim.
 Let us consider an edge of $G_K$, between $(i,j)$ and $(i+1,j)$.
Notice that either $\mu^{-1}(i,j)$ is labelled by some $(\gamma, R,
\rightarrow)$, or  $\mu^{-1}(i+1,j)$ is labelled by some $(\gamma, R,
\leftarrow)$. The two cases are symmetrical, we prove only one case,
and assume that $\mu^{-1}(i,j)$ is labelled by some clause $(\gamma,
R, \rightarrow)$, or equivalently, that the position $i$ of $trace(j)$
is of the form $(\gamma, R, \rightarrow)$. Let us also assume without
loss of generality that the position $(i,j)$ of $G_K$ corresponds to
the maximal occurrence (for the vertical order) of the clause $(\gamma, R, \rightarrow)$ in the
$i$th full profile of $w$. The case where it corresponds to the
minimal occurrence is, again, symmetrical. Suppose that there is no
edge $(\mu^{-1}(i,j), \mu^{-1}(i+1,j))$ in $G_{\profseq(w)}$, we will
show a contradiction. The considered vertices and edges are depicted
on Fig.~\ref{fig:complete}. 


\begin{figure*}[t]
\begin{center}
\begin{tikzpicture}

\foreach \i/\j in {0/i,2/i+1,6/i,8/i+1}{
\draw[color=blue] (\i,0) -- (\i,3);
\node () at (\i,-0.5) {$\j$};
}

\node () at (1,-1) {$G_{\text{seq(w)}}$};
\node () at (7,-1) {$G_k$};

\node (clause) at (0,1) {$\bullet$};
\node () at (0,1) [left,color=green!50!black] {$(\gamma,R,\rightarrow)$};

\node (y) at (2,2.5) {$\bullet$};
\node () at (2,2.5) [left] {$y$};

\node (sy) at (2,1.8) {$\bullet$};

\draw[->, >=stealth,thick,color=green!50!black] (clause) -- (y);

\node (x) at (6,2.5) {$\bullet$};
\node () at (6,2.5) [left,color=green!50!black] {$(\gamma,R,\rightarrow)$};
\node () at (6,2.5) [above right] {$x$};

\node (my) at (8,2.5) {$\bullet$};
\node () at (8,2.5) [right] {$\mu(y)$};

\node (c2) at (6,1.5) {$\bullet$};
\node () at (6,1.5) [left,color=green!50!black] {$(\gamma,R,\rightarrow)$};

\node (smy) at (8,1.5) {$\bullet$};

\draw[->, >=stealth,thick,color=green!50!black] (x) -- (my);
\draw[->, >=stealth,thick,color=green!50!black] (c2) -- (smy);

\node (j) at (8,1.5) [right] {$j$};

\path[->,>=stealth,thick,dashed,color=white!50!black] (y) edge [near start,above,bend left=20] node {$\mu$} (my);
\path[->, >=stealth,thick,dashed,color=white!50!black] (sy) edge [near start,below right,bend left=19] node {$\mu$} (smy);
\path[->, >=stealth,thick,dashed,color=white!50!black] (clause) edge [below,bend right=20] node {$\mu$} (c2);

\end{tikzpicture}
\caption{Proof of Lemma~\ref{lem:consmax}}\label{fig:complete}
\end{center}
\end{figure*}
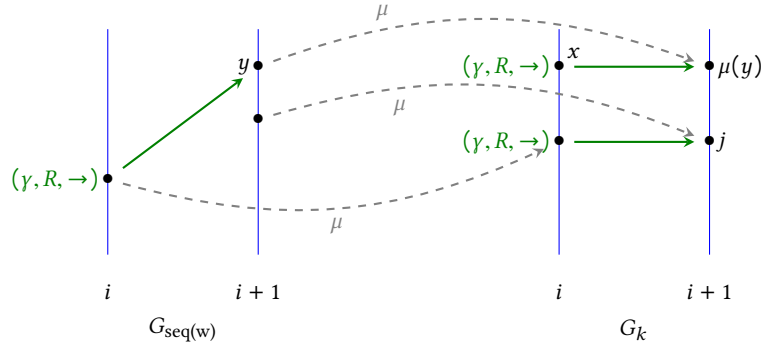

By definition of $G_{\profseq(w)}$, there is necessarily an outgoing edge
from $\mu^{-1}(i,j)$, say to a vertex $y$. This vertex $y$ is
necessarily above (in the vertical order) as $\mu^{-1}(i+1,j)$, by
construction of $G_{\profseq(w)}$. Let $\ell$ be the label of
$y$. Since $\ell$ is a successor clause of $(\gamma, R, \rightarrow)$
(according to the definition of successor clauses), $\ell$ is
necessarily of the form $(\gamma, R',\rightarrow)$ or
$(\gamma,\cdot)$. By Prop.~\ref{prop:uniquepred},
$(\gamma,R,\rightarrow)$ is the unique predecessor clause of $\ell$. Now,
consider the graph $G_K$. The vertex $\mu(y)$ is labelled by $\ell$
and is above $(i+1,j)$, since $\mu$ preserves labels and the vertical
order. The vertex $\mu(y)$ has a predecessor (in $G_k$), say $x$,
which is necessarily labelled by $(\gamma,R,\rightarrow)$. Indeed,  as we
saw, $\ell$ has a unique predecessor clause, and by definition of
$G_w$, the edge relation is compatible with the successor clause
relation (i.e. if there is an edge $(g,h)$ in $G_K$, $g$ is labelled
by the clause $c_1$ and $h$ by the clause $c_2$, then $c_2$ is a
successor clause of $c_1$). Moreover, $x$ is above $(i,j)$ in $G_K$,
which contradicts the fact that $(i,j)$ was the maximal occurrence of
$(\gamma,R,\rightarrow)$ in $G_K$. Therefore, 
$(\mu^{-1}(i,j), \mu^{-1}(i+1,j))$ is an edge of $G_{\profseq(w)}$.

The converse is proved with similar arguments, so we rather sketch the
proof than give the full details.  Let us consider an edge $(p,q)$ in $G_{\profseq(w)}$ with $p$ labelled by $(\gamma, R, \rightarrow)$ and with $\mu(p)=(i,j)$.
Let us assume that $p$ is the maximal occurrence of the clause $(\gamma, R, \rightarrow)$, in the $i$th profile of $w$. Then it must correspond to the maximal position of the $i$th column of $G_w$. Then the successor of $(i,j)$ in $G_w$ corresponds to the maximal position in the $i+1$th column of $G_w$ which is a successor of $(\gamma, R, \rightarrow)$ (otherwise $j$ would not be maximal). Hence the edge $(p,q)$ of  $G_{\profseq(w)}$ corresponds to the edge between $(i,j)$ and $(i+1,j)$ in $G_w$, which must also appear in $G_K$ since its vertices are not deleted.


Clearly in the full graph of $w$,  there is no edge with two incoming or two outgoing
edges, nor any crossing. Since $G_K$ is obtained by removing nodes
from the full graph, these properties are transfered to $G_K$, and by
the claim, to $G_{\profseq(w)}$, proving condition (3) of the
consistency definition.

The fact that $\lambda_1$ and $\lambda_n$ are respectively initial and
final are direct consequences of the definition of the profile of an
input position.

\paragraph{Maximality} To prove maximality, first, notice that the
sets $S_i$ are maximal. Assume it is not the case, \ie there is a
sequence $S'_1\dots S'_n$ such that $S_i\subseteq S'_i$ for all $i$ 
and there exists $s\in S'_i\setminus S_i$ for some $i$. Then by
condition $1$ of the definition of consistency, we could construct 
an accepting run of $A_{\psi}$ on $u$, reaching $s$ at position $i$,
and by definition of the set $S_i$, $s$ would have been already in
$S_i$.

The second condition of maximality is rather direct from the claim and
the definition of the sets $R$ in the profile of an input
position. Indeed, by this definition, the sets $R$ are ``maximal'' in
the full graph $G_w$, and this property is unchanged when going to
$G_K$ (since the sets $R$ are not modified). 
\end{proof}

\begin{restatable}{lemma}{lemModSeq}
\label{lem:model-to-seq}
Given an instance $C$ of \lpp and $w$ a model of $C$ then
$\profseq(w)$ is $C$-valid.
\end{restatable}

\begin{proof}

Let $C$ be an instance of \lpp, let $w=(u,(v,o))$ be a model of $C$ and
consider $\fprofseq(w)$ the full profile sequence of $w$. The notion of validity can be extended to full profiles in a natural way (\ie without changing the definition).
It is quite easy to show that if $w$ is a model then $\fprofseq(w)$ is valid. Indeed, by definition, if a profile of $\fprofseq(w)$ violates a universal constraint $(\gamma,\gamma',d,\psi)$ this means that we have two clauses $A_i=(\gamma,\cdot)$, $A_j=(\gamma',R,v)$ such that $R\cap \Select_\psi\neq\emptyset$ and $i$ and $j$ respect direction $d$. If we call $i^\out$ and $j^\out$ the $i$th and $j$th output positions, then we have that $i^\out$ and $j^\out$ violate the universal constraint.
Similarly, if the $i$th output position of $w$ labelled by $\gamma$ satisfies an existential constraint $(\gamma,E)$, then there is a triple $(\gamma',d,\psi)\in E$ and an output position labelled by $\gamma'$, the $j$th, such that they satisfy $d$ and $\psi$.

Now we show that, by construction, going from $\fprofseq(w)$ to $\profseq(w)$ preserves validity. First it is obvious that removing clauses can only increase the chance of satisfying a universal constraint. Secondly we show that if a clause appears in a full profile more than twice, the middle occurrences can be safely removed without removing necessary witnesses for existential constraints. Let us assume that in some profile $(\sigma,S,A_1\ldots A_{m})$, a clause $A$ appears at positions $i_0,\ldots,i_{n+1}$ with $n>0$. We claim that if $A_{i_j}$ with $1\leq j\leq n$ is a valid witness of another clause for some existential constraint, then either $A_{i_0}$ or $A_{i_{n+1}}$ as well. Going from the full profile to the profile preserves validity.
\end{proof}

\begin{remark}
Note that the converse of the previous lemma is false: there are o-graphs which are not models but whose profile sequence is valid. However we show in the following that a valid sequence of profile is always the profile sequence of \emph{some} model.
\end{remark}

Finally, we obtain the following completeness corollary:

\begin{corollary}[Completeness]\label{Cor-complete}
    For all models $w$ of an \lpp instance $C$, $\profseq(w)$ is good,
    i.e. it is maximal, consistent and $C$-valid. 
\end{corollary}

\subsubsection{Soundness}
Given a \emph{good} sequence of profiles for an \lpp instance $C$, \ie a valid and maximally consistent one, we prove that we are able to construct a valid o-graph of $C$. 
Its input word is the underlying word of the profile sequence, and the output of a given input position is given by the local clauses of its profile.
Since we know what the output positions are, all that is needed to get a valid o-graph is an order over the local clauses in such a way that all constraints are satisfied.
By definition of consistency local clauses are bijective with maximal paths.
Since a profile is made of a sequence of clauses, the maximal paths meeting at one profile are naturally ordered. 
Extending this to all profiles leads to a partial order on maximal paths, and hence on local clauses.
This partial order has to be verified by any o-graph having this sequence of profiles.
We prove that for any total ordering of local clause satisfying this partial order, we are able to construct an o-graph having this sequence of profiles, and that this procedure preserves validity.


\paragraph{Partial order on maximal paths} 
Given a sequence of
profiles $s = \lambda_1\dots \lambda_n$, 
we define a relation $\leqdot_s$ 
between maximal paths of $G_s$. 
Formally, $\leqdot_s$ is the
transitive closure of the relation $\leqdot'_s$ defined on maximal paths of $G_s$ by $\pi\leqdot'_s
\pi'$ if
there exist a vertex $(i,j_1)\in \pi$ and a vertex $(i,j_2)\in
\pi'$ such that $j_1\leq j_2$. 
The order is illustrated on Fig.~\ref{fig:abstraction}.
\begin{lemma}\label{lem-PartOrder}
Given a consistent sequence of profiles $s$, the relation $\leqdot_s$ is 
 a partial order on the maximal paths
 of $G_s$.
 \end{lemma}
 \begin{proof}
  The relation $\leqdot_s$ is clearly reflexive and it is transitive by definition.
  All that remains to prove is that for any two diffrent paths $\pi$ and $\pi'$, 
  we do not have $\pi\leqdot_s \pi'$ and $\pi'\leqdot_s\pi$.
  We prove that this situation induces a crossing at some point in the graph.
  Which concludes the proof since crossing does not happen in a consistent sequence by condition $(3)$ of Proposition~\ref{prop-consist}.
  
  Assume that we have two such different paths $\pi$ and $\pi'$.
  Then there exist two sequences of different paths $(\pi_k)_{k=1}^n$ (resp. $(\pi'_\ell)_{\ell=1}^m$) 
  such that $\pi=\pi_1 \leqdot_s \ldots\leqdot_s \pi_n=\pi'$ (resp. $\pi'=\pi'_0 \leqdot_s \ldots\leqdot_s \pi'_m=\pi$).
  Moreover, for all $1\leq k<n$ (resp. $\ell<m$), let $i_k$ (resp. $i'_\ell$) be an input position on which we have two positions $j_{k,1} <j_{k,2}$ (resp. $(j'_{\ell,1} < j'_{\ell,2}$) such that $(i_k,j_{k,1})\in \pi_k$ and $(i_k,j_{k,2})\in \pi_{k+1}$ (resp. $(i'_\ell,j'_{\ell,1})\in \pi'_\ell$ and $(i'_\ell,j_{\ell,2})\in \pi'_{\ell+1}$). Figure~\ref{Fig-pathOrder} illustrates these notations.
  
  We prove that such sequences generate crossings by induction on $n+m$.
  If $n=m=1$, then the paths $\pi$ and $\pi'$ appear one on top on another. 
  In particular, we have that $(i_1,j_{1,1})$ and $(i'_1,j'_{1,2})$ belong to $\pi$, 
   $(i_1,j_{1,2})$ and $(i'_1,j'_{1,1})$ belong to $\pi'$,
   and on $i_1$ we have $j_{1,1}<j_{1,2}$ while on $i'_1$ we have 
   $j'_{1,2}>j'_{1,1}$. Since the paths are continuous by definition, this induces a crossing 
   between $\pi$ and $\pi'$.
   
  Let assume now that $n>1$, and that if two different paths are ordered both ways by $\leqdot_s$ by a sequence of length smaller than $n+m$, then there exists a crossing.
  First, if there is a crossing between $\pi=\pi_1$ and $\pi_2$, then we get our conclusion.    
  Secondly, if $\pi$ ranges over to $i_2$ as defined previously, then there is a clause $(i_2,j)$ that belongs to $\pi$.
  If $j>j_{2,1}$, it implies that there is a crossing between $\pi$ and $\pi_2$.
  If $j=j_{2,1}$, then the âths $\pi$ and $\pi_2$ have one common node. Since they are different paths, we have a node with at least two ingoing or outgoing arrows, which contradicts the consistency definition.
  So if $j<j_{2,1}$, then we also have $j<j_{2,2}$ and hence $\pi_2$ is not needed in the sequence.
  We get a strictly smaller witness sequence for $\pi$ and $\pi'$, and hence by induction there is a crossing in the subsequence with $\pi_2$ deleted.

	Now assume that $\pi$ does not range over to $i_2$.
	By transitivity we have that $\pi_2\leqdot_s\pi'$ and $\pi'\leqdot_s \pi_2$.
	Then if we consider the sequence $(\pi'_\ell)_{\ell=1}^m$, the sequence has to pass by $i_2$ 
	in order to reach the path $\pi$.
	In other words, there is a path $\pi'_h$ that has a node $(i_2,j_h)$.
	If $j_h>j_{2,1}$, we get that the sequences $(\pi_k)_{k=2}^n$ and $(\pi'_\ell)_{\ell=1}^h$ are respectively witnesses of $\pi_2\leqdot_s\pi'$ and $\pi'\leqdot_s \pi_2$ of length strictly smaller than $n+m$, and hence contain a crossing by induction.
	If $j_h=j_{2,1}$, we have a node with two ingoing or outgoing edges, which contradicts the consistency definition.
	If $j_h<j_{2,1}$, then the sequences $(\pi_k)_{k=1}^2$ and $(\pi'_\ell)_{\ell=h}^m$ are a witness of $\pi\leqdot_s\pi'_h$ and $\pi'_h\leqdot_s \pi$ of length strictly smaller than $n+m$, and hence contain a crossing.
	
	\begin{figure}
	\begin{center}
	\begin{tikzpicture}
	\foreach \i/\j/\k in {1/1/1, 2/2/1.5, 3/3/2.3, n-1/4.5/3,n/5.5/3.5}
	{
	\node (pi\i) at (\j-0.3,6.5-\k) {$\pi_{\i}$};
	\draw (pi\i) -- (\j+1.5,6.5-\k);	
	}
	
	\foreach \i/\j in {1/2.1,2/3.45,n-1/5.8}
	{
	\node[color=gray] (i\i) at (\j,0) [above] {$i_{\i}$}; 
	\draw[dashed,color=black!20!white] (i\i) -- (\j,6);
	}

	\foreach \i/\j/\k in {m/0.5/1, m-1/1/1.5, 2/4.5/2.2,1/5.5/3}
	{
	\node (pi\i) at (\j+2.3,\k) {$\pi'_{\i}$};
	\draw (pi\i) -- (\j+0.4,\k);	
	}	

	\foreach \i/\j in {m-1/1.45,1/6.3}
	{
	\node[color=gray] (i\i) at (\j,0) [above] {$i'_{\i}$}; 
	\draw[dashed,color=black!20!white] (i\i) -- (\j,6);
	}
	
	\node[rotate=135] (dot) at (4.7,4) {$\cdots$};
	\node[rotate=45] (dot2) at (4.2,2) {$\cdots$};
	
	\node[color=gray,above right] (j11) at (2.1,5.5) {\footnotesize $j_{1,1}$};
	\node[color=gray,below right] (j12) at (2.1,5) {\footnotesize $j_{1,2}$};	
	\node[color=gray,above right] (j21) at (3.45,5) {\footnotesize $j_{2,1}$};
	\node[color=gray,below right] (j22) at (3.45,4.2) {\footnotesize $j_{2,2}$};
	\node[color=gray,above right] (jn1) at (5.8,3.5) {\footnotesize $j_{n-1,1}$};
	\node[color=gray,above right] (jn2) at (5.8,3) {\footnotesize $j_{n-1,2}$};		

	\node[color=gray,below right] (jp11) at (6.3,3) {\footnotesize $j'_{1,1}$};
	\node[color=gray,below right] (jp12) at (6.3,2.2) {\footnotesize $j'_{1,2}$};
	\node[color=gray,above left] (jpm1) at (1.45,1.5) {\footnotesize $j'_{m-1,1}$};
	\node[color=gray,above left] (jpm2) at (1.45,1) {\footnotesize $j'_{m-1,2}$};		
	
	\end{tikzpicture}
	\caption{The sequence of ordered path back and forth between $\pi$ and $\pi'$. Recall that $\pi_1=\pi=\pi'_m$ and $\pi_n=\pi'=\pi'_1$. The $j_{a,b}$ points indicate intersections.}\label{Fig-pathOrder}
	\end{center}
	\end{figure}
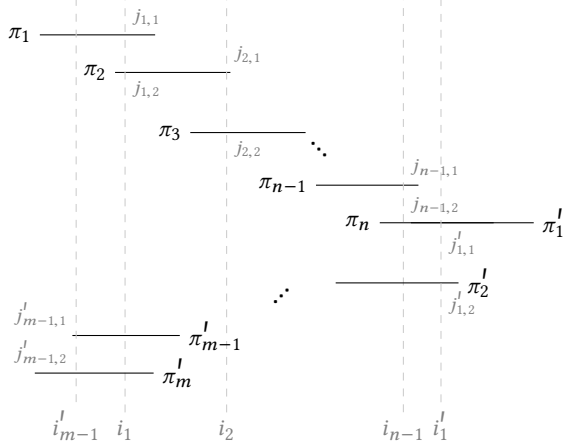
 \end{proof}

We can now prove a soundness result: from any $C$-valid and maximally
consistent sequence of profiles $s$, we can reconstruct models of $C$. As
we have seen, the elements of the relation $\leqdot_s$ are in
bijection with the local clause occurrences of $s$. 
By Lemma~\ref{lem-PartOrder}, the relation $\leqdot_s$ is a partial order when $s$ is consistent, and hence can be linearised.
We call 
\emph{linearisation} of $\leqdot_s$ any total order on the elements 
of $\leqdot_s$ which is compatible with this partial order. 
Since the paths of $\leqdot_s$ are in bijection with the
local clause occurrences of $s$, any such linearisation $\leq$ induces an
output word, and thus an o-graph $(u,(v,o))$. 
Formally, it is 
defined by $u$ the sequence of $\sigma$-symbols in $s$, and for any
occurrence $(i,j)$ of a local clause $(\gamma,\cdot)$, if the path $\pi_{i,j}$ is
the $k$th in the linearisation, then the $k$th position of $v$ is
labelled by $\gamma$, and its origin is $o(k) = i$.

The next Lemma proves that for any linearisation of $\leqdot_s$, we
can construct a model with $s$ as a sequence of profiles that respects
the partial order $\leqdot_s$. 


\begin{restatable}{lemma}{lemLinear}\label{lem:lemLinear}
Given an instance $C$ of \lpp, a valid and maximally consistent
sequence of  $C$-profiles $s = \profile_1\ldots \profile_n$,
and a linearisation $\leq$ of $\leqdot_s$, the o-graph induced by $\leq$ is $C$-valid
and verifies $\profseq_C(w)=s$.
\end{restatable}

\begin{proof}
Let $w=(u,(v,o))$ be the o-graph induced by $\leq$ and we set $s'=\profile_1'\ldots \profile_{n'}'$ the sequence of profiles of $w$. 
We now prove that $s=s'$.
First, let us remark that by construction the underlying word of $s'$ is $u$, and equal to the underlying input word of $s$. Thus $n=n'$, and for $k\leq n$, $\sigma_k=\sigma_k'$ and since $s$ is maximal, $S_k=S_k'$.
All that is left to show is that for any $k\leq n$, the sequence of clauses 
of $\profile_k$ is equal to the sequence 
 of $\profile_k'$.

The local clauses of $s$ are exactly the output positions of $w$.
Then the local clauses of $s'$ are exactly the same as the ones of $s$ by construction of $w$, and there is a canonical bijection between the two.
Now since $s'$ is consistent,
each consistency clause belongs to a maximal path that contains exactly one local clause.
Let $((i_\ell,j_\ell))_{\ell=1}^n$ (resp. $((i_\ell,j_\ell))_{\ell=1}^m$) be the sequence of local clauses
whose maximal path have a node in the profile $\profile_k$ (resp. $\profile'_k$), ordered by their appearance in $\profile_k$ (resp. $\profile'_k$).
We aim to prove that these two sequences are equal, which would conclude the proof.
Indeed, as $s$ and $s'$ have the same input word and the same local clauses, each local clause generates the same path of clauses in each sequence profile, and if the profiles of the same input position have consistency clauses belonging to the same sequence of paths, then the sequences of clauses are the same.
%

First, since $w$ comes from a linearisation of $\leqdot_s$, the sequence $((i_\ell,j_\ell))_{\ell=1}^n$ appears appear in the same order in $w$, and thus in $s'$, and it generates $\profile_k$ as a subsequence of the full profile of $k$ of $w$, \ie the sequence of all traces of output positions of $w$ on $k$ before suppressing redondancy.
Now take a clause $B^k_i$ of $\profile_k'$, and let $B^k_i=(\gamma,R,v)$. Then $B^k_i$ is either the minimal or maximal clause $(\gamma,R,v)$ in the full profile of $k$ of $w$.
Without loss of generality, assume $B^k_i$ is minimal.
Let $(i_h,j_h)$ be the local clause whose maximal path contains $B^k_i$.
Since $s$ and $s'$ have the same local clauses, there exists a local clause $(f,g)$ in $s$ which is the bijective image of $(i_h,j_h)$. 
Assume now that the path $\pi_{f,g}$ is not minimal at position $k$ (in $s$).
This means there exists an element $e$ of $((i_\ell,j_\ell))_{\ell=1}^n$ that generates the clause $(\gamma,R,v)$ on $k$ such that $(i_e,j_e)$ is strictly smaller than $(f,g)$ in the $\leqdot_s$ order, and hence the local clause bijective to $(i_e,j_e)$ in $s'$ is smaller than $(i_h,j_h)$.
But since $\profile_k$ is a subsequence of the full profile of $k$ of $w$, 
we get that the local clause bijective to $(i_e,j_e)$ in $s'$ generates a clause $(\gamma,R,v)$ in the full profile of $k$ of $w$ that is smaller than $B^k_i$, which is a contradiction to the fact that $B^k_i$ was the minimal occurrence of such a clause in the full profile of $k$ of $w$.
Consequently, the clause $B^k_i$ is associated to the same local clause as the minimal clause $(\gamma,R,v)$ in $\profile_k$ (up to the bijection between local clauses of $s$ and $s'$).
Applying this to any clause of $\profile_k'$, we get that the profiles $\profile_k$ and $\profile_k'$ are associated to the same set of local clauses, concluding the proof as explained above.

Now it remains to prove that $w$ is $C$-valid.
This is a direct consequence of the facts that its sequence of profiles $s$ is $C$-valid and 
that the output positions are exactly the local clauses of $s$.
Indeed, since if $w$ does not satisfy a universal constraint, then there exists two output positions of $w$ that violates it.
Then there exists two local clauses that violate it, and by definition of clauses there is a profile with two clauses that violates this constraint. Since $s$ is $C$-valid, this is not the case, and hence $w$ satisfies all universal constraints.
Now given an output position of $w$ and an existential constraint, we know that the associated local clause either does not satisfy the constraint's label, 
or it has a valid witness in the form of a clause in its profile. 
This consistency clause belongs to a maximal path to a local clause whose associated output position is a valid witness, concluding the proof.
\end{proof}

\subsection{Back to the theorems}

In this section, we prove Proposition~\ref{prop:regDom} about the
regularity of the input domain of any \lorigin-o-tranduction, a result
which we later use to prove our main result (Theorem~\ref{thm:unif})
about the regular synthesis of \lorigin. 

The proof of these results rely mainly on the profile approach developped in the
previous section. The proof of synthesis unfolds into two major steps. First, we use the
profile automaton to create a non-deterministic one-way transducer
which associates to an input word the set of its valid sequences of
profiles. Then by realising it by a regular function, we associate to each input word a
unique sequence of profiles.
For the second part, we use some results by Courcelle to prove
that any partial order (seen as DAG) can be linearised by some
MSO-transduction, \ie there exist DAG-to-words MSO-transductions that
define linearisations.

Let us start with a key lemma. The size of an instance $C = (C_\exists,C_\forall)$ of \lpp is
$|C| = (\sum_{(\gamma,\gamma',d,\Psi)\in C_\forall} |\Psi|+3) +
\sum_{(\gamma,E)\in C_\exists} 1+\sum_{(\gamma',d,\Psi)\in E}
2+|\Psi|$, where $|\Psi|$ is the number of symbols of the MSO-formula
$\Psi$. If MSO-predicates are given by query automata, then we do not
include the size of these automata in the size of $C$, i.e., it is
defined as before where $|\Psi|$ is just replaced by $0$. 

\begin{restatable}{lemma}{lemAutomProf}\label{lem-ProfileAut}
Given an instance $C$ of \lpp, the set $\{ \profseq(w)\mid w\models
C\}$ is effectively regular. 

Moreover, if the \mso predicates of $C$ are given as query automata, then checking emptiness 
of $C$ can be done in \textsf{PSpace} in the size of $C$ and
\textsf{ExpSpace} in the number of states of the query automata.
\end{restatable}
\begin{proof}
Let $G$ be the set of $C$-valid and maximally consistent
sequences of profiles. We claim that $G = \{ \profseq(w)\mid w\models C\}$. If
$w\models C$, then $\profseq(w)$ is $C$-valid and maximally consistent by Corollary~\ref{Cor-complete}, hence $\profseq(w)\in G$. Conversely, if
$\lambda$ is a $C$-valid and maximally consistent sequences of
profiles, then take any linearisation $\leq$ of $\leqdot_s$. It
induces an o-graph $w$ such that $\profseq(w) = \lambda$ and $w\models C$ by
Lemma~\ref{lem:lemLinear}. Hence $\lambda\in \{ \profseq(w)\mid
w\models C\}$.

Then, it suffices to show that $G$ is regular.
We construct a single deterministic automaton $\mathcal{A}$ that reads sequences of profiles, 
and checks validity, consistency and maximality.
Validity is a local property that can be checked on the profile read (as long as consistency holds), 
thus it only corresponds to restricting the profiles that can be read by $\mathcal{A}$ to the valid ones.
Consistency and maximality are checked simultaneously.
Indeed, consistency requires that clauses of two consecutive profiles
are matched, and ensures that the runs appearing in one profile are
present in the other one, up to one transition taken by the predicate
automaton. Maximality is dual in the sense that once the matching is done by consistency, 
maximality amounts to check that the runs not considered in one
profile do not merge with runs in the other one, and that we cannot
globally add any accepting run.

Let us now be more precise. Let $\text{Cons}$ (resp. $\text{Val}$) be
the set of pairs $(\lambda,\lambda')$ of consistent profiles (resp. of
$C$-valid profiles). The states of the automaton $\mathcal{A}$ are
profiles enriched with some information that allows one to check
maximality. First, when $\mathcal{A}$ reads a profile $\lambda\not\in
\text{Val}$, it rejects. If from a state consisting of the enrichment
of a profile $\lambda$ it reads a profile $\lambda'$ such that
$(\lambda,\lambda')\not\in \text{Cons}$, it rejects as well. Let us now
explain what is the extra information added to the profiles. 
%
%
Let $S_\Psi$ be the union of the states of the predicate automata. 
An \emph{enrichment} of a profile $\lambda = (\sigma, S, A_1\dots
A_n)$ is a tuple $\overline{\lambda} = (\sigma,(S,S'), C_1\dots C_n)$ such that
$S'\subseteq S_\Psi$ and $C_i = (\gamma,\cdot)$ if $A_i =
(\gamma,\cdot)$, and $C_i = (\gamma, (R,R'),v)$ for $R'\subseteq
S_\Psi^2$ if $A_i = (\gamma, R, v)$. The tuple $\overline{\lambda}$ is
called an \emph{enriched profile}. The states of $\mathcal{A}$ are
enriched profiles. Intuitively, the set $S'$ consists of all states
not in $S$ which have been reached so far. If at some point, there
exist $\sigma$, a state $s\in S$ and a state $s'\in S'$, and two
transitions from $s$ and $s'$ on reading $\sigma$, towards the same
state, the automaton rejects (because $s'$ could be added to $S$,
contradicting its maximality). Otherwise, $S$ and $S'$ are updated
according to the transitions of the predicate automaton, as for a
subset construction. If eventually $\mathcal{A}$ reads the whole word
and ends up with some $S'$ containing some accepting state, it
rejects. Likewise, the information contained in $R'$ is used by $\mathcal{A}$
to monitor candidate pairs of states that could be added to $R$.
%

Formally, suppose $\mathcal{A}$ read some profile $\profile=(\sigma,S,C_1\ldots C_n)$ and is in some state
$\overline{\lambda}$ where we enrich $S$ with $S'$ and each clause $(\gamma,R,v)$ is enriched with $R'$.
%
Upon reading a profile $\profile''=(\sigma'',S'',A_1''\ldots A_m'')$,
if $\profile''\not\in \text{Val}$ or $(\lambda,\profile'')\not\in\text{Cons}$, then $\mathcal{A}$ rejects (i.e.,
there is no transition). Otherwise, let $T = S'\cdot \sigma$. If
$T\cap S''\neq \emptyset$, then $A$ rejects. Otherwise
we update $\mathcal{A}$ to the state $(\sigma,(S'',T),D_1\ldots D_m)$
constructed as follows. Let $j\leq m$, if $C_i$ is matched with
$A''_j$ (according to consistency), we define $D_j$ in the following way:
\begin{itemize}
\item if $C_i=(\gamma,R,\rightarrow)$ and $A''_j=(\gamma,R'',\rightarrow)$, then we set
$Q=\set{(p'',q)\mid \exists (p',q)\in R'\wedge p''\in p'\cdot\sigma}$. If $R''\cap Q\neq \emptyset$, then we reject. Otherwise, we set
$D_j=(\gamma,(R'',Q),\rightarrow)$.
\item if $C_i=(\gamma,R,\rightarrow)$ and $A''_j=(\gamma,\cdot)$, 
then if there is a pair $(p,q)\in R'$ where $p\in q\cdot\sigma$ and $q\in S''$ we reject, otherwise
we set $D_j=(\gamma,\cdot)$.
\item if $C_i=(\gamma,\cdot)$ and $A''_j=(\gamma,R'',\leftarrow)$, then we define
 $Q=\set{(p,q)\mid \exists p\in T\wedge p\in q\cdot\sigma}$.
 If $R''\cap Q\neq \emptyset$, then we reject, otherwise we set
$D_j=(\gamma,(R'',Q),\leftarrow)$.
\item if $C_i=(\gamma,R,\leftarrow)$ and $A''_j=(\gamma,R'',\leftarrow)$, then we set 
$Q=\set{(p'',q)\mid \exists (p',q)\in R'\wedge p''\in p'\cdot\sigma}$. If $R''\cap Q\neq \emptyset$, then we reject. Otherwise, we set
$D_j=(\gamma,(R'',Q),\leftarrow)$.
\end{itemize}

The initial state is a special state $init$.
The transition from $init$ upon reading a profile $\profile=(\sigma,S,C_1\ldots C_n)$ exists only if $S$ contains only initial states and for each $C_i$ is of the form $(\gamma,\cdot)$ or $(\gamma,R,\rightarrow)$, where the first component of $R$ contains only initial states.
This transition goes to state $\overline{\profile}=(\sigma,(S,S'),D_1\ldots D_n)$ where $S'$ is the set of all initial states of the query automata that are not in $S$, and $D_i=C_i$ if $C_i=(\gamma,\cdot)$ and $D_i=(\gamma,(R,R'),\rightarrow)$ where $R'=\{(p,q)\mid p\in S'\}$ otherwise.
The accepting states are the accepting profiles
where there is not any addition $S'$ or $R'$ containing a final state.

The language recognized by $\mathcal{A}$ is exactly the set $G$ of
good sequences of profiles. The size of an enriched profile is
$1+2^{O(|S_\Psi|)}$, because it may contain all pairs of subsets of
$S_\Psi$. Hence the number of states of $\mathcal{A}$ is
doubly exponential in the number of states of $S_\Psi$. 
Constructing this automaton explicitly would give a doubly exponential
time algorithm. Instead, we can use the classical \textsf{NLogSpace} emptiness
checking algorithm of finite automata by constructing $\mathcal{A}$
on-the-fly. The algorithm needs a counter up to a doubly exponential
value (which is singly-exponentially represented). We have to be
careful though because constructing the automaton on-the-fly requires
to check the properties $\lambda''\not\in \text{Val}$ and
$(\lambda,\lambda'')\not\in \text{Cons}$ with a reasonable
complexity. Checking validity a profile $\profile''$ requires for each
clause of $\profile''$ to scan all constraints of $C$, and 
possibly again all clauses of $\profile''$ (because constraints of $C$
are between two clauses). Overall, this can be achieved in polynomial
time w.r.t. the size of $C$ and of the profile, hence in time
polynomial in $|C|$ and exponential in the number of states of the
query automata, and so in space polynomial in $|C|$ and exponential in
the number of states of the query automata. Consistency between two
profiles can be checked with the same complexity. Overall, one obtains
a non-deterministic algorithm to check emptiness of $G$, which runs in
space polynomial in $|C|$ and exponential in the number of states of
the query automata. The result follows since $\textsf{PSpace =
  NPSpace}$ and $\textsf{NExpSpace = ExpSpace}$.

\end{proof}

We are now able to prove domain regularity. 

\propRegDom*

\begin{proof}
    Let $\varphi\in\lorigin$. By Proposition~\ref{Thm-NonErasing},
    there exists a non-erasing o-tranduction definable by some
    \lorigin-formula $\varphi'$ such that $\dom(\sem{\varphi}) =
    \dom(\sem{\varphi'})$.  By Lemma~\ref{lemma-Scott}, the formula
    $\varphi'$ can in turn be converted into a formula $\varphi''$ in
    SNF which have the same models as $\varphi'$, up to some output
    label morphism. In particular, $\dom(\sem{\varphi'}) =
    \dom(\sem{\varphi''})$. By Proposition~\ref{prop-MSOLPP}, the
    formula $\varphi''$ can be transformed into an equivalent set of
    contraints $C$, \ie such that any non-erasing o-graph $w$
    satisfies $\varphi''$ iff it satisfies $C$. Since $\varphi'$, and
    so $\varphi''$, are non-erasing, for all o-graphs $w$,
    $w\models \varphi''$ iff $w\models C$. Hence, $\dom(\varphi)
    = \dom(\varphi'') = \{ u\mid \exists w = (u,(v,o))\cdot w\models
    C\}$. Now, remind that any profile is of the form $(\sigma, S,
    A_1\dots A_n)$ where $\sigma$ is an input label. Denote by $\pi_1$
    the first projection over these tuples. We then have 
    $\dom(\varphi) = \{ \pi_1(\profseq(w))\mid w\models C\}$, which is
    regular since by Lemma~\ref{lem-ProfileAut}, $\{ \profseq(w)\mid
    w\models C\}$ is regular and regularity is preserved by morphisms
    (and in particular projection).
    
    Now we can easily extend the result to any $\elorigin$ formula $\varphi$: If we see the second-order variables as new unary predicates, we obtain a formula $\varphi'$ over an extended alphabet. By the above proof, the domain of $\varphi'$ is effectively regular, and since regular languages are stable by alphabet projection, we can project away the additional letters and obtain effectively a regular domain for $\varphi$. 
\end{proof}

\thmUnif*

\begin{proof} To ease the reading the proof, we have divided in
    several parts. 

    \paragraph{Preliminaries}
    First, let us show that we can assume \wog that $\varphi$ defines
    a non-erasing o-tranduction. Indeed, if it is not the case, then by
    Proposition~\ref{Thm-NonErasing}, $\varphi$ can be converted into
    a non-erasing \lorigin-o-tranduction $\varphi'$ such that 
    $(u,(v,o))\models \varphi$ iff $(u,(v\#u,\ori'))\models \varphi'$,
    where $\ori'$ coincides with $\ori$ on $v$, maps $\#$ to $1$ and the
    $i$-th symbol of $u$ to position $i$ in the input. Suppose
    $\varphi'$ is realisable by some MSOT $T$. Then, it is not
    difficult to transform $T$ into a realisation of $\varphi$. It
    suffices to compose $T$ with an MSOT $T'$ which maps $v\# u$ to
    $v$ and use closure under composition of MSOT~\cite{journals/eatcs/CourcelleE12} and the fact
    that closure under composition can be done while preserving
    origins, as noticed in \cite{Bojanczyk14}. Hence, it suffices to
    show the result for non-erasing \lorigin-o-tranductions.

    Let $\varphi$ be a non-erasing \lorigin-sentence over input
    alphabet $\Sigma$ and output alphabet $\Gamma$. By
    Lemma~\ref{lemma-Scott}, one can construct a non-erasing
    $\lorigin$-sentence $\varphi'$ in Scott normal form, over input alphabet $\Sigma$ and
    output alphabet $\Gamma\times \Gamma'$ which have the same models
    as $\varphi$ up to projection of $\Gamma\times \Gamma'$ on
    $\Gamma$, \ie $\sem{\varphi} = \{ (u,(\pi_\Gamma(v),o))\mid
    (u,(v,o))\in\sem{\varphi'}\}$. Any synthesis of $\varphi'$ can be
    composed with an \msot which defines the projection on $\Gamma$,
    and once again we use closure under composition of \msot to get
    the result. Therefore, we now focus on realising $\varphi'$,
    and write $\Lambda = \Gamma\times \Gamma'$ (hence $\varphi'$
    defines an o-tranduction from $\Sigma^*$ to $\Lambda^*$). By
    Proposition~\ref{prop-MSOLPP}, $\varphi'$ is equivalent
    to a system of constraints $C\in\text{MCP}$, in the sense that
    $\semo{\varphi'} = \semo{C}$.

    \paragraph{General scheme of the proof}
    
    We first define four o-transductions, some of them being functional, whose composition has the same domain
    as $\varphi'$ and is included in $\semo{\varphi'}$. Then, by
    realising in a regular manner all non-functional transductions
    of this composition, and by composing all syntheses, 
    we will get a regular synthesis of
    $\semo{\varphi'}$. The four o-transductions are defined as follows
    (they are denoted with an $f$-symbol when they are functional):
    \begin{enumerate}[leftmargin=*]
      \item $R_{pro}$ associates with any input word $u$ the set
        $R_{pro}(u) = \{ (\profseq(u,(v,o)),\ori_1)\mid (u,(v,o))\models
        C\wedge \forall i\in\dom(u), \ori_1(i) = i\}$,

      \item $f_G$ which associates with any consistent sequence of
        profiles $s$ the graph $G_s$, with origin mapping $\ori_2$
        taking any vertex of $G_s$ corresponding to an input position
        $p$ of the sequence $s$, to $p$. 

      \item $f_{par}$ which takes $G_s = (V,E,\ell)$ as input and outputs the
        partial order denoted $\leqdot_s^\Gamma$, obtained by
        restricting $\leqdot_s$ to the nodes labelled in
        $\Gamma$, \ie  $\leqdot_s^\Gamma = \leqdot_s\cap \{ (x,y)\mid
        \ell(x),\ell(y)\in\Gamma\}$. It corresponds to the partial
        order depicted in red in Fig.~\ref{fig:abstraction}, where
        each maximal path is identified by a single local clause. 
        The origin mapping is the identify: $\ori_3(x)=x$ for any
        element $x$ of the partial order $\leqdot_s^\Gamma$. 

      \item $R_{lin}$ which inputs $\leqdot_s^\Gamma$ and outputs all
        linearisations of it, again with an identity origin mapping $\ori_4$. 
    \end{enumerate}

      \textbf{Claim}  Let\footnote{Note that we compose here
        relations with origin information: the origin mappings are as
        well composed. Formally, for any two relations $R_1,R_2$, and
        element $u$, we define $(R_2\circ R_1)(u)$ as the set 
        $\{ (v,\ori_1\circ \ori_2)\mid \exists
        (v',\ori_1)\in R_1(u), \exists (v,\ori_2)\in f_2(v')\}$}  $g
      =R_{lin} \circ f_{par}\circ f_G \circ R_{pro}$. Then, 
      \begin{enumerate}
        \item $\dom(g) = \dom(\varphi')$
        \item for all $u\in\dom(g)$, all $(v,o)\in g(u)$, $(u,(v,o))\in\sem{\varphi'}$.
      \end{enumerate}
      Before proving these two points, note that they imply that any realisation of
      $g$ which preserves origins is a synthesis of $\sem{\varphi'}$.

      Let us now prove these two points. Suppose that $u\in \dom(\varphi')$, then there exists
      $(v,o)$ such that $w=(u,(v,o))\models \varphi'$, hence
      $w\models C$ and $\profseq(w)\in R_{pro}(u)$. Since $f_{G}$
      is defined for all consistent sequences of profiles, 
      $f_{par}$ is defined for all partial orders and $R_{lin}$
      is total, we get that $\profseq(w)\in\dom(R_{lin}
      \circ f_{par} \circ f_G)$ and hence $u\in \dom(g)$.

      The inclusion $\dom(g)\subseteq \dom(\varphi')$ is a consequence
      of item $2$, so let us prove item 2. Let $u\in \dom(g)$ and
      $(v,o)\in g(u)$. By definition of $g$, there exists an o-graph
      $w$ with input $u$, a linearisation $l_s$ of $\leqdot^\Gamma_s$ for
      $s = \profseq(w)$, and origin mappings $\ori_1,\ori_2,\ori_3,\ori_4$ such
      that:
      \begin{itemize}
        \item $\ori = \ori_1\circ \ori_2\circ \ori_3\circ \ori_4$
        \item $(s,\ori_1)=(\profseq(w),\ori_1)\in R_{pro}(u)$ (hence $w\models C$)
        \item $(G_s,\ori_2) = f_G(s)$
        \item $(\leqdot_s^\Gamma,\ori_3)  = f_{par}(G_s)$
        \item $(v,\ori_4)\in R_{lin}(\leqdot_s)$.
      \end{itemize}
      Then, $(u,(v,o))$ is the o-graph induced by some linearisation
      of $\leqdot_s$. Since $w\models C$, then $s = \profseq(w)$ is $C$-valid by
      Lemma~\ref{lem:model-to-seq}. It is also maximally consistent by
      Lemma~\ref{lem:consmax}. Therefore we can apply
      Lemma~\ref{lem:lemLinear} and get that 
      $(u,(v,o))\models C$, and so $(u,(v,o))\models \varphi'$.

      \paragraph{MSOT-definability of $f_{par}$, $f_G$ and
        MSOT-Synthesis of $R_{pro}$ and $R_{lin}$} Our goal now
      is to show that the functions $f_{par}$ and $f_G$ are
      MSOT-definable, and that the relation $R_{pro}$ and $R_{lin}$
      are realisable by MSOT-definable functions. Conclusion will
      follow as MSOT transductions are closed under
      composition~\cite{journals/eatcs/CourcelleE12} and moreover, the
      composition procedure preserves origins.

      \emph{MSOT-synthesis of $R_{pro}$}. By
      Lemma~\ref{lem-ProfileAut}, the set $\{ \profseq(w)\mid w\models
      C\}$ is regular. It implies that $R_{pro}$ is
      rational~\cite{Bers79}, \ie is definable by a non-deterministic
      (one-way) transducer. Indeed, if $B$ is an automaton defining $\{ \profseq(w)\mid w\models
      C\}$, it suffices to turn each of its transitions on a clause
      $(\sigma,S,A_1\dots A_n)$ into a (transducer) transition on input $\sigma$
      producing output $(\sigma,S,A_1\dots A_n)$. It is well-known
      that rational relations can be realised by rational
      functions~\cite{Bers79}, and most known realisation
      procedures (for instance based on a lexicographic ordering of
      runs) preserve origin mappings. We can conclude since rational
      functions, as a special case of functions definable by two-way
      deterministic transducers, are MSOT-definable~\cite{EH01}.

      \emph{MSOT-definability of $f_{G}$}. The function $f_G$ inputs a
      consistent sequence $s$ and outputs $G_s$, which is a
      graph. Hence by MSOT-definability we mean MSOT from string to
      graphs. We should now make clear how we represent $G_s$ as a
      structure. We use the signature $\mathcal{D} = \{ E(x,y),
      (c(x))_{c\in\mathcal{C}}, \rightarrow, \downarrow\}$ where $E$ is
      the edge relation, $\mathcal{C}$ is the set of clauses, 
      $c\in\mathcal{C}$ are monadic
      predicates for node labels, $\rightarrow$ and $\downarrow$ are
      respectively induced by the input order on abscissas of $G_s$
      and $\downarrow$ by ordinates of $G_s$ ($(p_1,p_2)\rightarrow
      (p'_1,p'_2)$ if $p_1\leq p'_1$ and 
      $(p_1,p_2)\downarrow
      (p'_1,p'_2)$ if $p_2\leq p'_2$). 
      Let us
      now sketch the definition of an MSO-transduction producing $G_s$
      from $s$: since $s$ is a word and we aim to produce a graph
      whose nodes are the clause occurrences of $s$, we use as many
      copies of $s$ as the maximal number $m$ of clause occurrences in
      a profile of $s$. A copy node $(i,j)$ thus denote the $j$th
      clause of the $i$th profile of $s$. The predicate $\rightarrow$
      is then naturally defined by a formula
      $\phi_{\rightarrow}^{i,j}(x,y)\equiv x\preceq y$ where $\preceq$
      is the linear order on positions of $s$. We also have
      $\phi_{\downarrow}^{i,j}(x,y)\equiv x = y$ if $i<j$, and $\bot$ otherwise. 
      To define the
      edge relation, we have to come back to the definition of $G_s$. 
      For instance, there is an edge between $x^i$ and $y^j$ if $y$ is
      the successor of $x$ in the input sequence $s$, the $i$th clause $A$
      of $\lambda_x$ is of the form $(\gamma,R,\rightarrow)$, $i$ is
      the smallest occurrence of $A$ in $\lambda_x$, and
      the $j$th clause $B$ of profile $\lambda_y$ is a successor of
      $A$, according to the definition of the successor relation between
      clauses, and $j$ is the smallest successor of $A$ in
      $\lambda_y$. Other cases are similar and it shall be clear that
      all these properties are MSO-definable.

      \emph{MSOT-definability of $f_{par}$} Now, $f_{par}$ must output
      a partial order from $G_s$. We represent this partial order naturally as a
      DAG. To implement $R_{lin}$, we also need this DAG to be
      \emph{locally ordered}, \ie all the successors of a node are
      linearly ordered by some order we denote $\leq_{succ}$. Hence, the
      output structure is over the signature $\{ E(x,y), \leq_{succ},
      (\gamma(x))_{\gamma\in\Gamma}\}$. Since $f_{par}$ only adds
      edges, we just take a single copy of the input structure, and we
      filter the nodes which are not labelled in $\Gamma$ (thanks to
      monadic MSO formulas).  
      Then, there is an edge between two vertices $(x,y)$ in the DAG
      iff there was an edge in $G_s$ between these two vertices, or
      there is a vertex $x'$ on the maximal path of $x$ in $G_s$, and
      a vertex $y'$ on the maximal path of $y$ in $G_s$, such that 
      $x'\downarrow y'$ and $x',y'$ have the same abscissa,
      \ie $x'\rightarrow y'\wedge y'\rightarrow x'$. Since
      connectivity is MSO-definable on graphs, it should be clear that
      these properties are MSO-definable over $G_s$. The local order
      $\leq_{succ}$ is defined by the formula $\phi_{\leq_{succ}}(x,y)
      \equiv x\rightarrow y$. Finally, the labels are preserved, hence
      defined by a formula $\phi_\gamma(x) = \gamma(x)$.

      \emph{MSOT-synthesis of $R_{lin}$} We use a known result
      by Courcelle~\cite{cou95} about MSO-definable topological
      sorts of graphs. More precisely, it is shown in Theorem 2.1 of
      \cite{cou95} that there exists an MSOT that, given
      any locally ordered DAG, produces a linear order of the dag
      compatible with its edge relation. This MSOT uses only one copy
      of the input DAG structure, and is defined by some MSO formula
      $\phi_{<}(x,y)$
      over the signature of $\leqdot_s^\Gamma$, with two free variables, which defines a linearisation $<$ of the
      DAG. Since one also needs to preserve
      the labels of the nodes of the DAG, we augment this MSOT with
      label formulas $\phi_A(x) = A(x)$ for all clauses $A$. 
\end{proof}

\section{Data words}

\thmOrtoData*
\begin{proof}
The first part of the statement is a direct consequence of the definition of $\td$. 
We now focus on the equivalence of logics.

Let $\varphi$ be an \lorigin-sentence defining a non-erasing
transduction, we want to obtain an \ldata-sentence $\phi$ defining its
encoding as a typed data word. First we transform $\varphi$ into $\varphi'$ a formula where all quantifications are either input or output quantifications.
This can be done inductively on \lorigin-formula by replacing $\exists
x\ F(x)$ by $\exinput x\ F(x) \vee \exoutput x\ F(x)$.
Then, we
simplify the resulting formula by removing inconsistent use of
variables in the predicates with respect to the type of their quantifiers. 
For that, we say that the occurrence of a term $t$ is of type $\inp$
if it is equal to $x$ where $x$ is quantified over the input, or of
the form $\ori(t')$ for some term $t'$.
It is of type
$\out$ if $t=x$ for $x$ a variable quantified over the output. Now, we replace in $\phi$ all occurrences of
the following atoms by $\bot$ under the following conditions: 
\begin{itemize}
    \item the atom is $\gamma(t)$, for $\gamma\in\Gamma$, and $t$ is not of type $\out$,
    \item the atom is $t_1\leqoutput t_2$ and some $t_i$ is not of
      type $\out$,
    \item the atom is $\{\psi\}(t_1,t_2)$ and some $t_i$ is not
      of type $\inp$.
\end{itemize}
By doing this we obtain a new formula which is equivalent to $\phi$,
and makes a consistent use of its variables. We do not give a name to
this new formula and rather assume that $\phi$ satisfies this
property. 

Then, we do the following replacement in $\phi$ to transform it into an
equivalent \ldata-formula. First, similarly to the bijection $\td$ 
in which the origin of a position
becomes its data value, any term of the form $\ori^n(x)$ is
replaced by $x$. Then, any occurrence of 
an MSO predicate $\{\psi\}(x,y)$ is replaced by $\{\psi'\}(x,y)$,
where $\psi'$ is obtained by replacing in $\psi$ all atoms of the form
$x\leqinput y$ by $x\preceq y$.
We also replace the atom of the form $x\leqoutput y$ by $x\leq y$.
If we denote by $\phi'$ the obtained formula, by construction we have $(u,(v,o))\models
\phi$ iff $\td(u,(v,o))\models \phi'$. 

%

\begin{example}
For instance, consider the following formula $\phi$:
$$
\forall x\  \{ \sigma(x')\}(x)\rightarrow \exists y\ 
\{ y'\leqinput x'\}(\ori(y),x)
$$
where $\sigma\in\Sigma$. It expresses the fact for any input position
labelled $\sigma$, there is another input position before which is the
origin of some output position. First, note that $\forall x\ \psi$ being a shortcut for
$\neg \exists x\ \neg\psi$, the first replacement by typed quantifiers
gives the formula $\forinput x\ \psi \wedge \foroutput x\
\psi$. Then, the first rewriting step of $\phi$ gives:
$$
\begin{array}{ccc}
\forall^\inp x\  \{ \sigma(x')\}(x) & \rightarrow & \exists^\inp y\ 
\{ y'\leqinput x'\}(\ori(y),x) \\
& & \vee \\
& & \exists^\out y\ 
\{ y'\leqinput x'\}(\ori(y),x) \\
& \wedge & \\
\forall^\out x\  \{ \sigma(x')\}(x) & \rightarrow & \exists^\inp y\ 
\{ y'\leqinput x'\}(\ori(y),x) \\
& & \vee \\
& &  \exists^\out y\ 
\{ y'\leqinput x'\}(\ori(y),x) \\
\end{array}
$$
After the simplification step according to types, we get:
$$
\begin{array}{ccc}
\forall^\inp x\  \{ \sigma(x')\}(x) & \rightarrow & \exists^\inp y\ 
\bot \\
& & \vee \\
& & \exists^\out y\ 
\{ y'\leqinput x'\}(\ori(y),x) \\
& \wedge & \\
\forall^\out x\  \bot & \rightarrow & \exists^\inp y\ 
\bot \\
& & \vee \\
& &  \exists^\out y\ 
\{ y'\leqinput x'\}(\ori(y),x) \\
\end{array}
$$
which could be again simplified into:
$$
\forall^\inp x\  \{ \sigma(x')\}(x)  \rightarrow   \exists^\out y\ 
\{ y'\leqinput x'\}(\ori(y),x)
$$
Then, according to all the replacement rules, one gets the
$\ldata$-formula
$$
\forall x\  \{ \sigma(x')\}(x)  \rightarrow   \exists y\ 
\{ y'\preceq x'\}(y,x)
$$
which expresses that for all positions $x$, if the data type of $x$ is
$\sigma$, then there is a position $y$ whose data is smaller than that
of $x$. 
\end{example}

The converse is slightly easier, since we do not have to deal with
inconsistent use of variables. Any \ldata-sentence $\psi$ is converted
into an \lorigin-sentence $\psi'$ by doing the following replacements:
\begin{itemize}
  \item any quantifier $\exists$ is replaced by $\exists^\out$ (any
    variable is assumed to be quantified over outputs)
  \item $x\leq y$ is replaced by $x\leq_\out y$
  \item predicates $\{\phi\}(x,y)$ are replaced by $\{\phi'\}(\ori(x),\ori(y))$ where $\phi'$ is obtained from $\phi$
    by replacing $\preceq$ by $\leq_\inp$. 
\end{itemize}
By construction, a typed data word $w$ satisfies $\psi$ iff
$\td^{-1}(w)$ satisfies $\psi'$. 

\end{proof}

\section{Complexity of satisfiability}

\ThmComplex*

\begin{proof}
Let us first remark that the complexity of \ldata and \lorigin is equivalent since the translation between the two is linear.
Moreover, since the logic $\fo^2[\Gamma,\preceq, S_\preceq,\leq]$ is \textsc{ExpSpace}-complete~\cite{SZ12}, we get ExpSpace-hardness as \ldata strictly extends this logic.

Let us now prove the ExpSpace solvability of \lorigin.
As stated by Lemma~\ref{lem-ProfileAut}, satisfiability of an \lpp instance $C$ can be solved in ExpSpace with respect to the number of states in the query automata and in PSpace with respect to the size of $C$.
Now given a formula in Scott Normal Form, we obtain a \lpp instances $C$ whose size is linear in the size of the output alphabet and the size of the query automata.
However, the construction from a formula of \lorigin to a formula in Scott Normal Form 
trades each quantification for a unary predicate, whichis then incorporated into the extended alphabet of the SNF formula. Since these predicates are not mutually exclusive, this results in an exponential blow up of the output alphabet, and hence an exponential number of constraints, while using the same query automata.
Combining this and the complexity result from Lemma~\ref{lem-ProfileAut}, we get an ExpSpace complexity for the satisfiability of \lorigin, and hence \ldata.

\end{proof}

\section{Extensions of \elorigin}

\paragraph{Existential \lorigin}

\NmsoToElorigin*

\begin{proof}
Here the proof is exactly the same as the one of Theorem~\ref{thm:msoexpr} where the existentially quantified monadic predicates play the role of the parameters.
\end{proof}

\UnifElorigin*

\begin{proof}

The synthesis result of \lorigin can be extended to any \elorigin formula $\psi=\exists X_1 \ldots X_n \phi$: if we consider the monadic second-order variables as additional unary predicates, we obtain a formula $\phi$ over a signature extended with new unary predicates. Using Theorem~\ref{thm:unif}, we are able to obtain, for instance, a deterministic two-way transducer $T$ realising $\phi$ but over an extended alphabet. By projecting back to the original alphabets what we obtain is a non-deterministic two-way transducer $\overline T$ realizing a relation included in $\sem\varphi$ and with the same domain. We can then make $\overline T$ deterministic using the result of \cite{dS13} -- or even reversible using \cite{DFJL17} -- and thus obtain a synthesis of $\psi$.

\end{proof}

\UnivElorigin*

\begin{proof}
The main idea is to show that the satisfiability problem for the $\forall \lorigin$ logic, \ie formulas with a block of universal monadic quantifications followed by an \lorigin-formula, is undecidable. To this end we encode the same transduction as in the proof of Proposition~\ref{prop:undecfo2}.
We re-use the notations and definitions of the proof of Proposition~\ref{prop:undecfo2} and our goal is to define a sentence $\psi$ defining the same transduction as $\phi$ as:
$\psi \equiv \foroutput X\ \psi_{\text{well-formed}}(X)\wedge \bigwedge_{\ell=1}^2
     \phi_{\text{bij},\ell}\wedge \phi_{\text{ord-pres},\ell}\wedge \phi_{\text{lab-pres},\ell} $.
We have left to define $\psi_{\text{well-formed}}(X)$, whose role is to ensure that the output word belongs to the language $(+ _{1\leq i\leq n}1(u_i)2(v_i))^*$.
First we define a predicate  which states that $X$ is a contiguous set of positions:
$$\begin{array}{rl}
\text{cont}(X)\equiv \neg\big(\exoutput x\ x\in X  & \wedge
(\exoutput y\ x<_{\out} y \wedge y\notin X \\
 & \wedge (\exoutput x\ y<_{\out}x \wedge x\in X ))\big)
\end{array}$$
Similarly, for a word $w\in (A_1 + A_2)^*$ we define in the logic a new predicate $w(X,x)$ which states that $w$ is a subword of the positions of $X$, starting at position $x$. The predicates are defined by induction. For a letter $\sigma$ and a word $w$: $\sigma(X,x)\equiv\sigma(x)\wedge x\in X$ and  $\sigma w(X,x)\equiv\sigma (X,x) \wedge \exoutput y\ x<_{\out}y \wedge w(X,y)$. Using the same technique, without considering labels, we can define predicates $|X|\bowtie i$ for any integer $i$ and $\bowtie \in \set{<,>,\leq,\geq, =}$. We also define $w(X)\equiv \exoutput x\ w(X,x)\wedge |X|=|w|$. Let $m=\max\set{|u_i|+|v_i||\ 1\leq i \leq n }+1$.

We define $F$ as the set of factors of words of the language $(+ _{1\leq i\leq n}1(u_i)2(v_i))^*$ of size $\leq m$, \ie the set of acceptable output factors of length less than $m$.
We also define $P$ as the set of words in $(+ _{1\leq i\leq n}1(u_i)2(v_i))A_1$ and $S$ as the set $A_2(+ _{1\leq i\leq n}1(u_i)2(v_i))$.
Let $\min(X)$ and $\max(X)$ denote that the minimum, and respectively the maximum, position of the word belongs to $X$.

$$\begin{array}{llll}

\psi_{\text{well-formed}}(X)&\equiv&
\text{cont}(X)\wedge |X|\leq m\\
& \rightarrow&  \bigvee_{w\in F} w(X)\\
&& \wedge \bigvee_{w\in P} \min(X)\wedge |X|=|w| \rightarrow w(X)\\
&& \wedge \bigvee_{w\in S} \max(X)\wedge |X|=|w| \rightarrow w(X)
\end{array}$$

Note that this formula does not consider the solutions of size one of the PCP instance, which is not a problem since if $u_i=v_i$ is a solution, then so is $u_iu_i=v_iv_i$.
Intuitively, $\foroutput X\ \psi_{\text{well-formed}}(X)$ ensures that all factors (up to some length) of the output, the prefix and the suffix are of the correct form, which guarantees that the output word is indeed in the language $(+ _{1\leq i\leq n}1(u_i)2(v_i))^*$.
\end{proof}

\paragraph{Single-origin predicates}

\NFTtoLTSO*

\begin{proof}
Let us consider a one-way non-deterministic transducer $T$. Our goal is to define a formula of $\lorigin^{\mathrm{so}}$ defining the same transduction. We can see $T$ as a deterministic automaton over the alphabet $\Sigma\uplus\Gamma$, with a state space $Q$, initial state $q_0\in Q$, set of final states $F\subseteq Q$. For a state $q\in Q$ and a letter $a\in \Sigma\uplus\Gamma$, we denote by $q{\cdot} a$ the state reached upon reading the letter $a$ from state $q$. In order to simplify the proof, and without loss of generality, we assume that the initial state $q_0$ has to read an input letter. We see $T$ as a transduction with the following semantics: let $a_1v_1\cdots a_nv_n$ be a word accepted by $T$ such that for $1\leq i \leq n$ we have $a_i\in \Sigma$ and $v_i\in \Gamma^*$.
We define $(u,(v,o))$ with $u=a_1\cdots a_n$, $v=v_1\cdots v_n$ and for $|v_1\cdots v_{j-1}|< i \leq |v_1\cdots v_{j}|$, $o(i)=j$. Then we have $(u,(v,o))\in \sem T_o$. Let $p,q\in Q$ and let $L_{p,q}$ denote the set of words of $\Gamma^*$ which go from $p$ to $q$.
We define the $\lorigin^{\mathrm{so}}$-formula $\phi_{\mathrm{pres}}\wedge \{ \exists_{p,q,r\in Q}\ X^p_{q,r} \phi \}$, with $\phi_{\mathrm{pres}}$ the same formula defined in Example~\ref{ex:lorigin}. The formula $\phi$ will be a conjunction of four formulas: $\phi_{var}\wedge \phi_{succ}\wedge\phi_{min}\wedge\phi_{max}$. The formula $\phi_{var}$ will encode that each variable $X^p_{q,r}$ contains the input positions which can go from $p$ to $q$ upon reading the letter and then produce a word in $L_{q,r}$. The formula $\phi_{succ}$ encodes that two successive input positions must belong to some $X^p_{q,r}$ and $X^r_{s,t}$ respectively. The $\phi_{min}$ formula states that the first input position starts in the initial state and $\phi_{max}$ that the last input position produces a word which goes to some accepting state.

$$
\begin{array}{lcl}
\phi_{var} & = & \forinput x\ X^p_{q,r}(x) \rightarrow \left(L_{q,r}(x) \bigwedge_{\sigma\in \Sigma} \sigma(x)\rightarrow p.\sigma =q \right)\\
\phi_{succ} & = &\forinput x,y \bigwedge_{p,q,r} \left( S_{\inp}(x,y) \wedge X^p_{q,r}(x) \rightarrow \bigvee_{s,t} X^r_{s,t}(x) \right) \\
\phi_{min} & = & \forinput x\ \min(x) \rightarrow \bigvee_{p,q} X^{q_0}_{p,q}(x) \\
\phi_{max} & = & \forinput x\ \max(x) \rightarrow \bigvee_{r_f\in F, p,q} X^{p}_{q,q_f}(x) \\

\end{array}
$$
\end{proof}

\UnifLTso*

\begin{proof}
Here we sketch how to synthesise a regular function from an $\lorigin^{\mathrm{so}}$ formula $\phi$.
The extension to $\elorigin^{\mathrm {so}}$ follows by the same procedure as Proposition~\ref{prop-UnifElorigin}.
 The idea is to view $\sem \phi_o$ as the composition of two transductions $\tau_2\circ \tau_1$, with $\tau_1$ being rational (\ie given by a one-way transducer), and $\tau_2$ being an \lorigin-transduction.
We first synthesise $f_2$, realising $\tau_2$, using Theorem~\ref{thm:unif}.
Then the restriction $\tau_1'$ of $\tau_1$ to $\Sigma^*\times dom(f_2)$ is also rational since $f_2$ has a regular domain.
Using~\cite{Elgot:Mezei:ibmjrd:1965}, we can then synthesise $f_1$, realising $\tau_1'$, and the composition $f_2\circ f_1$ realises the original formula $\phi$.
The transduction $\tau_1$ is just a rational transduction which after a letter $a\in \Sigma$ produces an arbitrary word of $a\Gamma^*$.
The idea is that now the single-origin predicates can talk directly about the input and in Fig.~\ref{fig:so} we give an example of what $\tau_2$ is supposed to do. We transform $\phi$ syntactically into an \lorigin-formula $\phi'$ over an enriched signature, with the input alphabet being now $\Sigma\uplus\Gamma$ (we actually use a distinct copy of $\Gamma$ but don't write it differently for simplicity). 
We need to define a binary input predicate which relates input positions labelled in $\Gamma$ to the previous input position labelled by $\Sigma$, \ie its ``origin'' with respect to $\tau_1$ (which we call its \emph{virtual origin}): $$\vo(x,y)= \Sigma(y) \wedge y <_{\inp} x \wedge (\forinput z\ y <_{\inp} z  \leqinput x \rightarrow \Gamma(z))$$
For any regular language $L$ over $\Gamma^*$ we denote by $\phi_L$ the MSO formula recognizing it.
The syntactic transformation only modifies the input predicates and is done in three steps: 

1) We guard all quantifications such that they only talk about positions labelled in $\Sigma$, $(\exinput x\ \psi(x))'=\exinput x\ \Sigma(x) \wedge \psi'(x)$, $(\forinput x\ \psi(x))'=\forinput x\ \Sigma(x) \rightarrow \psi'(x)$.

2) The binary predicates of the form $\{ P(x,y) \}$ are replaced by $\{\exinput z,t\ \vo(x,z) \wedge \vo(y,t) \wedge P(z,t) \}$.

3) Finally  all predicates $L(x)$ are replaced by $\phi_L^{\vo(\_,x)}$ where $\vo(\_,x)$ means that all quantifications of $\phi_L$ are restricted to positions with virtual origin $x$.

The final formula is $\phi'\wedge\phi_{\mathrm{well-formed}}$ where $\phi_{\mathrm{well-formed}}$ states that:

1) The input positions labelled by $\gamma\in \Gamma$ produce exactly one position labelled by $\gamma$

2) The input positions labelled in $\Sigma$ don't produce anything

3) Given two output positions $x,y$, if there exists an input position $z$ such that $\vo(\ori(x),z)\wedge \vo(\ori(y),z)$ then $x\leqoutput y \leftrightarrow \ori(x)\leqinput \ori(y)$

\begin{figure}

\begin{tikzpicture}[scale=.5,every node/.style={scale=0.8}]
\node (s1) at (0,2) {\color{red} $\sigma_1$};
\node (s2) at (2,2) {\color{red}$\sigma_2$};
\node (s3) at (4,2) {\color{red}$\sigma_3$};

\node (g1) at (0,0) {\color{blue}$\gamma_1$};
\node (g2) at (1,0) {\color{blue}$\gamma_2$};
\node (g3) at (2,0) {\color{blue}$\gamma_3$};
\node (g4) at (3,0) {\color{blue}$\gamma_4$};
\node (g5) at (4,0) {\color{blue}$\gamma_5$};
\node (g6) at (5,0) {\color{blue}$\gamma_6$};

\foreach \a/\b in {g2/s1,g4/s1,g1/s2,g3/s2,g5/s3,g6/s3}
\draw[->, >=stealth] (\a) -- (\b);

\draw (6.5,2) -- (6.5,0);

\end{tikzpicture}
~\begin{tikzpicture}[scale=.5,every node/.style={scale=0.8}]
\node (s1) at (0,2) {\color{red}$\sigma_1$};
\node (gi2) at (1,2) {\color{blue}$\gamma_2$};
\node (gi4) at (2,2) {\color{blue}$\gamma_4$};
\node (s2) at (3,2) {\color{red}$\sigma_2$};
\node (gi1) at (4,2) {\color{blue}$\gamma_1$};
\node (gi3) at (5,2) {\color{blue}$\gamma_3$};
\node (s3) at (6,2) {\color{red}$\sigma_3$};
\node (gi5) at (7,2) {\color{blue}$\gamma_5$};
\node (gi6) at (8,2) {\color{blue}$\gamma_6$};

\node (g1) at (0,0) {\color{blue}$\gamma_1$};
\node (g2) at (1,0) {\color{blue}$\gamma_2$};
\node (g3) at (2,0) {\color{blue}$\gamma_3$};
\node (g4) at (3,0) {\color{blue}$\gamma_4$};
\node (g5) at (4,0) {\color{blue}$\gamma_5$};
\node (g6) at (5,0) {\color{blue}$\gamma_6$};

\foreach \a/\b in {g2/gi2,g4/gi4,g1/gi1,g3/gi3,g5/gi5,g6/gi6}
\draw[->, >=stealth] (\a) -- (\b);

\node at (-1.0,0) {};

\end{tikzpicture}

\caption{An o-graph of $\tau$ and its translated version as an o-graph of $\tau_2$.}
\label{fig:so}
\end{figure}
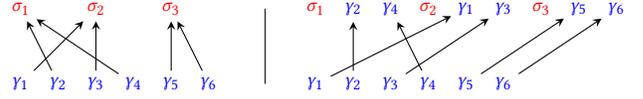
\end{proof}